\documentclass[aps,pra,superscriptaddress,onecolumn,10pt]{revtex4-2}
\usepackage{amsmath,mathtools,amsthm,amssymb}
\usepackage{newtxtext,newtxmath}
\usepackage{microtype}
\usepackage{aliascnt}
\usepackage{physics}
\usepackage{dsfont}
\usepackage{comment}
\usepackage{enumitem}
\usepackage[dvipsnames]{xcolor}

\definecolor{myrefcolor}{rgb}{0.067,0.5,0.5}
\definecolor{myurlcolor}{rgb}{0.1,0,0.9}
\usepackage[breaklinks,colorlinks=true,linkcolor=myrefcolor,citecolor=myrefcolor,urlcolor=myurlcolor]{hyperref}
\usepackage[capitalize]{cleveref}
\newtheorem{theorem}{Theorem}
\crefname{theorem}{Theorem}{Theorems}
\Crefname{theorem}{Theorem}{Theorems}
\newaliascnt{lemma}{theorem}
\newtheorem{lemma}[lemma]{Lemma}
\aliascntresetthe{lemma}
\newcommand{\Var}{\operatorname{Var}}
\crefname{lemma}{Lemma}{Lemmas}
\Crefname{lemma}{Lemma}{Lemmas}
\newtheorem{proposition}{Proposition}
\newtheorem{corollary}{Corollary}
\crefname{corollary}{Corollary}{Corollaries}
\Crefname{corollary}{Corollary}{Corollaries}
\crefname{proposition}{Proposition}{Propositions}
\Crefname{proposition}{Proposition}{Propositions}
\newcommand{\poly}{\operatorname{poly}}
\newcommand{\id}{\mathds{1}}
\newcommand{\idmap}{\mathrm{id}}
\DeclareMathOperator{\supp}{supp}
\newcommand{\be}{\begin{equation}\begin{aligned}\hspace{0pt}}
\newcommand{\ee}{\end{aligned}\end{equation}}

\allowdisplaybreaks
\allowdisplaybreaks
\begin{document}

\title{Universal quantum coding}
\author{Jacopo Rizzo}
\email{jacopo.rizzo@fu-berlin.de}
\affiliation{\fu}
\author{Filippo Girardi}
\email{filippo.girardi@sns.it}
\affiliation{\sns}
\author{Ludovico Lami}
\affiliation{\sns}
\author{Jens Eisert}
\affiliation{\fu}
\author{Lorenzo Leone}
\affiliation{\sa}
\affiliation{\infn}
\newcommand{\fu}{Dahlem Center for Complex Quantum Systems, Freie Universit\"at Berlin, 14195 Berlin, Germany}
\newcommand{\sns}{Scuola Normale Superiore, Piazza dei Cavalieri 7, 56126 Pisa, Italy}
\newcommand{\sa}{Dipartimento di Ingegneria Industriale, Università degli Studi di Salerno, Via Giovanni Paolo II, 132, 84084 Fisciano (SA), Italy}
\newcommand{\infn}{INFN, Sezione di Napoli, Gruppo Collegato di Salerno, Italy}

\begin{abstract}
Entanglement distillation from noisy bipartite states and reliable quantum communication over noisy channels are fundamental tasks in quantum information theory. Yet existing coding schemes typically require prior knowledge of the underlying state or channel, or assume independent or weakly correlated quantum information sources.
Here, using Schur--Weyl duality, we establish a universal quantum error-correction principle that overcomes these limitations. We show that the irreducible permutation spaces arising from the local Schur--Weyl decomposition of multiple copies of an unknown bipartite mixed state form maximally entangled quantum code spaces whose error structure is governed entirely by representation theory. Correcting these symmetry-resolved errors yields protocols for entanglement distillation and quantum communication that achieve coherent-information rates with exponentially vanishing error and optimal second-order corrections, even in the case of almost i.i.d. resources.
Our approach also provides sharp finite-copy guarantees that are worst-case optimal at every finite block length, including in the sample-efficient regime, where the number of copies scales only polynomially with the number of qubits.
Together, these results establish Schur--Weyl duality as a general mechanism for universal, robust, and optimal quantum coding, unifying entanglement distillation, quantum error correction, and quantum communication within a common representation-theoretic framework.
\end{abstract}

\maketitle
\let\savedaddcontentsline\addcontentsline
\renewcommand{\addcontentsline}[3]{}

\section{Introduction}
Entanglement distillation and quantum communication are two central primitives of quantum information theory. Indeed, they form two pillars on which much of its abstract theory rests. In the most standard setting, the former task aims at converting many copies of a noisy bipartite state into high-fidelity maximally entangled states by making use of suitable \emph{local operations and classical communication} (LOCC), while the latter concerns the reliable transmission of quantum information through a noisy quantum channel. Despite their central role, both problems are notoriously difficult to fully characterize in general: neither the distillable entanglement of an arbitrary mixed state nor the quantum capacity of an arbitrary quantum channel is known to admit a single-letter, explicit characterization.
Nevertheless, an essential entropic quantity, the coherent information, provides a fundamental operational benchmark in both settings. 
For an independent and identically distributed (i.i.d.)\ bipartite source $\rho_{AB}$, the hashing inequality gives an asymptotically achievable one-way LOCC entanglement-distillation rate \cite{DevetakWinter2005}
\begin{align}
E_{\rm D}^{\to}(\rho_{AB})
\ge
I(A\rangle B)_\rho
\coloneqq
H(B)_\rho-H(AB)_\rho,
\end{align}
where $H(X)_\rho \coloneqq - \Tr (\rho_{X} \log_2 \rho_X)$.
For a quantum channel $\mathcal N:\mathcal L(A')\to\mathcal L(B)$, with $\mathcal L(X)$ denoting linear operators on a quantum system $X$, the optimal rate of quantum information transmission, known as quantum capacity, is given instead by the regularized coherent information, optimized over the channel input states
\begin{align}
Q(\mathcal N)
=
\lim_{m\to\infty}
\frac1m
\max_{\xi}
I_c\!\left(\xi,\mathcal N^{\otimes m}\right),
\qquad
I_c(\xi,\mathcal N)
\coloneqq
H(B)_\omega-H(RB)_\omega ,
\label{eq:intro-quantum-capacity}
\end{align}
where $\xi\in\mathcal D(A')$, $\phi_\xi=\ketbra*{\phi_\xi}{\phi_\xi}_{RA'}$ is a purification with $R\simeq A'$, and $\omega_{RB}\coloneqq(\idmap_R\otimes\mathcal N)(\phi_\xi)\in\mathcal D(RB)$ \cite{Lloyd1997,Devetak2005Channel}. 

However, standard achievability proofs for these information-theoretic rates are typically built around the particular input state to be distilled or the respective target channel, through state-dependent typical subspaces, or channel-dependent random coding constructions \cite{DevetakWinter2005,Klesse_2007,HaydenHorodeckiWinterYard2008}. Universal coding under uncertainty is known from the compound-source and compound-channel results of Refs.~\cite{BjelakovicBocheJanssen2013,BjelakovicBocheNoetzel2009,BocheJanssen2014}, but there the code is constructed relative to a prescribed compound class, and its construction relies on a nontrivial code optimization over that class.
This raises a sharper question: can coherent-information rates be attained by a fully \emph{universal} code whose recovery is fixed by a simple, state-independent criterion, with the unknown source, channel, or uncertainty set entering only through the condition that certifies the target rate?

In this work, we settle this question by exhibiting a direct and universal coding mechanism based on Schur--Weyl duality.
For i.i.d. pure bipartite states, the local Schur--Weyl decomposition directly isolates irreducible permutation registers that are exactly maximally entangled, yielding universal, state-independent entanglement distillation at the optimal von-Neumann entropy rate \cite{PhysRevA.75.062338}. Here we show that, for arbitrary i.i.d. mixed states, these same registers persist as canonical quantum code spaces, while the unknown purifying environment induces a universal family of representation-theoretic errors with only state-dependent weights. 
The selected errors define a state-independent reference channel whose Choi state is flat on the corresponding error space. Designing a suitable random encoding for the fixed reference channel then gives, via polar recovery, the corresponding encoding and decoding scheme for the actual physical channel.

This fundamental mechanism provides a simple criterion for the achievable rate and, in turn, a Schur--Weyl information-spectrum quantity controlling the finite-copy entanglement distillation yield (\cref{thm:universal-distillation-main}).
We connect this Schur--Weyl information-spectrum quantity directly to Petz--R\'enyi divergences \cite{petz1986quasientropies}, showing that the universally achievable asymptotic distillation rate is governed by the coherent information at first order and by the coherent-information variance at second order, matching the known second-order achievability bound and, for maximally correlated states in the positive-variance regime, the converse bound \cite{Fang_2019}. This shows that a single universal LOCC family achieves the hashing rate with exponentially vanishing error and its optimal second-order correction, without any tomography, state estimation, or state-dependent decoding.
Moreover, we extend our guarantees to strongly correlated, almost i.i.d. quantum sources satisfying a mild condition on their average two-copy marginals, showing that coherent-information rates remain universally achievable well beyond the independent and identically distributed setting.
A more refined analysis also yields explicit sample-complexity bounds in two complementary regimes (\cref{sec:main-samples}). To approach the coherent-information rate, we show that the number of input copies is controlled by the ranks of the relevant local and global spectra rather than by the full Hilbert-space dimensions. 
Furthermore, under a Petz--R\'enyi conditional-entropy promise, the sample complexity depends only on the global rank, with no explicit dependence on local ranks or Hilbert space dimensions. 

In particular, for $n$-qubit systems with a supplied polynomial global-rank ceiling,
polynomially many copies suffice to achieve rates up to
$\min\{H_{1+s}^{\downarrow}(A|E)_\psi,\log_2 n\}$ for every fixed
$0<s<1$, where $H_{1+s}^{\downarrow}(A|E)_\psi\coloneqq-\frac1s\log_2\Tr\!\left[\rho_{AB}^{1-s}(\id_A\otimes\rho_B^s)\right]$ denotes the Petz--R\'enyi conditional entropy, with $\psi_{ABE}$ denoting any purification of $\rho_{AB}$. Similar guarantees also hold in
terms of the optimized sandwiched conditional R\'enyi entropies
$\widetilde H_\alpha^\uparrow(A|E)_\psi$, whose
$\alpha\to\infty$ limit recovers the conditional min-entropy. 
For pure bipartite inputs, this fully generalizes the results of Ref.~\cite{Leone_2025}, which identified the optimally achievable sample-efficient pure-state rate with the min-entropy of the reduced density matrix. 
In forthcoming work, some of us further show that the more general conditional R\'enyi entropy criterion established here is likewise sample-efficiently optimal for mixed states \cite{UtsumiRizzoTakagiLeone}.
Thus, non-vanishing, and even logarithmically growing, distillation rates remain achievable universally with polynomial sample complexity despite exponentially large local Hilbert spaces. Remarkably, our protocol also simultaneously preserves enough local information to reconstruct both local input marginals exactly.
Finally, we show that a similar construction applies to universal asymptotic rates for quantum communication (\cref{thm:main-universal-communication}). For any fixed input state and a positive channel coherent-information guarantee, we provide an encoder--decoder pair that works simultaneously for every quantum channel satisfying that guarantee with high probability over a shared random seed. Applying this threshold theorem to suitable entangled block inputs then recovers, for arbitrary compound channel families, the established optimal regularized coherent-information capacity for compound families of quantum channels \cite{BjelakovicBocheNoetzel2009}. 
These results reveal a broad and profound role for Schur--Weyl duality in universal quantum information processing. Across a wide class of tasks, the Schur decomposition does more than merely expose and simplify the underlying symmetry: it isolates representation sectors that directly support optimal entanglement distillation and transmission, while organizing the accompanying noise into a structure amenable to universal quantum error correction.

\section{Main results}
To set up the mathematical framework, we consider $k$ copies of an unknown bipartite (generally mixed) state $\rho_{AB}$. Throughout, $d_X\coloneqq\dim X$ and $r_X\coloneqq\rank\rho_X$ denote the dimension and spectral rank, respectively. Our starting point is the local Schur--Weyl decomposition \cite{weyl1946classical}
\begin{align}
(\mathbb C^d)^{\otimes k}
&\underset{S_k\times\mathcal U(d)}{\cong}
\bigoplus_{\lambda\vdash k,\;\ell(\lambda)\le d}
\mathcal P_\lambda\otimes\mathcal Q_\lambda^d ,
\label{eq:main-schur-weyl}
\end{align}
where $\lambda$ ranges over partitions of $k$ with at most $d$ nonzero entries. The unitary implementing this change of basis is known as Schur transform, denoted by $U_{\mathrm{Schur}}^X$ \cite{harrow2005applicationscoherentclassicalcommunication}.
The space $\mathcal P_\lambda$ carries the irreducible representation $p_\lambda(\cdot)$ of the symmetric group $S_k$, while $\mathcal Q_\lambda^d$ carries the corresponding irreducible representation $q_\lambda$ of the $d$-dimensional unitary group $\mathcal U(d)$, naturally extended to operators. We write
$d_\lambda\coloneqq \dim\mathcal P_\lambda$ and $D_\lambda^{(d)}\coloneqq \dim\mathcal Q_\lambda^d$. In the original Hilbert space, these actions correspond to the permutation operators $R_\pi$ and the tensor product unitaries $U^{\otimes k}$ (see \cref{app:schur}).
Throughout, Schur blocks are ordered as $\mathcal P\otimes\mathcal Q$. Since the reduced density matrix $\rho_A^{\otimes k}$ commutes with the permutation action, Schur's lemma gives
\begin{align}
\rho_A^{\otimes k}
\cong
\bigoplus_\lambda
I_{\mathcal P_\lambda}
\otimes
q_\lambda(\rho_A),
\label{eq:main-schur-marginal}
\end{align}
where $q_\lambda(\rho_A)$ is a positive semidefinite operator on $\mathcal Q_\lambda^{d_A}$. If $r_A=\rank\rho_A$, then $q_\lambda(\rho_A)=0$ whenever $\ell(\lambda)>r_A$, so only Young diagrams with at most $r_A$ rows have nonzero weight. Conditional on $\lambda$, the permutation register is therefore maximally mixed, while all state dependence is confined to the sector weight and the $\mathcal Q_\lambda^{d_A}$ register. The probability of sampling $\lambda$ reads $p_\lambda = d_\lambda\,\Tr q_\lambda(\rho_A)$.
Moreover, for large $k$, the normalized sampled Young diagram $\lambda/k$ approximates the spectrum of $\rho_A$ \cite{PhysRevA.64.052311}. Crucially, the Schur transform is a unitary change of basis and admits an efficient quantum circuit implementation \cite{BaconChuangHarrow2006,Krovi2019efficienthigh}.

\subsection{A tripartite Schur--Weyl decomposition}

We now introduce a tripartite Schur decomposition underlying our coding construction. We first start with the simpler bipartite case. For $k$ copies of a pure bipartite state vector $\ket{\psi}_{AR}$ shared between Alice and a reference system $R$, simultaneous permutation invariance forces the same Young diagram $\lambda$ on both sides of the bipartition and produces a maximally entangled state on the corresponding irreducible permutation registers \cite{PhysRevA.75.062338}, so that in the local Schur basis 
\be
(U_{\mathrm{Schur}}^A\otimes U_{\mathrm{Schur}}^R)\ket\psi_{AR}^{\otimes k}
=
\sum_{\substack{\lambda\vdash k\\ \ell(\lambda)\le r_A}}
\sqrt{p_\lambda}\,
\ket\lambda_A\ket\lambda_R
\otimes\ket*{\Phi_{\mathcal P_\lambda}}_{A:R}
\otimes\ket{\phi_\lambda}_{A:R},
\label{eq:bipartite-schur-purification}
\ee
where $\ket{\phi_\lambda}_{A:R}\in\mathcal Q_\lambda^A\otimes\mathcal Q_\lambda^R$ is normalized, and $\ket*{\Phi_{\mathcal P_\lambda}}_{A:R}\coloneqq\frac{1}{\sqrt{d_\lambda}}\sum_{i=1}^{d_\lambda}\ket{i}_A\ket{i}_R$,
with $\{\ket{i}\}_{i=1}^{d_\lambda}$ the Young-Yamanouchi basis of the irreducible $S_k$ representation $\mathcal P_\lambda$, naturally indexed by standard Young tableaux. In this basis, the irreducible action of the symmetric group is real orthogonal. A standard Young tableau of shape $\lambda\vdash k$ is a filling of the boxes of the Young diagram $\lambda$ with the integers $1,\ldots,k$, each used once, such that the entries increase along every row and every column.
Thus every Schur sector contains an exact, state-independent EPR state of Schmidt rank $d_\lambda$. This is the representation-theoretic mechanism underlying universal pure-state entanglement distillation \cite{PhysRevA.75.062338}.
For mixed states, we choose instead a purification $\ket{\psi}_{ABE}$ and first apply the Schur transform across the cut $A:BE$. Splitting the joint $BE$ registers into local $B$ and $E$ registers induces dual Clebsch--Gordan branchings on the permutation and unitary sectors, 
\begin{align}
\mathcal P_\mu\otimes\mathcal P_\nu
\underset{S_k}{\cong}
\bigoplus_\lambda
\mathcal P_\lambda\otimes\mathbb C^{g_{\lambda\mu\nu}}, \qquad
\mathcal Q_\lambda^{d_Bd_E}
\underset{\mathcal U(d_B)\times\mathcal U(d_E)}{\cong}
\bigoplus_{\mu,\nu}
\mathcal Q_\mu^{d_B}\otimes
\mathcal Q_\nu^{d_E}\otimes
\mathbb C^{g_{\lambda\mu\nu}} ,
\end{align}
governed by the same Kronecker coefficients. Here, the multiplicity $g_{\lambda\mu\nu}$ counts the independent ways in which the local sectors $(\mu,\nu)$ couple to the global $BE$ sector $\lambda$. We denote the corresponding branching unitaries by
\be
V_{\mu,\nu}:
\mathcal P_\mu^B\otimes\mathcal P_\nu^E
\longrightarrow
\bigoplus_\lambda
\mathcal P_\lambda^{BE}\otimes\mathbb C^{g_{\lambda\mu\nu}},
\qquad
W_\lambda:
\mathcal Q_\lambda^{d_Bd_E}
\longrightarrow
\bigoplus_{\mu,\nu}
\mathcal Q_\mu^B\otimes\mathcal Q_\nu^E\otimes
\mathbb C^{g_{\lambda\mu\nu}} .
\label{eq:CG-branching}
\ee
We then prove the following key generalization of \cref{eq:bipartite-schur-purification}.

\begin{theorem}[Schur--Weyl purification (informal version of \cref{thm:schurpurif})]\label{thm:tripartitemain}
Let $p_\lambda$ and $\ket{\phi_\lambda}_{A:BE}$ be the quantities in \eqref{eq:bipartite-schur-purification} for the cut $A:BE$. Then, for an arbitrary purification $\ket\psi_{ABE}$ of $\rho_{AB}$, the following decomposition holds
\be
&(U_{\mathrm{Schur}}^A\otimes U_{\mathrm{Schur}}^B\otimes U_{\mathrm{Schur}}^E)
\ket\psi_{ABE}^{\otimes k}\\
&=
\sum_{\lambda,\mu,\nu}
\sqrt{p_\lambda}\,
\ket\lambda_A\ket\mu_B\ket\nu_E
\otimes
\sum_{\alpha=1}^{g_{\lambda\mu\nu}}
\bigl(\id\otimes(V_{\mu,\nu}^{\lambda,\alpha})^\dagger\bigr)
\ket*{\Phi_{\mathcal P_\lambda}}_{A:BE}
\otimes
\bigl(\id\otimes W_{\mu,\nu}^{\lambda,\alpha}\bigr)
\ket{\phi_\lambda}_{A:BE}.
\label{eq:tripartite-schur-main}
\ee
The sum runs over $\lambda,\mu,\nu\vdash k$ satisfying the respective rank constraints. The maps $V_{\mu,\nu}^{\lambda,\alpha}$ and $W_{\mu,\nu}^{\lambda,\alpha}$ are the projected branch intertwiners on the $\mathcal P$ and $\mathcal Q$ registers from \cref{eq:CG-branching}. 
\end{theorem}
More precisely, the projected maps are
\begin{align}
V_{\mu,\nu}^{\lambda,\alpha}
\coloneqq(\id_{\mathcal P_\lambda}\otimes\bra{\alpha})\Pi_\lambda^{\mu,\nu}V_{\mu,\nu},\qquad
W_{\mu,\nu}^{\lambda,\alpha}
\coloneqq(\id_{\mathcal Q_\mu^B\otimes\mathcal Q_\nu^E}\otimes\bra{\alpha})\Pi_{\mu,\nu}^{\lambda}W_\lambda,
\end{align}
where $\Pi_\lambda^{\mu,\nu}$ and $\Pi_{\mu,\nu}^{\lambda}$ select the corresponding direct-sum components, and $\ket{\alpha}$ selects the multiplicity coordinate.
In the theorem, these maps resolve the combined $BE$ registers into separate $B$ and $E$ registers:
\begin{align}
(V_{\mu,\nu}^{\lambda,\alpha})^\dagger
:\mathcal P_\lambda^{BE}\longrightarrow\mathcal P_\mu^B\otimes\mathcal P_\nu^E,\qquad
W_{\mu,\nu}^{\lambda,\alpha}
:\mathcal Q_\lambda^{d_Bd_E}\longrightarrow\mathcal Q_\mu^B\otimes\mathcal Q_\nu^E.
\end{align}
In the following, subscripts on $O$ indicate the parameters on which its implicit constant may depend. For fixed local dimension $d$, the unitary registers have polynomial dimension, $\dim\mathcal Q_\lambda^d=k^{O(d^2)}$, and therefore carry at most $O_d(\log_2 k)$ entanglement. On the other hand, permutation registers can have exponentially large dimension in $k$ and retain the leading, extensive contribution. 
The decomposition in \cref{thm:tripartitemain} already contains the crucial insights that lead to most of our results. Indeed, the irreducible permutation branches have a particularly simple entropic structure. For each normalized pure branch $\ket*{B_{\mu,\nu}^{\lambda,\alpha}}\coloneqq(\id\otimes(V_{\mu,\nu}^{\lambda,\alpha})^\dagger)\ket*{\Phi_{\mathcal P_\lambda}}_{A:BE}$, permutation symmetry forces all one-party marginals to be maximally mixed. Consequently, \cref{lem:orthperm} gives the following conditional min-entropy content
\be
H_{\min}(\mathcal P_\lambda^A|\mathcal P_\nu^E)_{B_{\mu,\nu}^{\lambda,\alpha}}
=
\log_2\frac{d_\mu}{d_\nu}
\approx k\,I(A\rangle B)_\rho,
\ee
where $H_{\min}(X|Y)_\omega\coloneqq-\log_2\inf\{\Tr T_Y:T_Y\ge0,\ \omega_{XY}\le\id_X\otimes T_Y\}$. The last approximation holds on typical Schur sectors, up to sublinear corrections in $k$ at fixed local dimensions.
Importantly, the conditional min-entropy quantifies the single-shot LOCC-accessible entanglement \cite{BuscemiDatta2010}, so the previous equation directly links the entanglement accessible through Schur--Weyl structure to the asymptotic hashing yield.
This branchwise identity motivates the coding construction: the macroscopic entanglement resides in the permutation sectors in a state-agnostic form. The remainder of the paper develops this observation into a rigorous framework, casting entanglement distillation and quantum communication as the correction, on the maximally entangled permutation registers, of errors induced by the Clebsch--Gordan maps in \cref{eq:tripartite-schur-main}. In the following section, we make this intuition precise.

\subsection{Universal entanglement distillation from Schur--Weyl coding}\label{sec:maindist}

We start from the decomposition in \cref{eq:tripartite-schur-main}, and assume that Alice and Bob have measured their local Schur labels $(\lambda,\mu)$, and Alice has communicated classically her outcome $\lambda$ to Bob. If the unitary registers $\mathcal{Q}$ are traced out, it is relatively easy to show that the resulting conditional state is the output of a channel acting on one half of a fixed maximally entangled state across $A:BE$:
\be
\rho_{\mathcal P_\lambda^A\mathcal P_\mu^B}
=
\sum_\nu
\Tr_{\mathcal P_\nu^E}\!\left[
(\id_{\mathcal P_\lambda^A}\otimes V_{\mu,\nu}^\dag)
\left( \left(\Phi_{\mathcal P_\lambda}\right)_{A:BE}
\otimes S^{(\nu)}
\right)
(\id_{\mathcal P_\lambda^A}\otimes V_{\mu,\nu})
\right],
\label{eq:physical-schur-channel-main}
\ee
Here $S^{(\nu)}\ge0$ are Gram matrices in multiplicity space satisfying $\sum_\nu\Tr S^{(\nu)}=1$; they depend on the input state and observed labels $(\lambda,\mu)$, which we omit from the notation. This picture reveals a fundamental insight. To extract entanglement, Alice can encode quantum information using a suitable code isometry $C:\mathbb C^K\to\mathcal P_\lambda$ via the transpose trick:
\begin{align}
(C^T\otimes\id)
\bigl(\id\otimes(V_{\mu,\nu}^{\lambda,\alpha})^\dagger\bigr)
\ket*{\Phi_{\mathcal P_\lambda}}_{A:BE}
=
\sqrt{\frac{K}{d_\lambda}}\,
\bigl(\id_K\otimes(V_{\mu,\nu}^{\lambda,\alpha})^\dagger C\bigr)
\ket{\Phi_K}_{A:BE}.
\label{eq:schur-code-transpose}
\end{align}
The distillation problem then reduces to recovering entanglement from the Choi state of the induced channel
\(\mathcal N_{\lambda\to\mu}\) (\cref{eq:physical-schur-channel-main}).
Furthermore, the channel picture suggests replacing the state-dependent channel in \cref{eq:physical-schur-channel-main} with a worst-case, state-independent channel for the design of the final encoder and decoder. 
To formalize this intuition, for a suitable nonempty set $\mathcal T$ of admissible recoverable sectors indexed by $\nu$, we define the reference channel $\widehat{\mathcal N}_{\lambda\to\mu}^{\mathcal T}:\mathcal L(\mathcal P_\lambda^{BE})\to\mathcal L(\mathcal P_\mu^B)$ by
\be
\widehat{\mathcal N}_{\lambda\to\mu}^{\mathcal T}(X)
\coloneqq
\frac{1}{N_{\lambda\mu}(\mathcal T)}
\sum_{\nu\in\mathcal T} d_\nu
\Tr_{\mathcal P_\nu^E}\!\left[
V_{\mu,\nu}^\dagger
\left(
X\otimes \id_{\mathbb C^{g_{\lambda\mu\nu}}}
\right)
V_{\mu,\nu}
\right],
\label{eq:reference-schur-channel}
\ee
where $N_{\lambda\mu}(\mathcal T)\coloneqq\sum_{\nu\in\mathcal T}g_{\lambda\mu\nu}d_\nu$ and $X\in\mathcal L(\mathcal P_\lambda^{BE})$.  
The reference channel in \cref{eq:reference-schur-channel} assigns weights according to the effective dimensions of the selected error sectors. We then show that this state-independent choice yields a recovery map whose performance on the physical channel, combined with the Schur-spectrum bounds below, establishes achievability of the coherent-information rate.
Universal entanglement distillation can therefore be viewed, sector by sector, as a quantum error-correction, entanglement recovery problem for these canonical Schur channels.
This leads to a simple error-space packing picture. Bob's permutation register has dimension \(d_\mu\), whereas an environment sector \(\nu\) generates a representation-theoretic error space of size $d_\nu$. Since Bob does not know \(\nu\), our decoder includes all admissible error sectors of dimension at most $d_\nu$, together with their multiplicities. For a supplied rank ceiling \(r\ge\rank\rho_{AB}\), we define the corresponding cumulative error volume function
\begin{align}
\mathcal V_{\lambda\mu}^{(r)}(t)
&\coloneqq
\sum_{\substack{\eta\vdash k,\ \ell(\eta)\le r\\ d_\eta\le t}}
g_{\lambda\mu\eta}d_\eta,
\label{eq:sw-volume-information-main}
\end{align}
Then, intuitively, Bob is able to recover a $K$-dimensional logical subspace from Alice's logical encoding if the packing condition
\be
K \ll \frac{d_\mu}{\mathcal V_{\lambda\mu}^{(r)}(d_\nu)}
\label{eq:pack}
\ee
holds. We formalize this intuition by defining $J_k^{(r)} \coloneqq \log_2 \left(d_\mu/
\mathcal V_{\lambda\mu}^{(r)}(d_\nu)\right)$ and the corresponding Schur--Weyl information spectrum quantity as
\be
H_{\mathrm{SW},k}^{\varepsilon,(r)}(A|E)_\psi
\coloneqq
\max\left\{
x\in\mathbb R:
\Pr_{\psi^{\otimes k}}
\!\left[J_k^{(r)}<x\right]
\le\varepsilon
\right\},
\qquad
0\le\varepsilon<1.
\label{eq:sw-spectrum-entropymain}
\ee
This is the finite-copy \(\varepsilon\)-quantile of the packing ratio above: the amount of quantum information (measured in qubits) that fits in Bob's permutation space after accounting for the environment errors that must be corrected. The following theorem then rigorously bounds the recovery error.

\begin{theorem}[Universal entanglement distillation
(informal version of \cref{thm:universal-distillation})]
\label{thm:universal-distillation-main}
Fix $k$, a global rank ceiling $r$, and $0<\eta<\varepsilon<1$.
There is an explicit universal family of randomized one-way LOCC
protocols depending only on $(d_A,d_B,k,r,\eta)$ such that, for every
state $\rho_{AB}$ with $\rank\rho_{AB}\le r$, the family distills
a maximally entangled state of Schmidt rank
\begin{align}
\left\lfloor
\eta\,\exp\!\left[
H_{\mathrm{SW},k}^{\varepsilon-\eta,(r)}(A|E)_\psi
\right]
\right\rfloor
\label{eq:entropy-achievable-main}
\end{align}
with infidelity at most $\varepsilon$ and high success probability. Here, $\psi_{ABE}$ is any purification of $\rho_{AB}$; the bound
is independent of this choice. The exponential is understood as the
inverse of the logarithm used to define the Schur-Weyl information-spectrum.
\end{theorem}

Operationally, in our protocol, Alice and Bob first apply their local Schur transforms and measure their sector labels $(\lambda,\mu)$. Using a shared random seed, Alice applies a unitary drawn from a unitary $2$-design on $\mathcal P_\lambda^A$ (an ensemble matching Haar second moments), measures the code blocks, and sends the outcome together with $\lambda$ to Bob. Bob then applies the recovery map constructed from the reference Schur channel in \cref{eq:reference-schur-channel}. The packing criterion in \cref{eq:pack} bounds the infidelity averaged over the seed. For each fixed input, Markov's inequality guarantees infidelity at most $\sqrt{\varepsilon}$ with probability at least $1-\sqrt{\varepsilon}$; hence vanishing average infidelity implies vanishing error with probability tending to one over the seed.
The protocol is insensitive to local changes of basis: collective rotations $U_A^{\otimes k}$ and $U_B^{\otimes k}$ act only on the unitary registers, leaving the Schur labels and permutation registers unchanged. Moreover, retaining each party's Schur label and unitary register preserves enough information to reconstruct $\rho_A^{\otimes k}$ and $\rho_B^{\otimes k}$ exactly by state-independent local channels: each party replaces the permutation register by the maximally mixed state in the recorded sector and applies the inverse Schur transform (\cref{prop:perfect-local-recoverability}). This recovery holds after averaging over measurement outcomes and concerns the local marginals, not the original bipartite correlations.
The mechanism behind our protocol can also be understood as an approximate Knill--Laflamme \cite{BarnumKnill2002} quantum error correction condition for the reference channel. Writing its Kraus operators as $F_{\nu,\alpha,e}/\sqrt{N_{\lambda\mu}(\mathcal T)}$, with
\begin{align}
F_{\nu,\alpha,e}
\coloneqq
\sqrt{d_\nu}\,
(\id_{\mathcal P_\mu^B}\otimes\bra e_{\mathcal P_\nu^E})
(V_{\mu,\nu}^{\lambda,\alpha})^\dagger,
\qquad \nu\in\mathcal T,
\end{align}
with $\ket{e}_{\mathcal P_\nu^E}$ an orthonormal basis choice of $\mathcal P_\nu^E$, the entanglement recovery property of Alice's code $\mathcal{C}$ can be expressed as
\begin{align}
C^\dagger F_{\nu,\alpha,e}^\dagger F_{\eta,\beta,f}C
\approx
\delta_{\nu\eta}\delta_{\alpha\beta}\delta_{ef}\,\id_K.
\label{eq:schur-knill-laflamme}
\end{align}
The approximation is quantified by the averaged Gram-defect bound in \cref{lem:normalized-polar}.
Operationally, these conditions ensure that the environment carries essentially no information about the encoded state.

\subsection{Finite-shot analysis and asymptotic achievability}
We now connect the Schur--Weyl information spectrum appearing in \cref{thm:universal-distillation-main} to other known entropic quantities, and derive explicit asymptotically achievable universal entanglement distillation rates.
In particular, we show that the Petz--R\'enyi conditional entropy $H_{1+s}^{\downarrow}(A|E)_\psi$ introduced above bounds the lower tail of the Schur packing statistic and thereby controls finite-copy distillation rates. The downward arrow specifies evaluation at the actual conditioning marginal, without optimization; its expression in terms of $\rho_{AB}$ follows from purification duality \cite{TomamichelBertaHayashi2014}, as detailed in \eqref{eq:petz-conditional-duality}. More precisely, we show:

\begin{theorem}[Petz--R\'enyi spectrum bound and error exponent]
\label{thm:petz-spectrum}
Let $\psi_{ABE}$ purify $\rho_{AB}$ and fix a rank ceiling
$r\ge\rank\rho_{AB}$. For every $k\ge1$, $0<\varepsilon<1$,
and $0<s<1$, the Schur--Weyl information spectrum satisfies
\begin{align}
H_{\mathrm{SW},k}^{\varepsilon,(r)}(A|E)_\psi
\ge {}
kH_{1+s}^{\downarrow}(A|E)_\psi
-O_{r_A,r_B,r}\!\bigl(\log_2(k+1)\bigr)
-\frac{1+s}{s}\log_2\frac1\varepsilon.
\label{eq:main-petz-spectrum}
\end{align}
Consequently, for every $t>0$ and every nonnegative rate
\begin{align}
R<
\sup_{0<s<1}
\left\{
H_{1+s}^{\downarrow}(A|E)_\psi
-\frac{1+2s}{s}\,t
\right\},
\label{eq:main-petz-error-exponent}
\end{align}
the universal protocol with $K_k=2^{\lfloor kR\rfloor}$ achieves averaged infidelity at most $2^{-kt}$ for all sufficiently large $k$.
\end{theorem}

The bound separates three contributions: an extensive entropy term, a logarithmic representation-theoretic overhead, and a penalty for the desired error. Combined with \cref{thm:universal-distillation-main}, it gives a finite-copy distillation guarantee directly in terms of a standard conditional entropy. The order $s$ controls the tradeoff: taking $s$ closer to zero improves the entropy term but increases the error penalty.
This tradeoff also explains the transition to the asymptotic regime. At fixed ranks, the overhead per copy vanishes, while
\be
\lim_{s\downarrow0}H_{1+s}^{\downarrow}(A|E)_\psi
=
H(A|E)_\psi
=
I(A\rangle B)_\rho.
\ee
Here $H(A|E)_\psi\coloneqq H(AE)_\psi-H(E)_\psi$. Taking $k\to\infty$ at fixed $s$, followed by $s\downarrow0$, therefore recovers the coherent-information rate. The same argument allows the error to vanish whenever $\log_2(1/\varepsilon_k)=o(k)$. At any fixed rate strictly below coherent information, a sufficiently small fixed $s$ also allows exponentially decreasing error.
A more refined, regularized version in \cref{lem:regularized-schur} also removes explicit local-rank dependence from the overhead. 
With a more careful analysis, we can describe fluctuations around the coherent-information rate through a classical information spectrum. 
The coupling in \eqref{eq:advanced-sw-classical-coupling} relates the Schur packing statistic to a sum $\widehat Z_k\coloneqq\sum_{t=1}^k Z_t$ of independent outcomes of the single-copy observable $\id_B\otimes\log_2\rho_E-\log_2\rho_B\otimes\id_E$ measured on $\rho_{BE}$. In particular, $J_k^{(r)}=\widehat Z_k+O(\log_2(k+1))$ with high probability at fixed ranks.
Each variable has mean $\mathbb E Z_t=I(A\rangle B)_\rho$ and variance $\operatorname{Var}(Z_t)=V(A\rangle B)_\rho$, where
\be
V(A\rangle B)_\rho
=
\Tr\rho_{AB}\!\left[
\log_2\rho_{AB}-\log_2(\id_A\otimes\rho_B)-I(A\rangle B)_\rho\id_{AB}
\right]^2
\ee
is the coherent information variance. 
The mean determines the leading achievable rate, while the variance controls the finite-copy fluctuations. For $V>0$, the Berry--Esseen theorem places the lower $\varepsilon$-quantile of the classical sum at $kI(A\rangle B)_\rho+\sqrt{kV}\,\Phi^{-1}(\varepsilon)+O(1)$ for fixed $\varepsilon\in(0,1)$, where $\Phi$ is the standard normal distribution function. Combining this Gaussian approximation with the logarithmic Schur corrections and the packing criterion yields the following second-order achievability bound.

\begin{theorem}[Coherent-information rate and optimal second-order achievability (informal version of \cref{prop:first-order,prop:second-order})]
\label{thm:main-asymptotic}
The universal Schur--Weyl entanglement distillation protocol of \cref{thm:universal-distillation-main}, at fixed averaged infidelity $\varepsilon\in(0,1)$, achieves the asymptotic rate 
\begin{align}
\frac{\log_2 K_k}{k}
\ge I(A\rangle B)_\rho+\sqrt{\frac{V(A\rangle B)_\rho}{k}}\,\Phi^{-1}(\varepsilon)
-O\!\left(\frac{\log_2 k}{k}\right),
\label{eq:main-second-order}
\end{align}
in particular, the coherent information rate is achievable with exponentially vanishing error. Furthermore, this bound is optimal among the class of universal entanglement distillation protocols. Optimality is achieved on the class of maximally correlated states \cite[Proposition~10]{Fang_2019}.
\end{theorem}
Here the remainder constant may depend on the fixed input state, rank ceiling, and target error, and we omitted the direct dependence for simplicity of notation.
The two statements describe complementary regimes: a fixed rate below $I(A\rangle B)_\rho$ permits exponentially vanishing error, whereas a fixed error permits rates approaching $I(A\rangle B)_\rho$ with a correction of order $k^{-1/2}$ per copy. This result shows that, surprisingly, universality introduces no additional asymptotic penalty at order $\sqrt{k}$ relative to the known second-order achievable bound \cite[Theorem~7]{Fang_2019}. For maximally correlated states, a matching converse makes this expansion optimal to second order \cite[Proposition~10]{Fang_2019}.

\subsection{Sample complexity bounds}\label{sec:main-samples}

The asymptotic rates given in the previous section do not by themselves determine how many copies are needed to approach a given yield with high probability. We now give two complementary guarantees: one approaches coherent information with a sample-complexity cost controlled by the local and global ranks, while the other removes local-rank dependence under a conditional R\'enyi-entropy promise. Throughout, the target averaged infidelity $\varepsilon$ is fixed, $0<\delta\le1$, and $\widetilde O$ suppresses logarithmic factors. The displayed bounds use the supplied ceiling $r=r_{AB}$; for a looser ceiling, one can replace $r_{AB}$ by $r$. 
To achieve rate at least $I(A\rangle B)_\rho-\delta$, the sufficient conditions in \cref{eq:sample-condition-concentration,eq:sample-condition-overhead} show that it suffices to take
\begin{align}
k
=
\widetilde O\!\left(
\frac{r_A r_B r_{AB}}{\delta^2}
\right)
\label{eq:main-coherent-sample-complexity}
\end{align}
input copies.
Importantly, the cost is controlled by spectral ranks rather than Hilbert space dimensions, largely bypassing the cost of full-state tomography.
A complementary bound removes this local-rank dependence by using an entropy promise to control the lower tail of the Schur spectrum. The underlying Petz--R\'enyi entropy estimate is developed in Appendix~\ref{app:min-entropy}. To connect with conditional min-entropy and the pure-state results of Ref.~\cite{Leone_2025}, where the min-entropy emerged as a sample-efficient achievable rate, we state its consequence for the optimized sandwiched conditional R\'enyi entropy $\widetilde H_\alpha^\uparrow(A|E)_\psi\coloneqq-\inf_{\sigma_E\in\mathcal D(E)}\widetilde D_\alpha(\rho_{AE}\Vert\id_A\otimes\sigma_E)$, where $\widetilde D_\alpha(\omega\Vert\sigma)\coloneqq(\alpha-1)^{-1}\log_2\Tr[(\sigma^{(1-\alpha)/(2\alpha)}\omega\sigma^{(1-\alpha)/(2\alpha)})^\alpha]$. For $\alpha>1$, negative powers act on the reference support, and the divergence is $+\infty$ unless $\supp\omega\subseteq\supp\sigma$. The upward arrow denotes optimization over the conditioning state; see also \eqref{eq:conditional-renyi}. For every fixed $\alpha>1$, \cref{prop:min-entropy-samples} shows that the promise $\widetilde H_\alpha^\uparrow(A|E)_\psi\ge R+\delta$, with $R\ge0$, guarantees rate $R$ using
\begin{align}
k
=
\widetilde O_\alpha\!\left(
\frac{r_{AB}^4\,2^{2R}}
{\delta^{\,5+2/(\alpha-1)}}
\right)
\label{eq:main-renyi-sample-complexity}
\end{align}
input copies.
Unlike \cref{eq:main-coherent-sample-complexity}, this bound has no explicit dependence on local ranks. The tradeoff is that the target rate is certified by a R\'enyi entropy, which can be smaller than coherent information, and the dependence on the rate gap $\delta$ is worse. Although $\widetilde H_\alpha^\uparrow(A|E)_\psi$ approaches coherent information as $\alpha\downarrow1$, the sample bound is not uniform in this limit.
Under the conditional min-entropy promise $H_{\min}(A|E)_\psi\ge R+\delta$, the corresponding bound becomes
\begin{align}
k
=
\widetilde O\!\left(
\frac{r_{AB}^4\,2^{2R}}
{\delta^5}
\right).
\label{eq:main-minentropy-sample-complexity}
\end{align}
For a bipartite pure state, one has $H_{\min}(A|E)_\psi=-\log_2\|\rho_A\|_\infty$, where $\|\rho_A\|_\infty$ denotes the largest eigenvalue of $\rho_A$. Thus, in the pure-state setting, our guarantee reduces to the entropy criterion underlying the benchmark of Ref.~\cite{Leone_2025}, based on the universal protocol of Ref.~\cite{PhysRevA.75.062338}. In particular, Ref.~\cite{Leone_2025} established the sample-efficient achievable rate $O(\min\{H_{\min}(A)_\psi,\log_2 n\})$.

Our result extends this picture to mixed states. For an $n$-qubit input with polynomially bounded global rank $r_{AB}$, and for fixed $\alpha>1$ and rate tolerance $\delta>0$, we achieve
\begin{align}
\frac{\log_2 K}{k}
\ge
\min\!\left\{
\widetilde H_\alpha^\uparrow(A|E)_\psi,\log_2 n
\right\}
-\delta,
\qquad
k=\poly(n),
\label{eq:main-polynomial-renyi}
\end{align}
thus recovering the conditional min-entropy rate when $\alpha\to\infty$.
Hence, these bounds identify a nonzero-rate sample-efficient regime even when the underlying Hilbert-space dimensions and the local ranks grow exponentially with the number of qubits. Moreover, in forthcoming work, some of us show that the conditional R\'enyi entropy criterion established here is in fact sample-efficiently optimal for mixed states \cite{UtsumiRizzoTakagiLeone}. These results show that Schur--Weyl duality provides a unified framework for sample-efficient entanglement distillation across both pure and mixed states.
We further note that in our setting Alice's encoder can be made computationally efficient: it consists of an efficient Schur transform, an approximate unitary $2$-design on the Specht register, and a measurement into fixed $K$-dimensional code blocks. Its gate complexity is polynomial in $(k,\log_2 d_A,\log_2(1/\varepsilon_{\mathrm{enc}}))$, where $\varepsilon_{\mathrm{enc}}>0$ bounds the additional term in the distillation error guarantee; see the implementation discussion in Appendix~\ref{app:universal}. Efficient implementation of Bob's polar recovery remains instead a separate question.

\subsection{Universal quantum communication}
The same Schur--Weyl error model introduced in \cref{sec:maindist} yields a universal scheme for transmitting quantum information through an unknown quantum channel. We consider a memoryless channel: $k$ independent uses act as the tensor product $\mathcal N^{\otimes k}$. Here, the aim is to transmit a quantum system while preserving its entanglement with an inaccessible reference.
Fix $\xi\in\mathcal D(A')$ as part of the code design, and let $\ket*{\phi_\xi}_{RA'}\in R\otimes A'$ be a purification, with $R\simeq A'$. For a channel $\mathcal N:\mathcal L(A')\to\mathcal L(B)$, the corresponding channel output state is
$\omega_{\mathcal N}\coloneqq(\idmap_R\otimes\mathcal N)(\phi_\xi)\in\mathcal D(RB)$,
whose coherent information is precisely $I_c(\xi,\mathcal N)$ from \eqref{eq:intro-quantum-capacity}. This induced bipartite state provides the link to the distillation setting developed above.
We now formulate the corresponding communication task. A $K$-dimensional code over $k$ channel uses has rate $\log_2 K/k$ qubits per channel use and consists of an encoding channel $\mathcal E:\mathcal L(\mathbb C^K)\to\mathcal L((A')^{\otimes k})$ and a decoding channel $\mathcal D:\mathcal L(B^{\otimes k})\to\mathcal L(\mathbb C^K)$. Together they induce the logical channel
$\Lambda\coloneqq\mathcal D\circ\mathcal N^{\otimes k}\circ\mathcal E$.
Its performance is measured by the entanglement fidelity
$F_e(\pi_K,\Lambda)\coloneqq\bra{\Phi_K}(\idmap_K\otimes\Lambda)(\Phi_K)\ket{\Phi_K}$,
where $\pi_K\coloneqq\id_K/K$; the transmission error is quantified as $1-F_e(\pi_K,\Lambda)$. In our universal construction, Alice and Bob use shared classical randomness to sample a pair $(\mathcal E_s,\mathcal D_s)$ before transmission, with the choice of seed independent of the unknown channel.
Universality means that the ensemble is channel-independent; our high-probability guarantee then holds simultaneously over the prescribed channel family.
The connection with the distillation construction is particularly direct at the encoder. In Schur coordinates,
\be
U_{\rm Schur}^{A'}\xi^{\otimes k}(U_{\rm Schur}^{A'})^\dagger
=
\sum_\lambda
p_\lambda\,
\ketbra{\lambda}{\lambda}
\otimes
\pi_{\mathcal P_\lambda}
\otimes
\xi_{\mathcal Q_\lambda}.
\ee
Thus, conditioned on $\lambda$, the unitary register is in the known state $\xi_{\mathcal Q_\lambda}$, while the permutation register is maximally mixed. To transmit a $K$-dimensional system, Alice simply replaces this maximally mixed register by a codeword: for an isometry $C:\mathbb C^K\to\mathcal P_\lambda$, Alice's encoding channel reads
\be
\mathcal E_{\lambda,C}(X)
\coloneqq
(U_{\rm Schur}^{A'})^\dagger
\left(
\ketbra{\lambda}{\lambda}
\otimes
CXC^\dagger
\otimes
\xi_{\mathcal Q_\lambda}
\right)
U_{\rm Schur}^{A'} .
\label{eq:main-channel-schur-encoder}
\ee
This is the quantum communication counterpart of the transpose encoding used in distillation: there Alice applies $C^T$ to one half of the maximally entangled permutation register, whereas here Alice inserts the logical system directly through $C$ before the physical channel.
After $k$ uses of the channel $\mathcal N$, Bob's Schur sectors carry exactly the same representation-theoretic error family as in the distillation problem, albeit with a different Gram matrix $S^{(\nu)}$. Crucially, permutation covariance implies that, for fixed input sector $\lambda$, the probability of Bob's output label $\mu$ is independent of both the transmitted state and the choice of $C$; see \eqref{eq:schur-output-label-flatness}. Alice's encoder can therefore be fixed before the unknown channel is used, while Bob applies in each $\mu$ sector the same reference-channel polar recovery developed in the distillation protocol above. A finite channel-independent ensemble of such encoder--decoder pairs satisfies the same Schur information-spectrum bound; a covering argument shows that, wit high probability over a shared seed, the sampled pair works simultaneously for the whole promised channel class. More formally, we prove:

\begin{theorem}[Universal quantum communication (informal version of \cref{thm:sw-transmission})]
\label{thm:main-universal-communication}
Fix an input state $\xi$, a rank bound $r$, and $\delta>0$. For every rate $R\ge0$, there exist channel-independent codes with quantum communication rate $R-o(1)$ whose decoding error vanishes exponentially, uniformly for all channels satisfying
\be
\rank\omega_{\mathcal N}\le r,
\qquad
I_c(\xi,\mathcal N)\ge R+\delta.
\label{eq:main-universal-communication}
\ee
The codes can be sampled using shared classical randomness, and a single seed can be fixed to obtain a deterministic universal code.
More generally, for any nonempty channel family $\mathfrak I$, every rate
\be
R<
Q_{\rm reg}(\mathfrak I)
\coloneqq
\sup_{m\ge1}\frac1m
\max_{\xi\in\mathcal D((A')^{\otimes m})}
\inf_{\mathcal N\in\mathfrak I}
I_c(\xi,\mathcal N^{\otimes m})
\label{eq:main-regularized-coherent-information}
\ee
is universally achievable, recovering the optimal compound-channel quantum capacity of Ref.\,\cite{BjelakovicBocheNoetzel2009}.
\end{theorem}

The state and channel constructions are therefore two sides of the same coding mechanism. In distillation, a maximally entangled permutation register is already present and Alice projects it onto a code; in communication, she prepares that code directly in the corresponding input permutation register. In either case, the unknown state or channel changes only the weights of the Schur errors, while the code geometry and recovery remain universal. 
We also note that the sample-complexity bounds of \cref{sec:main-samples} carry over directly to the quantum-channel setting, with the rank parameters identified with those of the induced state $\omega_{\mathcal N}=(\id\otimes\mathcal N)(\phi^\xi)$ and its marginals. Thus, our protocols also enable sample-efficient quantum communication at rates that may grow logarithmically with system size.

\subsection{Robustness on almost i.i.d.\ states}
\label{sec:non-iid}

Physically, the independent and identically distributed (i.i.d.)\ description of a system stems from the assumption that the source of states is memoryless, 
a somewhat idealized assumption:
any real device can at best be modeled by an approximate i.i.d.\ source~\cite{Renner-talk-Cambridge, Mazzola_2026}. Studying robustness of operational tasks in quantum Shannon theory and learning theory has already attracted considerable interest  \cite{Mazzola_2026_2,datta2026robustnessentanglementmanipulationiid, ye2026robustclassicalquantumcommunication, almost_iid, a_tale, datta2026entropyconcentrationuniversaltypicality, more_robust, Mazzola_2026_3, visnevskyi2026robustnessquantumspectrumestimation}. Three main approaches have been introduced in order to mathematically formalize the notion of \emph{almost i.i.d.}\ source. A first prototype was provided by Renner, who introduced a notion of `almost product' pure states~\cite[Theorem~4.3.2]{RennerPhD}, and later extended
to arbitrary states by Mazzola, Sutter, and Renner~\cite{Mazzola_2026} (which are often called \emph{MSR almost i.i.d.\ states}). Ref.~\cite{almost_iid} later introduced a more general approach based on the quantum Wasserstein distance of order 1 \cite{De_Palma_2021} as a metric for the deviations from the i.i.d.\ regime, which found several application in Shannon theoretic tasks \cite{a_tale, more_robust} and quantum spectrum estimation \cite{visnevskyi2026robustnessquantumspectrumestimation}. Finally,  Ref.~\cite{almost_iid} also introduced the weakest notion of almost i.i.d.\ source among the one mentioned above, only requiring that for every fixed $\ell\geq 1$, the $\ell$-body marginals are on average indistinguishable from the corresponding i.i.d.\ source. More precisely, they call a sequence $(\omega^{(k)})_{k\geq 1}$ of states $\omega^{(k)}\in\mathcal{D}(\mathcal{H}^{\otimes k})$ a \emph{weakly almost i.i.d.\ source along} $\rho\in\mathcal{D}(\mathcal{H})$ if
\begin{equation}\label{eq:marginals}
    \lim_{n\to\infty}\underset{|I|=\ell}{\mathbb{E}}\big\|(\omega^{(k)})_I-\rho^{\otimes I}\big\|_1=0\qquad \forall \ell\geq 1,
\end{equation}
where the expectation is taken over uniformly chosen subsets $I\subseteq[n]$ of size $\ell$ and $(\omega^{(k)})_I\coloneqq \Tr_{[k]\setminus I}\omega^{(k)}$. Datta showed that, for some relevant tasks such as quantum data compression and simple hypothesis testing, the condition in \eqref{eq:marginals} with $\ell=2$ is sufficient to ensure asymptotic robustness~\cite{datta2026entropyconcentrationuniversaltypicality}. 
In this work, we substantially strengthen the robustness framework
developed in previous works by introducing a weaker and more
physically motivated notion of approximate independence.
Inspired by \cite{datta2026entropyconcentrationuniversaltypicality},
we dispense with the requirement that individual two-body
marginals be close to a product state, and instead impose
a condition solely on their average. More precisely, we assume
that
\begin{equation}\label{eq:marginals}
    \left\|
    \underset{|I|=2}{\mathbb{E}}(\omega^{(k)})_I
    -\rho^{\otimes 2}
    \right\|_1\leq a_k,
\end{equation}
where the expectation is uniform over all pairs of subsystems. In the following, we will identify $\overline\omega_k^{(2)} \equiv \mathbb{E}_{|I|=2}(\omega^{(k)})_I$.
Crucially, we do not require $a_k\to 0$ as $k\to\infty$.
Our framework therefore accommodates both strongly
inhomogeneous correlations across different pairs and
persistent deviations from independence in the asymptotic
limit, as naturally expected in the presence of correlated
environmental noise. In particular, our robustness guarantees
depend only on the averaged two-body marginal, without
requiring the underlying many-body state to approach an i.i.d.\ structure.

Strikingly, with a small modification, the universal distillation and communication protocols described in the previous sections can also be used when the input pairs are correlated and only satisfy the mild condition in \cref{eq:marginals}. Consider an arbitrary state $\omega^{(k)}$ on $(AB)^{\otimes k}$, with $k\ge2$ and fixed local dimensions. Alice and Bob first apply the same uniformly random permutation to their positions. Averaging over this choice gives the one-way LOCC map
\begin{align}
\mathcal T_k(\omega^{(k)})
\coloneqq\frac1{k!}\sum_{\pi\in S_k}
(R_\pi^A\otimes R_\pi^B)\omega^{(k)}
(R_\pi^A\otimes R_\pi^B)^\dagger.
\label{eq:main-non-iid-twirl}
\end{align}
This step is needed for a general source: in this way, the averaged state is invariant under permutations. This will be fundamental to reduce to the previous i.i.d. case. After measuring the local Schur labels $(\lambda,\mu)$ and tracing the unitary registers, the conditional state therefore has the form
\begin{align}
\theta_{\lambda\mu}
\cong\bigoplus_\nu
 \frac{\id_{\mathcal P_\nu}}{d_\nu}\otimes S_{\lambda\mu}^{(\nu)}(\omega^{(k)}), \qquad
S_{\lambda\mu}^{(\nu)}&\ge0,
\qquad \sum_\nu\Tr S_{\lambda\mu}^{(\nu)}=1.
\label{eq:main-non-iid-identity-weights}
\end{align}
Thus the Schur channel's Choi state again has the $\id\otimes S$ structure: each irreducible carrier is maximally mixed, while the matrices $S$ contain the dependence on the source. The error operators are exactly those already used in \eqref{eq:weighted-schur-channel}.
Now one can immediately see why the same structure survives by expanding the permutation-invariant state in a basis of tensor powers:
\begin{align}
\mathcal T_k(\omega^{(k)})
=\sum_jc_j\rho_j^{\otimes k}, \qquad
\widehat S_{\lambda\mu}^{(\nu)}(\omega^{(k)})
=\sum_jc_j\widehat S_{\lambda\mu}^{(\nu)}(\rho_j^{\otimes k}).
\label{eq:main-non-iid-linear-weights}
\end{align}
Here the $\rho_j$ are one-pair states, and $\widehat S_{\lambda\mu}^{(\nu)}=p_{\lambda\mu}S_{\lambda\mu}^{(\nu)}$ includes the branch probability. These unnormalized matrices depend linearly on the input. The coefficients $c_j$ may be negative, but the final matrices $S$ are positive because they come from a physical state. The recovery of \cref{thm:universal-distillation} already allows all such positive matrices, so only the Schur-label probabilities need a new estimate.
For a fixed reference state $\rho_{AB}$, suppose that the averaged two-position marginal satisfies \cref{eq:marginals}.
Then, \cref{lem:two-marginal-schur} shows that this assumption controls the relevant Schur labels:
\begin{align}
\mathbb E\left\|\frac\mu k-\operatorname{spec}(\rho_B)\right\|_2^2
+\mathbb E\left\|\frac\nu k-\operatorname{spec}(\rho_{AB})\right\|_2^2
=O_{d_A,d_B}(a_k+k^{-1}),
\label{eq:main-non-iid-concentration}
\end{align}
where $\mu$ is Bob's label and $\nu$ is the joint $AB$ label. Together with the dimension and multiplicity bounds in \cref{lem:rank-dimensions,lem:rank-Q}, this gives the following consequence of the existing coding theorem.

\begin{corollary}[Universal entanglement distillation with almost i.i.d.\ sources]
\label{cor:non-iid-main}
Suppose \eqref{eq:marginals} holds with $a_k$ sufficiently small. Then, for every $0\le R<I(A\rangle B)_\rho$, the protocol of \cref{thm:universal-distillation}, preceded by the common random permutation and used with $r=d_Ad_B$, $L=k$ and $K_k=2^{\lfloor kR\rfloor}$, distills $\Phi_{K_k}$ with infidelity
\begin{align}
\varepsilon_k
=O_{d_A,d_B,\delta}\!\left(a_k+k^{-1}\right),
\label{eq:main-non-iid-error}
\end{align}
for every fixed gap $\delta>0$ such that $I(A\rangle B)_\rho\ge R+\delta$. Its operations depend only on $(d_A,d_B,k,K_k)$.
\end{corollary}

The proof follows the earlier rate argument. Concentration and entropy continuity make $H(\mu/k)-H(\nu/k)$ approach $I(A\rangle B)_\rho$. The dimension bounds and the choice $L=k$ cost only $O_{d_A,d_B}(\log k)$ bits, while the recovery error contributes $O(k^{-1})$.
The average condition \eqref{eq:marginals} permits
genuine correlations even for states maximally far from the i.i.d. case.
For example, for even $k$, consider the permutation-invariant state
\begin{align}
\frac{1}{k!}\sum_{\pi\in S_k}
R_\pi\bigl(\sigma_{AB}^{\otimes k/2}
\otimes\tau_{AB}^{\otimes k/2}\bigr)R_\pi^\dagger.
\end{align}
It satisfies the condition with $\rho=(\sigma+\tau)/2$ and
$a_k=\|\sigma-\tau\|_1^2/[4(k-1)]=O(k^{-1})$. 
For $\sigma$ and $\tau$ with orthogonal supports, the state is
nevertheless asymptotically far
from its i.i.d.\ counterpart in trace distance: $\|\omega^{(k)}-\rho^{\otimes k}\|_1
=2(1-\binom{k}{k/2}/2^k)\to2$.

\subsection{Almost i.i.d.\ quantum communication}
\label{sec:non-iid-communication}

The problem of robustness of classical communication was initially considered by Ref.~\cite{a_tale}, which introduced the notion of \emph{almost i.i.d.\ process} based on the Wasserstein analogue of the $1\to 1$ norm (see \cite[Def. 10 and 12]{a_tale}).  
That is: a sequence of channels $\big(\mathcal N^{(n)}_{A^n\to B^n}\big)_{n\geq 1}$ is an \emph{almost i.i.d.\ process} along $\mathcal N_{A\to B}$ if
\begin{align}
\lim_{k\to\infty}\frac{1}k \|\mathcal M_k-\mathcal M^{\otimes n}\|_\clubsuit =0\qquad \text{where}\qquad \|\mathcal M_k-\mathcal M^{\otimes k}\|_\clubsuit\coloneqq \sup_{\rho_{A^n}}\|\mathcal M_k(\rho)-\mathcal M^{\otimes k}(\rho)\|_{W_1}\in [0,k].
\end{align}
The case of quantum communication was later studied by Ref.~\cite{ye2026robustclassicalquantumcommunication}, which considered a deviations from the i.i.d.\ regime by treating the Choi state as Mazzola--Sutter--Renner almost i.i.d.\ sources with a sublinear number of defects in the number of uses of the channel. In this paper, we consider the following general framework for robustness of quantum communication, which is based on the deviation of the average two-body marginal with respect to the i.i.d.\ case when the input of the channel is an i.i.d.\ state on an extended system. As in the case of states discussed above, we allow for asymptotically non-vanishing discrepancy from the i.i.d.\ regime.
Here, we can retain the optimal, channel coherent information rate by imposing the two-body marginal discussed above to the channel acting on a reference state, on which the choice of the code ultimately depends.  
More precisely, we fix one known input state $\xi\in\mathcal D(A')$ as part of the code design and choose a purification $\phi_\xi\in\mathcal D(RA')$, with $R\simeq A'$. To describe the assumption on the channel, consider sending the input halves of $k$ copies of this purification through $\mathcal M_k$, while retaining the references. For a single-use reference channel $\mathcal N:A'\to B$, define
\begin{align}
\Omega_{\xi,\mathcal M_k}
&\coloneqq(\idmap_{R^{\otimes k}}\otimes\mathcal M_k)(\phi_\xi^{\otimes k}), \qquad
\omega_{\xi,\mathcal N}
\coloneqq(\idmap_R\otimes\mathcal N)(\phi_\xi),
\\
\overline\Omega_{\xi,\mathcal M_k}^{(2)}
&\coloneqq\frac1{k(k-1)}\sum_{i\ne j}
(\Omega_{\xi,\mathcal M_k})_{R_iB_iR_jB_j}.
\label{eq:main-channel-two-marginal}
\end{align}
Thus $\Omega$ describes the full reference--output experiment, whereas $\omega$ describes its single-use benchmark. The averaged two-position state retains two matched pairs $(R_i,B_i)$ and $(R_j,B_j)$, with distinct ordered positions chosen uniformly. Tracing out the unused reference registers leaves the corresponding physical inputs in $\xi$. The two-body marginal condition below therefore compares this two-position experiment, with background input $\xi$ elsewhere, to two independent copies of the benchmark. 
During communication, Alice uses the encoders of \cref{lem:universal-code-ensemble}, whose sector probabilities and unitary-register states are determined by $\xi$, and inserts the logical message into a permutation register as in \eqref{eq:main-channel-schur-encoder}.
Alice first applies a shared uniformly random permutation $R_\pi^{A'}$ to this encoded input, and Bob applies $(R_\pi^B)^\dagger$ to the output before the Schur decoder. Their average channel is
\begin{align}
\widetilde{\mathcal M}_k
\coloneqq\frac1{k!}\sum_{\pi\in S_k}
\operatorname{Ad}_{(R_\pi^B)^\dagger}\circ\mathcal M_k
\circ\operatorname{Ad}_{R_\pi^{A'}},
\label{eq:main-non-iid-channel-twirl}
\end{align}
with $\operatorname{Ad}_U(X)\coloneqq UXU^\dagger$.
This permutation step is needed to make an arbitrary block channel permutation covariant. Its probe state is $\mathcal T_k(\Omega_{\xi,\mathcal M_k})$, so the conditional Schur Choi states have the positive-weight form \eqref{eq:main-non-iid-identity-weights}. Covariance also makes Bob's output-label probabilities independent of the state inserted into the permutation register. These facts let us reuse the code-ensemble proof and the two-position concentration estimate from the entanglement-distillation setting; the details are in \cref{app:non-iid-communication}.

\begin{corollary}[Robust universal quantum communication]
\label{cor:non-iid-communication}
Fix $\xi$, $R_0\ge0$ and $\delta>0$. For each $k\ge2$, let $\mathfrak F_k$ be a nonempty family of channel pairs $(\mathcal M_k,\mathcal N)$ satisfying
\begin{align}
I_c(\xi,\mathcal N)&\ge R_0+\delta, \qquad
\sup_{(\mathcal M_k,\mathcal N)\in\mathfrak F_k}
\left\|\overline\Omega_{\xi,\mathcal M_k}^{(2)}
-\omega_{\xi,\mathcal N}^{\otimes2}\right\|_1 \leq a_k.
\label{eq:uniform-channel-two-marginal-promise}
\end{align}
The code ensemble of \cref{lem:universal-code-ensemble}, extended by the permutation twirl and used with row ceiling $r=d_{A'}d_B$, slack $L=k$ and $K_k=2^{\lfloor kR_0\rfloor}$, satisfies
\begin{align}
\sup_{(\mathcal M_k,\mathcal N)\in\mathfrak F_k}
\mathbb E_s\!\left[1-F_e\!\left(\pi_{K_k},
\mathcal D_{k,s}\circ\mathcal M_k\circ\mathcal E_{k,s}\right)\right]
=O_{d_{A'},d_B,\delta}\!\left(a_k+k^{-1}\right).
\label{eq:non-iid-channel-error-main}
\end{align}
Its isometric encoders and CPTP decoders depend only on $(\xi,d_{A'},d_B,k,K_k)$. The shared seed is sampled independently of the channel before transmission.
\end{corollary}

The supremum over channels remains outside the seed average: the guarantee holds for every channel selected before the shared seed. The simultaneous guarantee of \cref{thm:main-universal-communication} uses a covering argument for memoryless channels and does not follow here for unrestricted block-channel families.

\section{Relation to previous work}\label{sec:previous-results}
Bjelakovi\'c, Boche and Jan{\ss}en established universal state merging and distillation for compound source families \cite{BjelakovicBocheJanssen2013,BocheJanssen2014}. For finite families, their construction selects a common decoding isometry through fidelity optimization, while arbitrary families are handled by finite coverings \cite{BjelakovicBocheJanssen2013}. Our construction removes the source family itself from the protocol specification: the decoder is fixed by representation theory and works for every state satisfying the respective performance guarantee. Thus the distinction is not only that no covering is required; universality is built directly, rather than established relative to a prescribed family of possible sources.

The closest Schur-based precedents have narrower input scope. Matsumoto and Hayashi treat universal distillation of unknown pure states \cite{PhysRevA.75.062338}, while Czechlewski et al.\ analyze projection-and-hashing protocols for structured two-qubit mixtures and selected higher-dimensional extensions \cite{Czechlewski2012Distillation}. Previously, representation theory has also been used for universal compression, classical--quantum coding, and resolvability \cite{PhysRevA.66.022311,Hayashi2009Universal,MatsuuraHayashiHsieh2025}, and transpose recovery is well established \cite{BarnumKnill2002,BeigiDattaLeditzky2016}.

Our structural universality of the decoder has also concrete finite-copy consequences. The covering-number dependence on the Hilbert space dimension present in Ref.~\cite{BjelakovicBocheJanssen2013} is absent from our distillation bounds; instead, the sample complexity is controlled by the relevant spectral ranks, including global-rank-only guarantees under suitable entropy promises. Moreover, our universal protocol achieves the same coherent-information second-order term as the state-aware one-way hashing bound \cite{Fang_2019}, with no additional $\sqrt{k}$ penalty at fixed ranks. For maximally correlated states, the converse of Ref.~\cite[Proposition~10]{Fang_2019} then gives second-order optimality whenever the information variance is positive. To our knowledge, no previous universal distillation protocol achieves this second-order correction. The Schur--Weyl information spectrum identifies the representation-theoretic fluctuations responsible for this behavior and connects them to the methods underlying quantum second-order asymptotics \cite{TCR2009,tomamichel2013hierarchy,li2014second}.

A close fully quantum comparison is the any-shot decoupling theorem of Berta, Cheng and Yao \cite{BertaChengYao2026}. Their framework gives sharp decoupling bounds and state-merging and coding exponents in terms of conditional R\'enyi quantities, but the state-merging recovery is selected for the input through Uhlmann's theorem. Our decoder, by contrast, is fixed independently of the unknown state. The role of the R\'enyi quantities is therefore also different: here they arise from the statistics of Schur--Weyl packing and directly bound the performance of a common physical recovery, rather than entering only after optimization over an input-dependent recovery. This connects our construction to the broader operational role of Petz--R\'enyi conditional entropies in privacy amplification and quantum coding \cite{Hayashi2015,BertaChengYao2026}, and, through duality, to the Petz--R\'enyi divergences underlying quantum Hoeffding bounds \cite{Hayashi2007,AudenaertEtAl2008}. 
Recent work also studies universal distillation under broader operation classes. Lin, Li and Fang relate universal distillation under non-entangling and PPT-preserving operations to composite hypothesis testing, obtaining finite-blocklength bounds and error exponents for families that may include correlated sources \cite{LinLiFang2026}. Lami, Regula and Takagi characterize optimal universal distillation under non-entangling operations via a composite generalized quantum Stein's lemma \cite{LamiRegulaTakagi2026}. These results provide strong information-theoretic benchmarks, but need not yield explicit LOCC implementations; our protocol is instead one-way LOCC.

\section{Discussion and outlook}
Our results show that fundamental quantum information-processing tasks such as entanglement distillation and quantum communication can be performed at optimal worst-case finite-shot rates without any prior knowledge of the underlying state or channel. Perhaps even more strikingly, the same universal framework naturally gives rise to robust protocols under remarkably weak assumptions on the underlying noise, requiring only limited information about the average two-body marginals. Moreover, optimal entanglement distillation can be achieved while leaving the local marginals essentially undisturbed.
A particularly suggestive aspect of this picture is the connection it reveals between Schur--Weyl duality and quantum error correction. Symmetry has long served as a tool for simplifying problems in quantum information; here, it takes on a genuinely operational role, directly giving rise to an error-correction structure with optimal coding guarantees. This operational meaning is perhaps even stronger in the sample-efficient regime: the same Schur--Weyl framework continues to achieve optimal rates when only polynomially many copies are available.
Understanding how far this correspondence extends may uncover new families of universal quantum codes rooted directly in representation theory. More broadly, the same perspective may prove useful in quantum control, sensing, and general quantum resource conversion.
The guiding question is then not how accurately an unknown quantum system must be characterized, but how much characterization is actually necessary to process it optimally. The fact that universality and robustness can emerge from the same representation-theoretic framework suggests a broader and conceptual separation between the complexity of learning a quantum system and that of making optimal use of its quantum information.

\section*{Acknowledgements}
The authors acknowledge insightful discussions with Aram Harrow, Ryuji Takagi, Takeru Utsumi and Vladyslav Visnevskyi. J.R. thanks the Scuola Normale Superiore di Pisa for its hospitality during a research stay.
The Berlin team has been supported by the BMFTR (QR.N, QuSol, PraktiQOM), the Clusters of Excellence (ML4Q, MATH+), the QuantERA (SDPCode), the Munich Quantum Valley, Berlin Quantum, the Quantum Flagship (Millenion, Pasquans2), the DFG (CRC 183, SPP 2514), and the European Research Council (DebuQC).

\let\addcontentsline\savedaddcontentsline

\clearpage
\appendix
\onecolumngrid

\renewcommand{\tocname}{Supplemental Material}
\tableofcontents

\section{Preliminaries}\label{app:preliminaries}
All systems are finite-dimensional Hilbert spaces, denoted by their register labels. We write $\mathcal L(X)$ for linear operators on $X$ and $\mathcal D(X)$ for density operators. Kets denote state vectors, while $\psi_X\coloneqq\ketbra{\psi}{\psi}_X$ denotes the corresponding pure-state density operator. We abbreviate $X\otimes Y$ by $XY$, write $\rho_X\coloneqq\Tr_Y\rho_{XY}$ for marginals, and use $d_X\coloneqq\dim X$ and $r_X\coloneqq\rank\rho_X$. A purification of $\rho_X$ is a pure state on a larger system $RX$ with marginal $\rho_X$. Subsystem labels are omitted when clear.

A quantum channel $\mathcal N:\mathcal L(X)\to\mathcal L(Y)$ is completely positive and trace preserving (CPTP): it maps states on $X$ to states on $Y$, including when $X$ is correlated with a reference system. We distinguish the identity operator $\id_X$ from the identity channel $\idmap_X$. A Stinespring isometry $V:X\to YE$ realizes the channel as $\mathcal N(\rho)=\Tr_E(V\rho V^\dagger)$. Measurement branches are completely positive, trace-nonincreasing maps; their outputs are normalized by dividing by the outcome probability, when nonzero. Other positive operators used below need not have unit trace.
We denote the maximally mixed state by $\pi_X\coloneqq\id_X/d_X$ and the maximally entangled vector on two $d$-dimensional registers by $\ket{\Phi_d}\coloneqq d^{-1/2}\sum_j\ket j\ket j$, with density operator $\Phi_d$. The normalized Choi state of $\mathcal N$ is $(\idmap_R\otimes\mathcal N)(\Phi_d)$, where $R$ is a reference copy of the input. We use squared fidelity, $F(\rho,\sigma)\coloneqq\|\sqrt\rho\sqrt\sigma\|_1^2$. Entanglement fidelity $F_e(\rho,\Lambda)$ is the fidelity with which $\idmap_R\otimes\Lambda$ preserves a purification of $\rho$; in particular, $1-F_e(\pi_X,\Lambda)$ is the Bell-state infidelity.
We write $\log_2$ for base-two logarithms and $\ln$ for natural logarithms. We use the von Neumann entropy $H(X)_\rho\coloneqq-\Tr\rho_X\log_2\rho_X$, conditional entropy $H(X|Y)_\rho\coloneqq H(XY)_\rho-H(Y)_\rho$, and coherent information $I(X\rangle Y)_\rho\coloneqq-H(X|Y)_\rho$. For R\'enyi quantities, a tilde denotes the sandwiched variant, while $\downarrow$ and $\uparrow$ distinguish the physical marginal from an optimized conditioning state. We use the trace norm $\|X\|_1\coloneqq\Tr\sqrt{X^\dagger X}$, Hilbert--Schmidt norm $\|X\|_2\coloneqq\sqrt{\Tr X^\dagger X}$, operator norm $\|X\|_\infty$ (the largest singular value), and diamond norm $\|\cdot\|_\diamond$. Trace distance is $D_{\mathrm{tr}}(\rho,\sigma)\coloneqq\tfrac12\|\rho-\sigma\|_1$, and $\supp\omega$ is the span of its nonzero-eigenvalue eigenvectors. Transposes refer to fixed bases, and negative powers are taken on the indicated support, subject to the support conditions stated below.
We use $O$ and $\Omega$ for asymptotic upper and lower bounds up to constants, and $f=o(g)$ when $f/g\to0$. Subscripts specify allowed constant dependence; $\widetilde O$ suppresses logarithmic factors, and $\poly$ denotes a polynomial bound.

\subsection{Schur--Weyl duality}\label{app:schur}
In this section, we review the representation-theoretic background needed later, focusing on Schur--Weyl duality. This is a fundamental tool in quantum information theory, with applications including spectrum estimation, tomography, entanglement distillation, and quantum source compression \cite{Haah_2017,odonnell2015efficientquantumtomography,PhysRevA.64.052311,PhysRevA.75.062338,PhysRevA.66.022311}. We begin with a brief recap of the basic notions that we will use.
Let $G$ be a group. A unitary representation of $G$ on a finite-dimensional Hilbert space $\mathcal H$ is a map $g\mapsto R_g\in\mathcal U(\mathcal H)$ such that $R_{gh}=R_gR_h$ for all $g,h\in G$. A subspace $\mathcal K\subseteq\mathcal H$ is invariant if $R_g\mathcal K\subseteq\mathcal K$ for every $g\in G$, and the representation is irreducible if its only invariant subspaces are $\{0\}$ and $\mathcal H$. Every finite-dimensional unitary representation can, after a suitable change of basis, be written as a direct sum of irreducible representations, or irreps \cite{weyl1946classical}. Given two representations $(R_g,\mathcal H)$ and $(R'_g,\mathcal H')$, an intertwiner is a linear map $K:\mathcal H\to\mathcal H'$ such that $KR_g=R'_gK$ for all $g\in G$. The two representations are said to be isomorphic if there exists an invertible intertwiner between them; in that case we write $\mathcal{H} \underset{G}{\cong} \mathcal{H}'$. We will sometimes omit the group $G$ when that causes no ambiguity. We will repeatedly use Schur's lemma in the following form:

\begin{lemma}[Schur's lemma \cite{weyl1946classical}]\label{lemma:schur}
Let $\mathcal H,\mathcal H'$ carry complex irreducible representations $R_g,R'_g$ of $G$, and let $K:\mathcal H\to\mathcal H'$ be an intertwiner. If the representations are not isomorphic, then $K=0$. If instead $\mathcal H=\mathcal H'$ and $R_g=R'_g$ for every $g\in G$, then $K\propto\id_{\mathcal H}$.
\end{lemma}

We now specialize to the symmetric group of permutations of $k$ elements $S_k$ and the unitary group $\mathcal U(d)$, both acting on $(\mathbb C^d)^{\otimes k}$. Their respective group actions are
\be
R_\pi\bigl(\ket{\psi_1}\otimes\cdots\otimes\ket{\psi_k}\bigr)
 &\coloneqq \ket*{\psi_{\pi^{-1}(1)}}\otimes\cdots\otimes\ket*{\psi_{\pi^{-1}(k)}},\\
T_U\bigl(\ket{\psi_1}\otimes\cdots\otimes\ket{\psi_k}\bigr)
 &\coloneqq U^{\otimes k}\bigl(\ket{\psi_1}\otimes\cdots\otimes\ket{\psi_k}\bigr).
\ee
We will often write $U^{\otimes k}$ instead of $T_U$. These actions commute, $[R_\pi,T_U]=0$, and under both actions, Schur--Weyl duality gives the multiplicity-free decomposition into irreducible representations of $\mathcal U(d)\times S_k$
\be
(\mathbb C^d)^{\otimes k}
\underset{\mathcal U(d)\times S_k}{\cong}
\bigoplus_{\lambda\vdash k,\;\ell(\lambda)\le d}
\mathcal P_\lambda\otimes\mathcal Q_\lambda.
\label{eq:schur-weyl}
\ee
Here $\mathcal Q_\lambda$ carries an irrep of $\mathcal U(d)$ and $\mathcal P_\lambda$ the corresponding irrep of $S_k$~\cite{weyl1946classical}. The notation $\lambda\vdash k$ denotes a partition of $k$, i.e., a non-increasing sequence of positive integers $\lambda=(\lambda_1,\ldots,\lambda_{\ell(\lambda)})$ such that $\sum_i \lambda_i=k$, where $\ell(\lambda)$ denotes the number of parts of $\lambda$.
The associated Young diagram has $\lambda_i$ boxes in row $i$ and has exactly $\ell(\lambda)\leq d$ non-empty rows. 

A standard Young tableau (SYT) of shape $\lambda$ is a filling of the Young diagram with $1,\ldots,k$, each appearing exactly once, such that the entries are strictly increasing along rows and down columns. A semistandard Young tableau (SSYT) of shape $\lambda$ is a filling with entries in $[d]\coloneqq\{1,\ldots,d\}$, weakly increasing along rows and strictly increasing down columns. We take Young's orthogonal basis for $\mathcal P_\lambda$, indexed by $S\in\mathrm{SYT}(\lambda)$, in which the representation matrices of $S_k$ are real and orthogonal, and the Gelfand-Tsetlin basis for $\mathcal Q_\lambda$, indexed by $T\in\mathrm{SSYT}(\lambda)$. We write $d_\lambda\coloneqq\dim\mathcal P_\lambda$, which is independent of the local dimension $d$ \cite{weyl1946classical}, and set $D_\lambda^{(d)}\coloneqq \dim\mathcal Q_\lambda^d$. For a linear map $M:X\to Y$, its tensor power induces the Schur map $q_\lambda(M):\mathcal Q_\lambda^X\to\mathcal Q_\lambda^Y$ through $M^{\otimes k}\cong\bigoplus_\lambda\id_{\mathcal P_\lambda}\otimes q_\lambda(M)$, with absent sectors interpreted as zero. These maps preserve composition and adjoints; in particular, $q_\lambda(\rho_X)\ge0$ for $\rho_X\ge0$.

The Schur transform \cite{harrow2005applicationscoherentclassicalcommunication} is the unitary change of basis associated with \eqref{eq:schur-weyl}. For a system $X$ with local dimension $d_X$, it simultaneously block-diagonalizes the actions of $\mathcal U(d_X)$ and $S_k$, namely
\be
U_{\mathrm{Schur}}^X U_X^{\otimes k}R_\pi^X
(U_{\mathrm{Schur}}^X)^\dagger
 =R_{\mathrm{Schur}}^X(U_X,\pi),
\label{eq:schurint}
\ee
where
\be
R_{\mathrm{Schur}}^X(U_X,\pi)
 \coloneqq\sum_{\lambda\vdash k,\;\ell(\lambda)\le d_X}
 \ketbra{\lambda}{\lambda}_X\otimes p_\lambda^X(\pi)\otimes q_\lambda^X(U_X).
\ee
Here $q_\lambda^X(U_X)$ and $p_\lambda^X(\pi)$ act on $\mathcal Q_\lambda^X$ and $\mathcal P_\lambda^X$, respectively. The choice of Young's orthogonal basis then makes $p_\lambda^X(\pi)$ real orthogonal. We will sometimes omit the superscript $X$ when identifying equivalent copies of the same symmetric-group irrep. We will also freely interchange between the explicit $\lambda$-register notation, using $\ketbra{\lambda}{\lambda}$, and the corresponding direct-sum notation.
We denote by $\Pi_\lambda^X$ the orthogonal projector onto the $\lambda$-isotypic subspace of $X^{\otimes k}$; in Schur coordinates, its range is $\mathcal P_\lambda^X\otimes\mathcal Q_\lambda^X$.
The state $\rho_X^{\otimes k}$ is supported on $(\mathbb C^{r_X})^{\otimes k}$ up to an isometry. Hence only Young diagrams with $\ell(\lambda)\le r_X$ can occur with nonzero weight, even though the full Schur decomposition allows $\ell(\lambda)\le d_X$. 
Omitting zero-probability sectors, we write
\be
U_{\mathrm{Schur}}^X\,\rho_X^{\otimes k}\,(U_{\mathrm{Schur}}^X)^\dagger
=
\sum_{\substack{\lambda\vdash k\\ \ell(\lambda)\le r_X}}
p_\lambda^X\,
\ketbra{\lambda}{\lambda}^X
\otimes \frac{\id_{\mathcal P_\lambda^X}}{d_\lambda}
\otimes \rho_{\mathcal Q_\lambda}^X,
\ee
where
\be
p_\lambda^X\coloneqq\Tr\left[ \rho_X^{\otimes k} \Pi_\lambda^X \right]
=
d_\lambda\,\Tr\!\left[q_\lambda^X(\rho_X)\right],
\qquad
\rho_{\mathcal Q_\lambda}^X
\coloneqq
\frac{q_\lambda^X(\rho_X)}
{\Tr[q_\lambda^X(\rho_X)]}.
\ee
The identity on $\mathcal P_\lambda^X$ follows from Schur's lemma, since $\rho_X^{\otimes k}$ commutes with the $S_k$ action.

When we restrict an irrep of $G$ to a subgroup $H\subseteq G$, it may decompose into several irreps of $H$. We first consider the diagonal permutations on $B,E$, and then the tensor-product unitaries acting on $BE$.

\begin{lemma}[Clebsch--Gordan decomposition of the symmetric-group irreps]\label{lem:clebsch}
Fix two irreps $\mathcal P_\mu^B,\mathcal P_\nu^E$ with $\mu,\nu\vdash k$. Their tensor product is an irrep of $S_k\times S_k$ under the respective tensor product actions. Restricting to the diagonal subgroup $S_k^{\mathrm{diag}}\coloneqq\{(\pi,\pi):\pi\in S_k\}$ gives
\be
\mathcal P_\mu^B\otimes\mathcal P_\nu^E
\underset{S_k^{\mathrm{diag}}}{\cong}
\bigoplus_{\lambda\vdash k}
\mathcal P_\lambda^{BE}\otimes\mathbb C^{g_{\lambda\mu\nu}}.
\label{eq:symmetric-CG}
\ee
The multiplicities $g_{\lambda\mu\nu}$ are the Kronecker coefficients. Equivalently, there is a unitary Clebsch--Gordan transform
\be
V_{\mu,\nu}:\mathcal P_\mu^B\otimes\mathcal P_\nu^E
\longrightarrow\bigoplus_{\lambda\vdash k}
\mathcal P_\lambda^{BE}\otimes\mathbb C^{g_{\lambda\mu\nu}}
\label{eq:Vmunu-def}
\ee
satisfying, for every $\pi\in S_k$, the intertwining relation
\be
V_{\mu,\nu}(p_\mu(\pi)\otimes p_\nu(\pi))V_{\mu,\nu}^\dagger
=\bigoplus_{\lambda\vdash k} p_\lambda(\pi)\otimes\id_{\mathbb C^{g_{\lambda\mu\nu}}}.
\label{eq:Vmunu-intertwining}
\ee
\end{lemma}

\begin{lemma}[Corresponding branching rule for the global unitary-group irreps]\label{lem:unit}
Let $\lambda\vdash k$ label a sector appearing in the Schur--Weyl decomposition of $(\mathbb C^{d_B}\otimes\mathbb C^{d_E})^{\otimes k}$; in particular, this implies $\ell(\lambda)\le d_Bd_E$. Then, under the tensor-product homomorphism
\[
\mathcal U(d_B)\times\mathcal U(d_E)\ni(U_B,U_E)
\longmapsto U_B\otimes U_E,
\]
the global irrep $\mathcal Q_\lambda^{BE}$ decomposes as
\be
\mathcal Q_\lambda^{BE}
\underset{\mathcal U(d_B)\times\mathcal U(d_E)}{\cong}
\bigoplus_{\substack{\mu,\nu\vdash k\\
\ell(\mu)\le d_B,\;\ell(\nu)\le d_E}}
\mathcal Q_\mu^B\otimes\mathcal Q_\nu^E
\otimes\mathbb C^{g_{\lambda\mu\nu}}.
\label{eq:unitary-branching}
\ee
Hence there exists a unitary $W_\lambda$ such that
\be
W_\lambda q_\lambda^{BE}(U_B\otimes U_E)W_\lambda^\dagger
=
\bigoplus_{\substack{\mu,\nu\vdash k\\
\ell(\mu)\le d_B,\;\ell(\nu)\le d_E}}
q_\mu^B(U_B)\otimes q_\nu^E(U_E)
\otimes\id_{\mathbb C^{g_{\lambda\mu\nu}}}.
\label{eq:Wlambda-intertwining}
\ee
Moreover, for any $\mu,\nu\vdash k$ with $\ell(\mu)\le d_B$ and $\ell(\nu)\le d_E$,
\be
g_{\lambda\mu\nu}>0
\quad\Longrightarrow\quad
\ell(\lambda)\le d_Bd_E.
\label{eq:kronecker-length-bound}
\ee

\begin{proof}
Using
\[
(\mathbb C^{d_B}\otimes\mathbb C^{d_E})^{\otimes k}
\cong
(\mathbb C^{d_B})^{\otimes k}\otimes(\mathbb C^{d_E})^{\otimes k},
\]
regard the space as a representation of
$\mathcal U(d_B)\times\mathcal U(d_E)\times S_k^{\mathrm{diag}}$.
Applying Schur--Weyl duality separately to $B$ and $E$, and then
\cref{lem:clebsch}, gives
\be
\bigoplus_{\substack{\mu,\nu,\lambda\vdash k\\
\ell(\mu)\le d_B,\;\ell(\nu)\le d_E}}
\mathcal P_\lambda^{BE}\otimes
\mathcal Q_\mu^B\otimes\mathcal Q_\nu^E
\otimes\mathbb C^{g_{\lambda\mu\nu}}.
\label{eq:local-global-decomp}
\ee
Applying Schur--Weyl duality instead directly to the global $BE$ system gives
\be
\bigoplus_{\substack{\lambda\vdash k\\
\ell(\lambda)\le d_Bd_E}}
\mathcal P_\lambda^{BE}\otimes\mathcal Q_\lambda^{BE}.
\label{eq:global-schur-be}
\ee
These are two decompositions of the same
$\mathcal U(d_B)\times\mathcal U(d_E)\times S_k^{\mathrm{diag}}$
representation. Since the $\mathcal P_\lambda^{BE}$ are pairwise inequivalent irreducible $S_k$ representations, their multiplicity spaces must agree. For every $\lambda$ occurring in \eqref{eq:global-schur-be}, comparison of the $\mathcal P_\lambda^{BE}$ multiplicity spaces gives \eqref{eq:unitary-branching}, and choosing a unitary implementing this equivalence gives \eqref{eq:Wlambda-intertwining}.
Finally, if $\ell(\lambda)>d_Bd_E$, the sector $\mathcal P_\lambda^{BE}$ is absent from the global decomposition \eqref{eq:global-schur-be}. Its multiplicity space in \eqref{eq:local-global-decomp} must therefore vanish, which implies $g_{\lambda\mu\nu}=0$ for every admissible $\mu,\nu$. This proves \eqref{eq:kronecker-length-bound}.
\end{proof}
\end{lemma}

For completeness, we now give a self-contained proof of the pure-state Schur decomposition of Ref.~\cite{PhysRevA.75.062338} and derive its form when Bob uses the conjugate Schur basis. The former will be heavily used throughout the rest of the manuscript.

\begin{lemma}[Schur--Weyl decomposition of a bipartite pure state \cite{PhysRevA.75.062338}]\label{lem:pureschurweyl}
Let $\ket\psi_{AB}\in A\otimes B$ be a unit vector, with marginals $\rho_A=\Tr_B\psi_{AB}$ and $\rho_B=\Tr_A\psi_{AB}$, and Schmidt rank
$r_A=r_B$.
Then
\be
(U_{\mathrm{Schur}}^A\otimes U_{\mathrm{Schur}}^B)\ket\psi_{AB}^{\otimes k}
=
\sum_{\substack{\lambda\vdash k\\ \ell(\lambda)\le r_A}}
\sqrt{p_\lambda}\,
\ket\lambda_A\ket\lambda_B
\otimes\ket*{\Phi_{\mathcal P_\lambda}}_{A:B}
\otimes\ket{\phi_\lambda}_{A:B},
\label{eq:fund}
\ee
where $\ket*{\Phi_{\mathcal P_\lambda}}
=d_\lambda^{-1/2}\sum_{S\in\mathrm{SYT}(\lambda)}\ket S\ket S$, and
$\ket*{\phi_\lambda}\in\mathcal Q_\lambda^A\otimes\mathcal Q_\lambda^B$
is normalized.
More explicitly, let
$V:\mathbb C^{d_A}\to\mathbb C^{d_B}$ be a partial isometry such that
$V^\dagger V=\Pi_{\operatorname{supp}\rho_A}$ and
$\ket\psi=(\sqrt{\rho_A}\otimes\overline V)\ket*{\Gamma_{d_A}}$, where
$\ket*{\Gamma_{d_A}}\coloneqq\sum_{i=1}^{d_A}\ket i\ket i$ is unnormalized.
If Bob uses the conjugate Schur transform $\overline U_{\mathrm{Schur}}^B$ instead,
the same decomposition holds with
\be
\ket{\phi_\lambda}
=
\sqrt{\frac{D_\lambda^{(d_A)}}{\Tr q_\lambda(\rho_A)}}
\bigl(
q_\lambda(\sqrt{\rho_A})
\otimes
\overline{q_\lambda(V)}
\bigr)
\ket*{\Phi_{\mathcal Q_\lambda^A}},
\label{eq:pure-conjugate}
\ee
where $q_\lambda(V):\mathcal Q_\lambda^A\to\mathcal Q_\lambda^B$ is the Schur functor induced by $V$, and
$\ket*{\Phi_{\mathcal Q_\lambda^A}}
=(D_\lambda^{(d_A)})^{-1/2}\sum_{T\in\mathrm{SSYT}_{d_A}(\lambda)}\ket T\ket T$.

\begin{proof}
Expanding the tensor power in the local Schur bases gives
\be
(U_{\mathrm{Schur}}^A\otimes U_{\mathrm{Schur}}^B)\ket\psi^{\otimes k}_{AB}
=
\sum_{\lambda,\lambda'}
\ket\lambda_A\ket{\lambda'}_B
\otimes\ket*{\Omega_{\lambda,\lambda'}},
\ee
where
$\ket*{\Omega_{\lambda,\lambda'}}
\in
(\mathcal P_\lambda^A\otimes\mathcal P_{\lambda'}^B)
\otimes
(\mathcal Q_\lambda^A\otimes\mathcal Q_{\lambda'}^B)$ are generally unnormalized.
Since $\ket\psi^{\otimes k}$ is invariant under
$R_\pi^A\otimes R_\pi^B$, for each choice of unitary-sector basis vectors the corresponding vector in
$\mathcal P_\lambda^A\otimes\mathcal P_{\lambda'}^B$
is invariant under
$p_\lambda(\pi)\otimes p_{\lambda'}(\pi)$.
Writing such a vector as
$\sum_{S,S'}M_{SS'}\ket S\ket{S'}$, invariance is equivalent to
\be
p_\lambda(\pi)M p_{\lambda'}(\pi)^T=M,
\ee
and hence, since the Young representation matrices are real orthogonal,
\be
p_\lambda(\pi)M=M p_{\lambda'}(\pi).
\ee
Thus $M$ is an intertwiner from $\mathcal P_{\lambda'}$ to
$\mathcal P_\lambda$. By Schur's lemma it vanishes for
$\lambda\neq\lambda'$ and is proportional to the identity for
$\lambda=\lambda'$. The corresponding normalized invariant vector is
$\ket*{\Phi_{\mathcal P_\lambda}}$, since
\be
(p_\lambda(\pi)\otimes p_\lambda(\pi))
\ket*{\Phi_{\mathcal P_\lambda}}
=
\ket*{\Phi_{\mathcal P_\lambda}}.
\ee
Therefore $\ket*{\Omega_{\lambda,\lambda'}}=0$ unless
$\lambda=\lambda'$, and for the surviving sectors
\be
\ket*{\Omega_{\lambda,\lambda}}
=
\ket*{\Phi_{\mathcal P_\lambda}}
\otimes\ket{\widetilde\phi_\lambda}.
\ee
Their squared norms are
\be
\|\widetilde\phi_\lambda\|^2
=
\bra{\psi}^{\otimes k}
(\Pi_\lambda^A\otimes\id_{B^{\otimes k}})
\ket{\psi}^{\otimes k}
=
\Tr(\Pi_\lambda^A\rho_A^{\otimes k})
=
p_\lambda.
\ee
Since $\rho_A$ has rank $r_A$, one has $p_\lambda=0$ whenever
$\ell(\lambda)>r_A$, while $p_\lambda>0$ for every
$\ell(\lambda)\le r_A$. Normalizing the surviving vectors proves
\eqref{eq:fund}.
For the conjugate-basis form, introduce a second copy $A'$ of $A$ and write
$\ket\psi=(\sqrt{\rho_A}\otimes\overline V)\ket*{\Gamma_{d_A}}$.
Using $(U\otimes\overline U)\ket\Gamma=\ket\Gamma$ and repeating the same argument that we applied to the symmetric group irreps gives
\be
(U_{\mathrm{Schur}}^A\otimes\overline U_{\mathrm{Schur}}^{A'})
\ket*{\Gamma_{d_A}}^{\otimes k}
=
\sum_{\substack{\lambda\vdash k\\ \ell(\lambda)\le d_A}}
\sqrt{d_\lambda D_\lambda^{(d_A)}}\,
\ket\lambda\ket\lambda
\otimes
\ket{\Phi_{\mathcal P_\lambda}}
\otimes
\ket*{\Phi_{\mathcal Q_\lambda^{A}}}.
\ee
Moreover,
\begin{eqnarray}
U_{\mathrm{Schur}}^A
(\sqrt{\rho_A})^{\otimes k}
(U_{\mathrm{Schur}}^{A})^\dagger
&=&
\bigoplus_\lambda
\id_{\mathcal P_\lambda}
\otimes q_\lambda(\sqrt{\rho_A}),\\
\overline U_{\mathrm{Schur}}^B
\overline V^{\otimes k}
(U_{\mathrm{Schur}}^{A'})^T
&=&
\bigoplus_\lambda
\id_{\mathcal P_\lambda}
\otimes\overline{q_\lambda(V)}.
\end{eqnarray}
Hence the unnormalized unitary-sector vector is
\be
\ket{\widetilde\phi_\lambda}
=
\sqrt{d_\lambda D_\lambda^{(d_A)}}\,
\bigl(
q_\lambda(\sqrt{\rho_A})
\otimes
\overline{q_\lambda(V)}
\bigr)
\ket*{\Phi_{\mathcal Q_\lambda^A}}.
\ee
We have $q_\lambda(V)^\dagger q_\lambda(V)
=q_\lambda(V^\dagger V)$.
Using $V^\dagger V=\Pi_{\operatorname{supp}\rho_A}$,
$q_\lambda(\rho_A)q_\lambda(V^\dagger V)=q_\lambda(\rho_A)$, and
$\langle\Phi_m|(X\otimes Y)|\Phi_m\rangle=\Tr(XY^T)/m$, we obtain
\be
\|\widetilde\phi_\lambda\|^2
=
d_\lambda\Tr q_\lambda(\rho_A)
=
p_\lambda.
\ee
Dividing by $\sqrt{p_\lambda}$ yields \eqref{eq:pure-conjugate}.
\end{proof}
\end{lemma}

Although we will not use the conjugate-Schur version directly, it provides a more complete picture of the structure of the decomposition.

\subsection{Entanglement recovery from orthogonal isotropic errors}\label{app:polar}
We begin by considering the related problem of recovering entanglement after sending one half of a maximally entangled state through a quantum channel with a particular error structure. Specifically, consider a quantum channel
$\mathcal N:\mathcal L(\mathbb C^d)\to\mathcal L(\mathbb C^D)$ of the form
\be
\mathcal N(X)\coloneqq\frac1N\sum_{i=1}^N F_iXF_i^\dagger,
\label{eq:chann1}
\ee
whose error operators $F_i:\mathbb C^d\to\mathbb C^D$, $i=1,\ldots,N$ are Hilbert--Schmidt orthogonal, $\Tr(F_i^\dagger F_j)=d\,\delta_{ij}$, and which preserves the maximally mixed state,
$\mathcal N(\pi_d)=\pi_D$. Its Kraus operators are $F_i/\sqrt N$.
For a Haar-random code isometry $C:\mathbb C^K\to\mathbb C^d$, with encoding channel $\mathcal C(X)\coloneqq CXC^\dagger$, we seek a recovery channel $\mathcal R_C:\mathcal L(\mathbb C^D)\to\mathcal L(\mathbb C^K)$ with high average entanglement fidelity,
\be
F_e(\pi_K,\mathcal R_C\circ\mathcal N\circ\mathcal C)
=
\bra{\Phi_K}
(\idmap_K\otimes\mathcal R_C\circ\mathcal N\circ\mathcal C)(\Phi_K)
\ket{\Phi_K}.
\label{eq:entanglement-fidelity}
\ee
The same recovery must also handle arbitrary positive weights and coherences between the error operators. The following lemma defines the weighted map $\mathcal N_T$ and its encoded output $\rho_{C,T}$ and bounds the recovery error uniformly for every $T\ge0$. The choice $T=\pi_N$ recovers the reference channel. This uniformity lets a state-independent reference recovery correct the unknown weights of the physical Schur channel.

\begin{lemma}[Polar entanglement recovery from orthogonal isotropic error channels]\label{lem:normalized-polar}
Let $1\le K\le d$, $d>1$, and consider a quantum channel $\mathcal N:\mathcal L(\mathbb C^d)\to\mathcal L(\mathbb C^D)$ with Kraus decomposition $\mathcal N(X)=\frac1N\sum_{i=1}^N F_iXF_i^\dagger$ such that
\be
\Tr(F_i^\dagger F_j)=d\,\delta_{ij},\qquad
\sum_iF_iF_i^\dagger=\frac{dN}{D}\id_D.
\label{eq:normalized-errors}
\ee
For every operator $T=\sum_{i,j}T_{ij}\ket i\!\bra j\ge0$ on $\mathbb C^N$, define the completely positive map
\be
\mathcal N_T(X)
\coloneqq
\sum_{i,j=1}^N T_{ij}F_iXF_j^\dagger .
\label{eq:weighted-error-map}
\ee
Let $C:\mathbb C^K\to\mathbb C^d$ be a Haar-random isometry and assemble the restricted errors into
\be
W_C:\mathbb C^K\otimes\mathbb C^N\longrightarrow\mathbb C^D,
\qquad
W_C(\ket{x}\otimes\ket i)\coloneqq F_iC\ket{x},
\qquad
G_C\coloneqq W_C^\dagger W_C.
\label{eq:error-assembly}
\ee
Writing $W_C=U_C\sqrt{G_C}$ with the polar partial isometry $U_C:\mathbb C^K\otimes\mathbb C^N\to\mathbb C^D$, so that $U_C^\dagger U_C=\Pi_{\supp G_C}$ and $U_CU_C^\dagger=\Pi_{\mathrm{im}\,W_C}$, define the recovery map $\mathcal R_C:\mathcal L(\mathbb C^D)\to\mathcal L(\mathbb C^K)$ to be
\be
\mathcal R_C(Y)
\coloneqq
\Tr_{\mathbb C^N}(U_C^\dagger YU_C)
+
\Tr[(\id_D-U_CU_C^\dagger)Y]\,\tau_K,
\label{eq:polar-recovery}
\ee
where $\tau_K\in\mathcal D(\mathbb C^K)$ is any fixed density operator. Then $\mathcal R_C$ is CPTP.
Furthermore, for every $T\ge0$, the corresponding positive output operator $\rho_{C,T}\in\mathcal L(\mathbb C^K\otimes\mathbb C^D)$, which need not have unit trace,
\be
\rho_{C,T}
\coloneqq
(\idmap_K\otimes\mathcal N_T\circ\mathcal C)(\Phi_K)
\label{eq:weighted-code-output-definition}
\ee
satisfies
\be
\mathbb E_C\,
\Tr\!\left[
(\id-\Phi_K)
(\idmap_K\otimes\mathcal R_C)(\rho_{C,T})
\right]
\le
\frac{\beta}{1+\beta}\,\Tr T,
\qquad
\beta
\coloneqq
\frac{dK-1}{d^2-1}
\left(\frac{dN}{D}-1\right)
\le
\frac{KN}{D}-\frac Kd .
\label{eq:polar-beta}
\ee
In particular, if $T=\pi_N=\id_N/N$, then $\mathcal N = \mathcal N_{\pi_N}$ and
\be
\mathbb E_C\!\left[
1-F_e(\pi_K,\mathcal R_C\circ\mathcal N\circ\mathcal C)
\right]
\le
\frac{\beta}{1+\beta} .
\label{eq:polar-reference-fidelity}
\ee
\end{lemma}

The Gram matrix $G_C$ measures the failure of the different error images of the code to be mutually orthogonal. If $G_C=\id$, then $W_C$ is an isometry and the polar decoder separates the logical system from the error register exactly, simultaneously for every $T\ge0$. The lemma controls the deviation from this ideal situation. Its essential strengthening is that the Haar-averaged Gram defect is scalar on the error register, making the recovery bound uniform over arbitrary positive weights and coherences $T$.

\begin{proof}
Since $T\ge0$, the map $\mathcal N_T$ is completely positive, and expanding $W_C$ in the error basis gives 
\be
\rho_{C,T} =
(\idmap_K\otimes\mathcal N_T\circ\mathcal C)(\Phi_K)
=
(\id_K\otimes W_C)
(\Phi_K\otimes T)
(\id_K\otimes W_C^\dagger).
\label{eq:weighted-code-output}
\ee
Now take the partial isometry $U_C$ from the polar decomposition $W_C=U_C\sqrt{G_C}$, where $G_C = W_C^\dag W_C$. Both terms in $\mathcal R_C$ are completely positive, since $U_CU_C^\dagger$ is a projection, and their traces sum to $\Tr Y$. Hence $\mathcal R_C$ is trace preserving and therefore CPTP. Furthermore, by definition, the support of $\rho_{C,T}$ lies inside $\mathbb C^K\otimes\mathrm{Im}(W_C)=\mathbb C^K\otimes\mathrm{Im}(U_CU_C^\dagger)$; hence
\be
\Tr[\bigl(\id_K\otimes(\id_D-U_CU_C^\dagger)\bigr)\rho_{C,T}]=0.
\ee
Now choose a possibly unnormalized purification
$\ket{\vartheta_T}$ of $T$, and set
\be
\ket\chi \coloneqq \ket{\Phi_K}\otimes\ket{\vartheta_T}.
\ee
After applying the first term of the decoder $U_C^\dag$ to $\rho_{C,T}$, and before tracing the error and purifying registers, since $U_C^\dagger W_C=\sqrt{G_C}$, the resulting vector is $\sqrt{G_C}\ket\chi$. Using that $(\id-\Phi_K) \ket{\chi}=0$, we then have
\be
\Tr\!\left[
(\id-\Phi_K)
(\idmap_K\otimes\mathcal R_C)(\rho_{C,T})
\right]
&=
\bigl\|
(\id-\Phi_K)\sqrt{G_C}\ket\chi
\bigr\|^2 \\
&=
\bigl\|
(\id-\Phi_K)(\sqrt{G_C}-\id)\ket\chi
\bigr\|^2
\\
&\le
\bra\chi(\sqrt{G_C}-\id)^2\ket\chi \\
&=
\Tr\!\left[
(\pi_K\otimes T)(\sqrt{G_C}-\id)^2
\right] \\
&\leq
\Tr\!\left[
(\pi_K\otimes T)(G_C-\id)^2
\right],
\ee
where in the last line we used that $(\sqrt{x}-1)^2\le(x-1)^2$ when $x\ge0$. This proves 
\be
\Tr\!\left[
(\id-\Phi_K)
(\idmap_K\otimes\mathcal R_C)(\rho_{C,T})
\right]
\le
\Tr\!\left[
(\pi_K\otimes T)(G_C-\id)^2
\right].
\label{eq:gbound}
\ee
It remains to take the average of the right-hand side over the random code $\mathcal{C}$. 
Fix an isometry $C_0:\mathbb C^K\to\mathbb C^d$ and write explicitly
\be
C=UC_0,\qquad
P_0=C_0C_0^\dagger,\qquad
P=CC^\dagger=UP_0U^\dagger,
\ee
with $U\in\mathcal U(d)$ Haar random. Thus $P$ is a Haar-random rank-$K$ projector. A standard calculation gives
\be
\mathbb E_U\,\Tr(PAPB)
=
\frac{
K(dK-1)\Tr(AB)
+
K(d-K)\Tr A\,\Tr B
}{
d(d^2-1)
}.
\label{eq:haar-projector-moment}
\ee
Since $(G_C)_{ij}=C^\dagger F_i^\dagger F_jC$, the two contractions entering the second moment of \cref{eq:gbound} are
\be
\sum_j
\Tr(F_i^\dagger F_jF_j^\dagger F_l)
&=
\Tr\!\left[
F_i^\dagger
\Bigl(\sum_jF_jF_j^\dagger\Bigr)
F_l
\right]
=
\frac{d^2N}{D}\,\delta_{il},
\\
\sum_j
\Tr(F_i^\dagger F_j)\,
\Tr(F_j^\dagger F_l)
&=
d^2\,\delta_{il},
\ee
where we used the properties assumed in \cref{eq:normalized-errors}.
Substituting these into \eqref{eq:haar-projector-moment} and using the 1-design identity $\mathbb E_U P=\frac Kd\,\id_d$ gives
\be
\frac1K
\mathbb E_U\,
\Tr_{\mathbb C^K}G_C^2
&=
\frac{
d(dK-1)N/D+d(d-K)
}{
d^2-1
}\id_N,
\\
\frac1K
\mathbb E_U\,
\Tr_{\mathbb C^K}G_C
&=
\id_N.
\ee
Therefore
\be
\frac1K
\mathbb E_U\,
\Tr_{\mathbb C^K}(G_C-\id)^2
=
\beta\,\id_N.
\label{eq:scalar-gram-moment}
\ee
Hence, for every fixed $T\ge0$,
\be
\mathbb E_U\,
\Tr\!\left[
(\pi_K\otimes T)(G_C-\id)^2
\right] &= \frac{1}{K} \mathbb E_U\,
\Tr\!\left[
(\id_K\otimes T)(G_C-\id)^2
\right] \\
&= \frac{1}{K} \Tr_N \left[ T \mathbb E_U \left[ \Tr_K (G_C-\id)^2 \right] \right] \\
&=\beta\,\Tr T.
\ee
This establishes the quadratic Gram-defect estimate used below. Finally, $\sum_iF_iF_i^\dagger=\frac{dN}{D}\id_D$ implies $D\leq \sum_i\rank F_i \leq dN$, so $\beta\ge0$. Moreover, since $K\leq d$
\be
\frac{dK-1}{d^2-1}\le\frac Kd,
\ee
which yields
\be
\beta
\le
\frac{KN}{D}-\frac Kd.
\ee
To obtain the stronger recovery bound, assume $t\coloneqq \Tr T>0$; the case $T=0$ is immediate. Let $A_T\coloneqq\pi_K\otimes T/t$ and form the probability measure obtained by averaging the spectral measure of $G_C$ weighted by $A_T$. If $X$ denotes its nonnegative spectral variable, the first and second moments computed above give
\be
\mathbb E X=1,\qquad \mathbb E X^2=1+\beta.
\ee
H\"older's inequality therefore implies
\be
1=\mathbb E X
\le (\mathbb E\sqrt X)^{2/3}(\mathbb E X^2)^{1/3},
\qquad
\mathbb E_C\Tr(A_T\sqrt{G_C})
\ge\frac1{\sqrt{1+\beta}}.
\label{eq:je-polar-moment-improvement}
\ee
In the purification used above, the Bell projector contains the normalized vector $\ket\chi/\sqrt t$ in its range. Consequently its success weight $f_C$ satisfies
\be
f_C\coloneqq \Tr\!\left[\Phi_K(\idmap_K\otimes\mathcal R_C)(\rho_{C,T})\right]
\ge\frac{|\bra\chi\sqrt{G_C}\ket\chi|^2}{t}
=t\bigl[\Tr(A_T\sqrt{G_C})\bigr]^2.
\ee
Jensen's inequality gives $\mathbb E_C f_C\ge t/(1+\beta)$. Since $\mathbb E_C\Tr\rho_{C,T}=t$, subtracting the success weight proves
\be
\mathbb E_C\Tr\!\left[(\id-\Phi_K)(\idmap_K\otimes\mathcal R_C)(\rho_{C,T})\right]
\le\frac{\beta}{1+\beta}\Tr T.
\ee
This improves the quadratic estimate for every $\beta>0$ and uses the same recovery map. The weaker bound by $\beta\Tr T$ remains available in the subsequent packing argument.
This concludes the proof.
\end{proof}

\section{A tripartite Schur--Weyl decomposition}\label{app:tripartite}
We now relate the local Schur transforms on $B,E$ to the global one on $BE$ and derive the properties of the resulting intertwiner. This will then be used to derive the tripartite pure-state Schur decomposition.

\begin{theorem}[Canonical decomposition of the local-to-global Schur intertwiner]\label{thm:local-global-schur}
Define the unitary operator
\be
R^{BE}\coloneqq
(U_{\mathrm{Schur}}^B\otimes U_{\mathrm{Schur}}^E)
(U_{\mathrm{Schur}}^{BE})^\dagger.
\ee
Then $R^{BE}$ intertwines the global and local representations of
$\mathcal U(d_B)\times\mathcal U(d_E)\times S_k$, namely
\be
R^{BE}R_{\mathrm{Schur}}^{BE}(U_B\otimes U_E,\pi)
=
\bigl(R_{\mathrm{Schur}}^B(U_B,\pi)\otimes
R_{\mathrm{Schur}}^E(U_E,\pi)\bigr)R^{BE}.
\label{eq:intert}
\ee
Define the branch co-isometries by
\be
V_{\mu,\nu}^{\lambda,\alpha}
&=
(\bra\lambda\otimes\id_{\mathcal P_\lambda}\otimes\bra\alpha)V_{\mu,\nu}:
\mathcal P_\mu^B\otimes\mathcal P_\nu^E
\to\mathcal P_\lambda^{BE},
\label{eq:decom}
\ee
and
\be
W_{\mu,\nu}^{\lambda,\beta}
&=
(\bra{\mu,\nu}\otimes
\id_{\mathcal Q_\mu^B\otimes\mathcal Q_\nu^E}\otimes\bra\beta)W_\lambda:
\mathcal Q_\lambda^{BE}
\to\mathcal Q_\mu^B\otimes\mathcal Q_\nu^E.
\label{eq:W-branch-def}
\ee
Then there are unitary matrices
$u^{\lambda,\mu,\nu}\in\mathcal U(g_{\lambda\mu\nu})$, which we write as
\be
u^{\lambda,\mu,\nu}
=
\sum_{\alpha,\beta=1}^{g_{\lambda\mu\nu}}
u_{\alpha\beta}^{\lambda,\mu,\nu}\ket\alpha\bra\beta,
\ee
such that
\be
R^{BE}
=
\sum_{\lambda,\mu,\nu}
\ket{\mu,\nu}\bra\lambda\otimes
\sum_{\alpha,\beta=1}^{g_{\lambda\mu\nu}}
u_{\alpha\beta}^{\lambda,\mu,\nu}
(V_{\mu,\nu}^{\lambda,\alpha})^\dagger
\otimes
W_{\mu,\nu}^{\lambda,\beta}.
\label{eq:R-final}
\ee

\begin{proof}
We first prove \eqref{eq:intert}. Applying \eqref{eq:schurint} locally on $B$ and $E$, with the same permutation on both systems, gives
\be
&(U_{\mathrm{Schur}}^B\otimes U_{\mathrm{Schur}}^E)
(U_B^{\otimes k}\otimes U_E^{\otimes k})
(R_\pi^B\otimes R_\pi^E)\\
&\qquad=
\bigl(R_{\mathrm{Schur}}^B(U_B,\pi)\otimes
R_{\mathrm{Schur}}^E(U_E,\pi)\bigr)
(U_{\mathrm{Schur}}^B\otimes U_{\mathrm{Schur}}^E).
\ee
The global Schur transform obeys
\be
U_{\mathrm{Schur}}^{BE}
(U_B\otimes U_E)^{\otimes k}
(R_\pi^B\otimes R_\pi^E)
=
R_{\mathrm{Schur}}^{BE}(U_B\otimes U_E,\pi)
U_{\mathrm{Schur}}^{BE}.
\ee
Combining these identities with the definition of $R^{BE}$ proves
\eqref{eq:intert}. Thus $R^{BE}$ is a unitary intertwiner mapping
\be
\bigoplus_\lambda
\mathcal P_\lambda^{BE}\otimes\mathcal Q_\lambda^{BE}
\longrightarrow
\bigoplus_{\mu,\nu}
\mathcal P_\mu^B\otimes\mathcal P_\nu^E
\otimes\mathcal Q_\mu^B\otimes\mathcal Q_\nu^E.
\label{eq:R-domain-codomain}
\ee
Now decompose
$R^{BE}
=\sum_{\lambda,\mu,\nu}
\ket{\mu,\nu}\bra\lambda\otimes
R^{BE}_{\mu,\nu:\lambda}$.
Since the Schur representations are block diagonal in these labels, each
$R^{BE}_{\mu,\nu:\lambda}$ is itself an intertwiner.
As throughout, we use the fixed ordering $\mathcal P\otimes\mathcal Q$.
Applying $W_\lambda$ to resolve the incoming global unitary factor and
$V_{\mu,\nu}$ to resolve the outgoing local permutation factors, define the rotated operator
\be
\widetilde R^{BE}_{\mu,\nu:\lambda}
\coloneqq
(V_{\mu,\nu}\otimes
\id_{\mathcal Q_\mu^B\otimes\mathcal Q_\nu^E})
R^{BE}_{\mu,\nu:\lambda}
(\id_{\mathcal P_\lambda}\otimes W_\lambda^\dagger).
\ee
Up to a fixed reordering of the multiplicity factors, its domain and codomain are
\be
\bigoplus_{\mu',\nu'}
\mathcal P_\lambda^{BE}\otimes
\mathcal Q_{\mu'}^B\otimes\mathcal Q_{\nu'}^E
\otimes\mathbb C^{g_{\lambda\mu'\nu'}},
\ee
and
\be
\bigoplus_{\lambda'}
\mathcal P_{\lambda'}^{BE}\otimes
\mathcal Q_\mu^B\otimes\mathcal Q_\nu^E
\otimes\mathbb C^{g_{\lambda'\mu\nu}},
\ee
respectively. The first three factors in every summand form an irreducible representation of
$\mathcal U(d_B)\times\mathcal U(d_E)\times S_k$.
Schur's lemma therefore allows a nonzero component only between matching irreps, namely
\be
\lambda'=\lambda,\qquad
\mu'=\mu,\qquad
\nu'=\nu.
\label{eq:schur1}
\ee
Hence the rotated block acts trivially on the common irreducible factor and only on the multiplicity space:
\be
\widetilde R^{BE}_{\mu,\nu:\lambda}
=
\id_{\mathcal P_\lambda^{BE}\otimes
\mathcal Q_\mu^B\otimes\mathcal Q_\nu^E}
\otimes u^{\lambda,\mu,\nu}.
\ee
The direct sums of $V_{\mu,\nu}$ and $W_\lambda$ are unitary changes of basis, so the fully rotated $R^{BE}$ remains unitary. Since it is block diagonal in the mutually orthogonal triples $(\lambda,\mu,\nu)$, each $u^{\lambda,\mu,\nu}$ is a unitary in the multiplicity space:
\be
u^{\lambda,\mu,\nu}\in\mathcal U(g_{\lambda\mu\nu}).
\ee
Finally, undoing the two branching transforms and expanding $u^{\lambda,\mu,\nu}$ in the multiplicity basis gives \eqref{eq:R-final}. Here $\beta$ is the incoming multiplicity index supplied by $W_\lambda$ and $\alpha$ is the outgoing multiplicity index supplied by $V_{\mu,\nu}$, so $u^{\lambda,\mu,\nu}$ is exactly the multiplicity-space unitary acting under the global-to-local change of Schur coordinates. It depends only on the relative choice of the two branching bases.
\end{proof}
\end{theorem}

We now derive a few additional properties of the operators appearing in \cref{eq:R-final}.

\begin{lemma}[Completeness and orthogonality of the branch co-isometries]\label{lem:orth}
The branch operators in \cref{eq:R-final} obey
\be
V_{\mu,\nu}^{\lambda,\beta}(V_{\mu,\nu}^{\lambda',\gamma})^\dagger
&=\delta_{\lambda\lambda'}\delta_{\beta\gamma}\id_{\mathcal P_\lambda^{BE}},\\
W_{\mu,\nu}^{\lambda,\alpha}(W_{\mu',\nu'}^{\lambda,\delta})^\dagger
&=\delta_{\mu\mu'}\delta_{\nu\nu'}\delta_{\alpha\delta}
\id_{\mathcal Q_\mu^B\otimes\mathcal Q_\nu^E}.
\label{eq:orth1}
\ee
Moreover,
\begin{eqnarray}
\sum_{\lambda,\beta}(V_{\mu,\nu}^{\lambda,\beta})^\dagger
V_{\mu,\nu}^{\lambda,\beta}
&=&\id_{\mathcal P_\mu^B\otimes\mathcal P_\nu^E},\\
\sum_{\mu,\nu,\alpha}(W_{\mu,\nu}^{\lambda,\alpha})^\dagger
W_{\mu,\nu}^{\lambda,\alpha}
&=&\id_{\mathcal Q_\lambda^{BE}}.
\label{eq:compl}
\end{eqnarray}
\begin{proof}
First, from the definition of the branch operators and the unitarity of $V_{\mu,\nu}$, we have
\be
V_{\mu,\nu}^{\lambda,\beta}(V_{\mu,\nu}^{\lambda',\gamma})^\dagger
&=(\bra\lambda\otimes\id\otimes\bra\beta)V_{\mu,\nu}V_{\mu,\nu}^\dagger
(\ket{\lambda'}\otimes\id\otimes\ket\gamma)\\
&=\delta_{\lambda\lambda'}\delta_{\beta\gamma}\id_{\mathcal P_\lambda^{BE}}.
\ee
Similarly, $W_\lambda W_\lambda^\dagger=\id$ gives the second orthogonality relation. This proves \eqref{eq:orth1}.
We also have the resolutions of the identity:
\be
\begin{aligned}
\sum_{\lambda,\beta}
(V_{\mu,\nu}^{\lambda,\beta})^\dagger
V_{\mu,\nu}^{\lambda,\beta}
&=
V_{\mu,\nu}^\dagger
\left(
\sum_{\lambda,\beta}
\ketbra{\lambda}{\lambda}\otimes\id_{\mathcal P_\lambda^{BE}}
\otimes\ketbra{\beta}{\beta}
\right)
V_{\mu,\nu}\\
&=
V_{\mu,\nu}^\dagger V_{\mu,\nu}
=
\id_{\mathcal P_\mu^B\otimes\mathcal P_\nu^E}.
\end{aligned}
\ee
Replacing $V_{\mu,\nu}$ by $W_\lambda$ and summing over $(\mu,\nu,\alpha)$ proves the second line in \eqref{eq:compl}.
\end{proof}
\end{lemma}

The following result has direct implications for the entanglement structure of the tripartite decomposition.

\begin{lemma}[Orthogonality after partial trace]\label{lem:parttrort}
The following relations hold
\begin{eqnarray}
\Tr_{\mathcal P_\mu^B}\!\left[(V_{\mu,\nu}^{\lambda,\alpha})^\dagger
V_{\mu,\nu'}^{\lambda,\beta}\right]
&=&\delta_{\nu\nu'}\delta_{\alpha\beta}\frac{d_\lambda}{d_\nu}
\id_{\mathcal P_\nu^E},\\
\Tr_{\mathcal P_\nu^E}\!\left[(V_{\mu,\nu}^{\lambda,\alpha})^\dagger
V_{\mu',\nu}^{\lambda,\beta}\right]
&=&\delta_{\mu\mu'}\delta_{\alpha\beta}\frac{d_\lambda}{d_\mu}
\id_{\mathcal P_\mu^B}.
\label{eq:part}
\end{eqnarray}
Extending each branch map by zero on the other summands of $\bigoplus_{\mu,\nu}\mathcal P_\mu^B\otimes\mathcal P_\nu^E$, the following Hilbert--Schmidt orthogonality condition holds
\be
\Tr\!\left[(V_{\mu,\nu}^{\lambda,\alpha})^\dagger
V_{\mu',\nu'}^{\lambda,\beta}\right]
=\delta_{\mu\mu'}\delta_{\nu\nu'}\delta_{\alpha\beta}d_\lambda.
\label{eq:branch-hs}
\ee
\begin{proof}
Set
\be
X_{\nu,\nu'}^{\alpha,\beta}
\coloneqq\Tr_{\mathcal P_\mu^B}\!\left[(V_{\mu,\nu}^{\lambda,\alpha})^\dagger
V_{\mu,\nu'}^{\lambda,\beta}\right]:
\mathcal P_{\nu'}^E\longrightarrow\mathcal P_\nu^E.
\ee
We first show that $X_{\nu,\nu'}^{\alpha,\beta}$ intertwines the symmetric group actions $p_{\nu'}(\pi)$ and $p_\nu(\pi)$. From \eqref{eq:Vmunu-intertwining}, projection gives
\be
V_{\mu,\nu}^{\lambda,\alpha}(p_\mu(\pi)\otimes p_\nu(\pi))
=p_\lambda(\pi)V_{\mu,\nu}^{\lambda,\alpha}.
\ee
Taking its adjoint and replacing $\pi$ by $\pi^{-1}$ gives
$(p_\mu(\pi)\otimes p_\nu(\pi))(V_{\mu,\nu}^{\lambda,\alpha})^\dagger
=(V_{\mu,\nu}^{\lambda,\alpha})^\dagger p_\lambda(\pi)$.
Combining the two relations yields
\be
(p_\mu(\pi)\otimes p_\nu(\pi))
(V_{\mu,\nu}^{\lambda,\alpha})^\dagger V_{\mu,\nu'}^{\lambda,\beta}
=(V_{\mu,\nu}^{\lambda,\alpha})^\dagger V_{\mu,\nu'}^{\lambda,\beta}
(p_\mu(\pi)\otimes p_{\nu'}(\pi)).
\ee
Multiply on the left by $p_\mu(\pi)^\dagger\otimes\id$ and take the partial trace over $B$. Since this trace is invariant under conjugation on $B$, we obtain
\be
p_\nu(\pi)\Tr_{\mathcal P_\mu^B}\!\left[(V_{\mu,\nu}^{\lambda,\alpha})^\dagger
V_{\mu,\nu'}^{\lambda,\beta}\right]
=\Tr_{\mathcal P_\mu^B}\!\left[(V_{\mu,\nu}^{\lambda,\alpha})^\dagger
V_{\mu,\nu'}^{\lambda,\beta}\right]p_{\nu'}(\pi).
\ee
Equivalently,
\be
p_\nu(\pi)X_{\nu,\nu'}^{\alpha,\beta}
=X_{\nu,\nu'}^{\alpha,\beta}p_{\nu'}(\pi),\qquad\pi\in S_k.
\ee
If $\nu\ne\nu'$, Schur's lemma gives $X_{\nu,\nu'}^{\alpha,\beta}=0$. If $\nu=\nu'$, it gives $X_{\nu,\nu}^{\alpha,\beta}=c_{\alpha\beta}\id_{\mathcal P_\nu^E}$. It remains to determine the scalar. Taking the trace and using \cref{lem:orth},
\be
c_{\alpha\beta}d_\nu
&=\Tr[(V_{\mu,\nu}^{\lambda,\alpha})^\dagger V_{\mu,\nu}^{\lambda,\beta}]\\
&=\Tr[V_{\mu,\nu}^{\lambda,\beta}(V_{\mu,\nu}^{\lambda,\alpha})^\dagger]
=\delta_{\alpha\beta}d_\lambda.
\ee
This proves the first identity in \cref{eq:part}. Interchanging $B$ and $E$ proves the second, and taking a trace gives \eqref{eq:branch-hs}.
\end{proof}
\end{lemma}
We now use the decomposition of the Schur--Weyl intertwiner in \cref{thm:local-global-schur} to derive the corresponding result of \cref{lem:pureschurweyl} for an arbitrary purification $\ket\psi_{ABE}$ of $\rho_{AB}$. Here, Eve's Schur transform is only a choice of coordinates on the purifying system, which does not affect any protocol acting on $AB$ alone. In particular, choosing the $A:BE$ splitting we get:

\begin{theorem}[Schur--Weyl purification]\label{thm:schurpurif}
Let $p_\lambda$ and $\ket{\phi_\lambda}_{A:BE}$ be the respective quantities in
\cref{lem:pureschurweyl} for the bipartition $A:BE$. Then
\be
&(U_{\mathrm{Schur}}^A\otimes U_{\mathrm{Schur}}^B\otimes U_{\mathrm{Schur}}^E)
\ket\psi_{ABE}^{\otimes k}\\
&=
\sum_{\substack{\lambda,\mu,\nu\vdash k\\
\ell(\lambda)\le r_A,\;
\ell(\mu)\le r_B,\;
\ell(\nu)\le r_E}}
\sqrt{p_\lambda}\,
\ket\lambda_A\ket\mu_B\ket\nu_E
\otimes
\sum_{\alpha=1}^{g_{\lambda\mu\nu}}
\bigl(\id\otimes(V_{\mu,\nu}^{\lambda,\alpha})^\dagger\bigr)
\ket*{\Phi_{\mathcal P_\lambda}}_{A:BE}
\otimes
\bigl(\id\otimes W_{\mu,\nu}^{\lambda,\alpha}\bigr)
\ket{\phi_\lambda}_{A:BE}.
\label{eq:tripartite-schur}
\ee
Here $V_{\mu,\nu}^{\lambda,\alpha}$ and $W_{\mu,\nu}^{\lambda,\alpha}$ are the maps defined in \cref{thm:local-global-schur}. In particular, the multiplicity bases of $W_\lambda$ are chosen, without loss of generality, so as to absorb the multiplicity unitaries $u^{\lambda,\mu,\nu}$ of \cref{thm:local-global-schur}; specifically, the resulting $W_\lambda$ are again valid unitary branching intertwiners. Terms with $g_{\lambda\mu\nu}=0$ are understood to be absent. Equivalently, the right-hand side of \cref{eq:tripartite-schur} is a purification of the state obtained by applying the Schur transform only to Alice and Bob.

\begin{proof}
We first justify the choice of multiplicity bases made in the statement. In the notation of \cref{thm:local-global-schur}, define
\be
\widetilde W_{\mu,\nu}^{\lambda,\alpha}
\coloneqq
\sum_{\beta=1}^{g_{\lambda\mu\nu}}
u_{\alpha\beta}^{\lambda,\mu,\nu}
W_{\mu,\nu}^{\lambda,\beta}.
\label{eq:rotated-W}
\ee
For fixed $(\lambda,\mu,\nu)$, this is a unitary change of basis in the Kronecker multiplicity space. Equivalently, $\widetilde W_\lambda$ is obtained from $W_\lambda$ by a block-diagonal unitary acting only on the multiplicity factors. Since these factors carry the trivial representation, this unitary commutes with
$q_\mu^B(U_B)\otimes q_\nu^E(U_E)\otimes\id_{g_{\lambda\mu\nu}}$.
Hence $\widetilde W_\lambda$ is again a unitary intertwiner satisfying the same branching relation as $W_\lambda$.
The same unitary rotation also preserves the branch orthogonality and completeness relations. Indeed,
\be
\begin{aligned}
\widetilde W_{\mu,\nu}^{\lambda,\alpha}
(\widetilde W_{\mu',\nu'}^{\lambda,\delta})^\dagger
&=
\delta_{\mu\mu'}\delta_{\nu\nu'}
\sum_\gamma
u_{\alpha\gamma}^{\lambda,\mu,\nu}
\overline{u_{\delta\gamma}^{\lambda,\mu,\nu}}\,
\id_{\mathcal Q_\mu^B\otimes\mathcal Q_\nu^E}\\
&=
\delta_{\mu\mu'}\delta_{\nu\nu'}\delta_{\alpha\delta}
\id_{\mathcal Q_\mu^B\otimes\mathcal Q_\nu^E},
\end{aligned}
\ee
while
\be
\sum_{\mu,\nu,\alpha}
(\widetilde W_{\mu,\nu}^{\lambda,\alpha})^\dagger
\widetilde W_{\mu,\nu}^{\lambda,\alpha}
=
\id_{\mathcal Q_\lambda^{BE}}.
\ee
We may therefore absorb this rotation into the definition of $W_\lambda$ and omit the tilde. With this convention, \eqref{eq:R-final} becomes
\be
R^{BE}
=
\sum_{\lambda,\mu,\nu}
\ket{\mu,\nu}\bra\lambda\otimes
\sum_{\alpha=1}^{g_{\lambda\mu\nu}}
(V_{\mu,\nu}^{\lambda,\alpha})^\dagger
\otimes
W_{\mu,\nu}^{\lambda,\alpha}.
\label{eq:R-aligned}
\ee
Now view the purification as a bipartite pure state across $A:BE$. By
\cref{lem:pureschurweyl},
\be
(U_{\mathrm{Schur}}^A\otimes U_{\mathrm{Schur}}^{BE})
\ket\psi_{ABE}^{\otimes k}
=
\sum_{\substack{\lambda\vdash k\\ \ell(\lambda)\le r_A}}
\sqrt{p_\lambda}\,
\ket\lambda_A\ket\lambda_{BE}
\otimes
\ket{\Phi_{\mathcal P_\lambda}}_{A:BE}
\otimes
\ket{\phi_\lambda}_{A:BE}.
\ee
The remaining change of basis acts only on $BE$. By the definition
$R^{BE}=(U_{\mathrm{Schur}}^B\otimes U_{\mathrm{Schur}}^E)
(U_{\mathrm{Schur}}^{BE})^\dagger$,
\be
U_{\mathrm{Schur}}^A\otimes U_{\mathrm{Schur}}^B\otimes U_{\mathrm{Schur}}^E
=
(\id\otimes R^{BE})
(U_{\mathrm{Schur}}^A\otimes U_{\mathrm{Schur}}^{BE}).
\ee
Using \eqref{eq:R-aligned} on a fixed global label gives
\be
R^{BE}
(\ket\lambda\otimes\ket{S_\lambda}\ket{T_\lambda})
=
\sum_{\mu,\nu,\alpha}
\ket{\mu,\nu}\otimes
(V_{\mu,\nu}^{\lambda,\alpha})^\dagger\ket{S_\lambda}
\otimes
W_{\mu,\nu}^{\lambda,\alpha}\ket{T_\lambda}.
\ee
Applying this identity to the $BE$ factors of each summand yields
\eqref{eq:tripartite-schur}, initially with the corresponding local length constraints.
It remains to sharpen these constraints to the ranks of the reduced states.
For Bob, let
$\Pi_{>r_B}^B\coloneqq\sum_{\ell(\mu)>r_B}\Pi_\mu^B$.
Since $\rho_B$ has rank $r_B$, the state $\rho_B^{\otimes k}$ is supported on
$(\operatorname{supp}\rho_B)^{\otimes k}\cong(\mathbb C^{r_B})^{\otimes k}$,
whose Schur--Weyl decomposition contains no sectors with $\ell(\mu)>r_B$.
Hence
\be
\Tr\!\left[\Pi_{>r_B}^B\rho_B^{\otimes k}\right]=0.
\ee
Using $\rho_B^{\otimes k}=\Tr_{AE}[\psi_{ABE}^{\otimes k}]$, this is equivalently
\be
\bra{\psi}^{\otimes k}
\bigl(\id_{A^{\otimes k}}\otimes\Pi_{>r_B}^B\otimes
\id_{E^{\otimes k}}\bigr)
\ket{\psi}^{\otimes k}
=
0.
\ee
Since the operator inside the expectation value is a projector,
\be
\bigl(\id\otimes\Pi_{>r_B}^B\otimes\id\bigr)
\ket{\psi}^{\otimes k}
=
0.
\ee
Thus every component with $\ell(\mu)>r_B$ vanishes identically. The same argument applied to $E$ shows that every component with $\ell(\nu)>r_E$ vanishes. Together with the bipartite $A:BE$ decomposition, which already gives $\ell(\lambda)\le r_A$, this yields the three restrictions appearing in \eqref{eq:tripartite-schur}.
Finally, tracing out Eve leaves the reduced state obtained by applying the Schur transforms only to $A$ and $B$, proving the purification statement.
\end{proof}
\end{theorem}

We now define the state vectors in the permutation sector:
\be
\ket*{B_{\mu,\nu}^{\lambda,\alpha}}_P
\coloneqq
(\id_{\mathcal P_\lambda^A}\otimes
(V_{\mu,\nu}^{\lambda,\alpha})^\dagger)
\ket*{\Phi_{\mathcal P_\lambda}}_{A:BE}.
\label{eq:invariant-B}
\ee
Each branch is invariant under the diagonal tripartite permutation action,
\be
\bigl(
p_\lambda(\pi)\otimes p_\mu(\pi)\otimes p_\nu(\pi)
\bigr)
\ket*{B_{\mu,\nu}^{\lambda,\alpha}}_P
=
\ket*{B_{\mu,\nu}^{\lambda,\alpha}}_P,
\qquad \pi\in S_k.
\label{eq:B-diagonal-invariance}
\ee
Indeed, using the intertwining relation for $V_{\mu,\nu}^{\lambda,\alpha}$ and the invariance of the Bell state vector, we obtain
\be
\begin{aligned}
\bigl(
p_\lambda(\pi)\otimes p_\mu(\pi)\otimes p_\nu(\pi)
\bigr)
\ket*{B_{\mu,\nu}^{\lambda,\alpha}}_P
&=
\bigl(\id\otimes(V_{\mu,\nu}^{\lambda,\alpha})^\dagger\bigr)
\bigl(p_\lambda(\pi)\otimes p_\lambda(\pi)\bigr)
\ket*{\Phi_{\mathcal P_\lambda}}_{A:BE}\\
&=
\ket*{B_{\mu,\nu}^{\lambda,\alpha}}_P.
\end{aligned}
\ee
To get more intuition about the properties of these state vectors, we consider the conditional min-entropy defined as
\be
H_{\min}(A|E)_\omega
\coloneqq
\sup_{\sigma_E\in\mathcal D(E)}
\sup\left\{
-\log_2 t:t>0,\quad
\omega_{AE}\le t\,\id_A\otimes\sigma_E
\right\},
\ee
and the optimized sandwiched conditional collision entropy \cite{MullerLennertEtAl2013,TomamichelBertaHayashi2014}
\be
\widetilde H_2^\uparrow(A|E)_\omega
\coloneqq
\sup_{\sigma_E\in\mathcal D(E)}
-\log_2\Tr\!\left[
\left(
(\id_A\otimes\sigma_E^{-1/4})
\omega_{AE}
(\id_A\otimes\sigma_E^{-1/4})
\right)^2
\right],
\ee
where $\omega_{AE}\in\mathcal D(AE)$ and negative powers act on $\supp\sigma_E$. In the collision-entropy optimization, a reference with $\supp\omega_E\not\subseteq\supp\sigma_E$ contributes $-\infty$. 
These conditional entropies are natural one-shot measures of decoupling and hence of one-shot entanglement distillation, with $H_{\min}$ controlling the recoverable entanglement and $\widetilde H_2^\uparrow$ arising naturally in random-coding and second-moment bounds \cite{DupuisBertaWullschlegerRenner2014,BuscemiDatta2010}.
Interestingly, the tripartite permutation-invariant state vectors $\ket*{B_{\mu,\nu}^{\lambda,\alpha}}_P$ already encode, at the one-shot level, the entanglement quantity whose asymptotic limit yields the coherent information. We make this precise in the following:

\begin{lemma}[Entanglement structure of the permutation sector]\label{lem:orthperm}
For fixed $(\lambda,\mu)$, the state vectors in \eqref{eq:invariant-B} are orthonormal in $(\nu,\alpha)$:
\be
\braket*{B_{\mu,\nu}^{\lambda,\alpha}}
{B_{\mu,\nu'}^{\lambda,\beta}}
=
\delta_{\nu\nu'}\delta_{\alpha\beta}.
\label{eq:B-orthogonality}
\ee
Moreover, every branch has maximally mixed one-party marginals,
\be
(B_{\mu,\nu}^{\lambda,\alpha})_A
=
\frac{\id_{\mathcal P_\lambda}}{d_\lambda},
\qquad
(B_{\mu,\nu}^{\lambda,\alpha})_B
=
\frac{\id_{\mathcal P_\mu}}{d_\mu},
\qquad
(B_{\mu,\nu}^{\lambda,\alpha})_E
=
\frac{\id_{\mathcal P_\nu}}{d_\nu}.
\label{eq:invariant-marginals}
\ee
Consequently,
\be
H_{\min}(A|E)_{B_{\mu,\nu}^{\lambda,\alpha}}
=
\widetilde H_2^\uparrow(A|E)_{B_{\mu,\nu}^{\lambda,\alpha}}
=
H(A|E)_{B_{\mu,\nu}^{\lambda,\alpha}}
=
\log_2 d_\mu-\log_2 d_\nu.
\label{eq:permutation-entropy}
\ee

\begin{proof}
For $\nu=\nu'$, the transpose trick and \cref{lem:orth} give
\be
\braket*{B_{\mu,\nu}^{\lambda,\alpha}}
{B_{\mu,\nu}^{\lambda,\beta}}
=
\frac{1}{d_\lambda}
\Tr\!\left[
V_{\mu,\nu}^{\lambda,\alpha}
(V_{\mu,\nu}^{\lambda,\beta})^\dagger
\right]
=
\delta_{\alpha\beta},
\ee
while different $\nu$ belong to orthogonal irrep sectors. This proves
\eqref{eq:B-orthogonality}.
The marginal on $A$ follows directly from the co-isometry relation,
\be
(B_{\mu,\nu}^{\lambda,\alpha})_A
=
\frac{
\left[
V_{\mu,\nu}^{\lambda,\alpha}
(V_{\mu,\nu}^{\lambda,\alpha})^\dagger
\right]^T}{d_\lambda}
=
\frac{\id_{\mathcal P_\lambda^A}}{d_\lambda}.
\ee
Using \cref{lem:parttrort},
\be
(B_{\mu,\nu}^{\lambda,\alpha})_B
&=
\frac{1}{d_\lambda}
\Tr_E\!\left[
(V_{\mu,\nu}^{\lambda,\alpha})^\dagger
V_{\mu,\nu}^{\lambda,\alpha}
\right]
=
\frac{\id_{\mathcal P_\mu^B}}{d_\mu},\\
(B_{\mu,\nu}^{\lambda,\alpha})_E
&=
\frac{1}{d_\lambda}
\Tr_B\!\left[
(V_{\mu,\nu}^{\lambda,\alpha})^\dagger
V_{\mu,\nu}^{\lambda,\alpha}
\right]
=
\frac{\id_{\mathcal P_\nu^E}}{d_\nu}.
\ee
This proves \eqref{eq:invariant-marginals}.
Since $\ket*{B_{\mu,\nu}^{\lambda,\alpha}}$ is pure,
\be
H(A|E)_{B_{\mu,\nu}^{\lambda,\alpha}}
=
H(B)_{B_{\mu,\nu}^{\lambda,\alpha}}
-
H(E)_{B_{\mu,\nu}^{\lambda,\alpha}}
=
\log_2 d_\mu-\log_2 d_\nu.
\ee
For the min-entropy, choose
$\sigma_E=\pi_E\coloneqq\id_{\mathcal P_\nu}/d_\nu$.
Since
$(B_{\mu,\nu}^{\lambda,\alpha})_{AE}$
and
$(B_{\mu,\nu}^{\lambda,\alpha})_B$
have the same nonzero spectrum,
\be
\left\|
(B_{\mu,\nu}^{\lambda,\alpha})_{AE}
\right\|_\infty
=
\frac{1}{d_\mu}.
\ee
Therefore
\be
(B_{\mu,\nu}^{\lambda,\alpha})_{AE}
\le
\frac{1}{d_\mu}\id_{AE}
=
\frac{d_\nu}{d_\mu}\,
\id_A\otimes\pi_E,
\ee
and hence
\be
H_{\min}(A|E)_{B_{\mu,\nu}^{\lambda,\alpha}}
\ge
\log_2 d_\mu-\log_2 d_\nu.
\ee
For the collision entropy, the same choice $\sigma_E=\pi_E$ gives
\be
\widetilde H_2^\uparrow(A|E)_{B_{\mu,\nu}^{\lambda,\alpha}}
&\ge
-\log_2\!\left[
d_\nu
\Tr\!\left(
(B_{\mu,\nu}^{\lambda,\alpha})_{AE}^2
\right)
\right]\\
&=
-\log_2\!\left[
d_\nu
\Tr\!\left(
(B_{\mu,\nu}^{\lambda,\alpha})_B^2
\right)
\right]
=
\log_2 d_\mu-\log_2 d_\nu,
\ee
where $\Tr[(B_{\mu,\nu}^{\lambda,\alpha})_B^2]=1/d_\mu$.
Finally, $H_{\min}(A|E)\le \widetilde H_2^\uparrow(A|E)\le H(A|E)$, so the preceding lower bounds and the exact von Neumann value force equality throughout, proving \eqref{eq:permutation-entropy}.
\end{proof}
\end{lemma}

For the following derivation, we write the normalized state conditioned on the local Schur outcomes $(\lambda,\mu)$. Let $p_{\lambda,\mu}>0$ denote the probability of this outcome, and define the unitary-sector state vectors
\be
\ket*{z_{\nu,\alpha}}_Q
\coloneqq
\sqrt{\frac{p_\lambda}{p_{\lambda,\mu}}}\,
(\id_A\otimes W_{\mu,\nu}^{\lambda,\alpha})
\ket{\phi_\lambda}_{A:BE}
=
\sqrt{q_{\nu,\alpha}}\ket*{U_{\nu,\alpha}}_Q,
\ee
where $q_{\nu,\alpha}\coloneqq\|z_{\nu,\alpha}\|^2$, and
$\ket*{U_{\nu,\alpha}}$ is normalized whenever $q_{\nu,\alpha}>0$; zero-weight terms are omitted. The conditional purification is therefore
\be
\ket*{\Psi_{\lambda,\mu}}
=
\sum_{\nu,\alpha}
\sqrt{q_{\nu,\alpha}}\,
\ket\nu_E\otimes
\ket*{B_{\mu,\nu}^{\lambda,\alpha}}_P
\otimes
\ket*{U_{\nu,\alpha}}_Q,
\qquad
\sum_{\nu,\alpha}q_{\nu,\alpha}=1.
\label{eq:normalized-schur-branch}
\ee
Indeed, the vectors
$\ket\nu_E\otimes\ket*{B_{\mu,\nu}^{\lambda,\alpha}}_P$
are orthonormal by \cref{lem:orthperm}, so normalization gives
$\sum_{\nu,\alpha}q_{\nu,\alpha}=1$. The state vectors
$\ket*{U_{\nu,\alpha}}$ need not be orthogonal. In particular,
$\Pr(\nu\mid\lambda,\mu)=\sum_\alpha q_{\nu,\alpha}$, and all dependence on the input state $\rho_{AB}$ within the fixed $(\lambda,\mu)$ sector is contained in the weights $q_{\nu,\alpha}$ and the unitary-sector state vectors $\ket*{U_{\nu,\alpha}}$.

\section{Universal entanglement distillation}\label{app:universal}
We now apply the random-subspace argument to the error model obtained from the local Schur transforms. As before, Alice chooses a code in her permutation register and Bob corrects the corresponding Clebsch--Gordan errors. We introduce a weighted reference channel and show how its recovery map also corrects the physical state. This gives the finite-copy bound in terms of a Schur--Weyl information-spectrum entropy.
Here one-way LOCC distillation converts $\rho_{AB}^{\otimes k}$ into an approximation of $\Phi_K$ using local operations and classical messages from Alice to Bob, with yield $\log_2 K/k$ ebits per copy. Fix an integer $r\ge\rank\rho_{AB}$. This is a supplied rank promise, not knowledge of the state or its support; with no sharper promise one may take $r=d_A d_B$. The Kronecker coefficient $g_{\lambda\mu\nu}$ is symmetric in its three labels, as follows from
$g_{\lambda\mu\nu}=(k!)^{-1}\sum_{\pi\in S_k}\chi_\lambda(\pi)\chi_\mu(\pi)\chi_\nu(\pi)$.

\subsection{From the physical channel to a weighted reference}
Fix local Schur labels $(\lambda,\mu)$ of positive probability. By \cref{thm:schurpurif}, and in particular \cref{eq:normalized-schur-branch}, the local Schur projections leave Alice and Bob with the post-measurement state vector
\be
\ket*{\Psi_{\lambda,\mu}}_{ABE}=\sum_{\nu,\alpha}(\id\otimes V_{\nu,\alpha}^\dagger)\ket{\Phi_d}_{\mathcal P_\lambda^A:\mathcal P_\lambda^{BE}}
\otimes\ket*{z_{\nu,\alpha}}_Q\otimes\ket\nu_E,
\qquad \sum_{\nu,\alpha}\|z_{\nu,\alpha}\|^2=1.
\label{eq:physical-branch-vector}
\ee
where we abbreviated $V_{\nu,\alpha}=V_{\mu,\nu}^{\lambda,\alpha}$. In the following, we will reduce the LOCC entanglement distillation problem to the setting of \cref{lem:normalized-polar}, with $d=d_\lambda$, $D=d_\mu$. Here, all state dependence lies in the unnormalized vectors $\ket*{z_{\nu,\alpha}}_Q$, which need not be orthogonal. 
It is helpful to expand the Clebsch--Gordan intertwiners as
\be
V_{\nu,\alpha}^\dagger=\sum_{e=1}^{d_\nu}E_{\nu,\alpha,e}\otimes\ket e,
\qquad E_{\nu,\alpha,e}\coloneqq(\id\otimes\bra e)V_{\nu,\alpha}^\dagger:
\mathcal P_\lambda^{BE}\longrightarrow\mathcal P_\mu^B.
\ee
The intertwiner identities in~\cref{lem:orth,lem:parttrort} then imply
\be
\sum_eE_{\nu,\beta,e}^\dagger E_{\nu,\alpha,e}
&=\delta_{\alpha\beta}\id_d,\qquad
\sum_eE_{\nu,\alpha,e}E_{\nu,\alpha,e}^\dagger=\frac dD\id_D,\\
\Tr \left( E_{\nu,\alpha,e}^\dagger E_{\eta,\beta,f} \right)
&=\frac d{d_\nu}\delta_{\nu\eta}\delta_{\alpha\beta}\delta_{ef}.
\label{eq:three-error-identities}
\ee
Tracing Eve's permutation register $\mathcal P_\nu$, the unitary registers $\mathcal{Q}$ and the Schur label register in \eqref{eq:physical-branch-vector} gives a channel description of the resulting permutation-register state:
\be
\rho_{\mathcal P_\lambda^A\mathcal P_\mu^B}
=(\idmap\otimes\mathcal N_{\lambda\to\mu})(\Phi_d), \qquad \mathcal N_{\lambda\to\mu}(X)
\coloneqq\sum_{\nu,e,\alpha,\beta}S_{\alpha\beta}^{(\nu)}
E_{\nu,\alpha,e}X E_{\nu,\beta,e}^\dagger,
\ee
where $\mathcal N_{\lambda\to\mu}:\mathcal L(\mathcal P_\lambda^{BE})\to\mathcal L(\mathcal P_\mu^B)$ and we define the Gram matrix
\be
S_{\alpha\beta}^{(\nu)}
&\coloneqq\langle z_{\nu,\beta}|z_{\nu,\alpha}\rangle
=\sqrt{q_{\nu,\alpha}q_{\nu,\beta}}\langle U_{\nu,\beta}|U_{\nu,\alpha}\rangle.
\label{eq:physical-schur-channel}
\ee
Here each $S^{(\nu)}$ is positive, $p_\nu\coloneqq \Tr S^{(\nu)}$ is the conditional probability $p_{\nu|\lambda,\mu}$, and $\sum_\nu p_\nu=1$. Thus $\mathcal N_{\lambda\to\mu}$ is completely positive. The first identity in \eqref{eq:three-error-identities} gives $\mathcal N_{\lambda\to\mu}^\dagger(\id_D)=\id_d$, proving trace preservation. 
Now choose any nonempty collection $\mathcal T$ of admissible environment sectors $\{\nu\}_{\nu \in \mathcal T}$, with $g_{\lambda\mu\nu}>0$, keeping all their multiplicity indices, and put
\be
Q_{\mathcal T}\coloneqq\sum_{\nu\in\mathcal T}g_{\lambda\mu\nu},\qquad
Q\coloneqq Q_{\lambda\mu}^{(r)}\coloneqq \sum_{\substack{\nu\vdash k\\\ell(\nu)\le r}}g_{\lambda\mu\nu},\qquad
N\coloneqq\sum_{\nu\in\mathcal T}g_{\lambda\mu\nu}d_\nu,
\qquad
\mathcal N_{\lambda\to\mu}^{(\nu,\alpha)}(X)\coloneqq\sum_eE_{\nu,\alpha,e}XE_{\nu,\alpha,e}^\dagger.
\ee
Each $\mathcal N_{\lambda\to\mu}^{(\nu,\alpha)}$ is a channel and maps $\pi_d$ to $\pi_D$. The uniform reference $\widehat{\mathcal N}_{\lambda\to\mu}^{\mathcal T,\mathrm{unif}}\coloneqq Q_{\mathcal T}^{-1}\sum_{\nu\in\mathcal T,\alpha}\mathcal N_{\lambda\to\mu}^{(\nu,\alpha)}$ would assign equal weight to these channels. We instead assign weights proportional to the number of error directions in each $\nu$. Setting $i \coloneqq (\nu,\alpha,e)$ we define
\be
\widehat{\mathcal N}_{\lambda\to\mu}^{\mathcal T}(X)
=\sum_{\nu\in\mathcal T,\alpha}\frac{d_\nu}{N}\mathcal N_{\lambda\to\mu}^{(\nu,\alpha)}(X)
=\frac1N\sum_{i=1}^NF_iXF_i^\dagger,
\qquad F_{\nu,\alpha,e}\coloneqq\sqrt{d_\nu}\,E_{\nu,\alpha,e}.
\label{eq:weighted-schur-channel}
\ee
This is the reference channel of \eqref{eq:reference-schur-channel}; its dependence on the code parameters enters through the set $\mathcal T$ chosen below. Here $N=N_{\lambda\mu}(\mathcal T)$. The weights sum to one. Equations~\eqref{eq:three-error-identities} show that the $F_i$ satisfy all the channel hypotheses in \cref{lem:normalized-polar}; in particular, $\widehat{\mathcal N}_{\lambda\to\mu}^{\mathcal T}(\pi_d)=\pi_D$.
To see the reason for these weights, let $\ket{f_i}\coloneqq(\id\otimes F_i)\ket{\Phi_d}$. These vectors are orthonormal, since $\langle f_i|f_j\rangle=d^{-1}\Tr F_i^\dagger F_j=\delta_{ij}$. The normalized Choi states of the channels $\widehat{\mathcal N}_{\lambda\to\mu}^{\mathcal T,\mathrm{unif}}$ and $\widehat{\mathcal N}_{\lambda\to\mu}^{\mathcal T}$ can then be written as
\be
(\idmap\otimes\widehat{\mathcal N}_{\lambda\to\mu}^{\mathcal T,\mathrm{unif}})(\Phi_d)
=\sum_{\nu\in\mathcal T,\alpha,e} \frac{1}{Q_{\mathcal T}d_\nu} \ketbra{f_{\nu\alpha e}}{f_{\nu\alpha e}},\quad
(\idmap\otimes\widehat{\mathcal N}_{\lambda\to\mu}^{\mathcal T})(\Phi_d)
=\frac{1}{N} \sum_{i=1}^N\ketbra{f_i}{f_i}.
\label{eq:flat-choi-reference}
\ee
Thus, the weighted reference Choi state is maximally mixed on an $N$-dimensional Schur error space.

\subsection{The Schur information spectrum and the universal theorem}
In this section, we formally describe the universal distillation protocol and prove the achievability result in terms of the Schur conditional information spectrum.
Since Alice and Bob do not know Eve's label $\nu$, we use one cutoff ordered by the permutation-sector dimension: to accept a sector of dimension $d_\nu$, we include all admissible sectors of dimension at most $d_\nu$, with all their multiplicities. The total number of error directions is
\be
\mathcal V_{\lambda\mu}^{(r)}(t)
\coloneqq \sum_{\substack{\eta\vdash k,\;\ell(\eta)\le r\\d_\eta\le t}}
g_{\lambda\mu\eta}d_\eta.
\label{eq:sw-volume-information}
\ee
Then, since Bob's register has size $D\equiv d_\mu$, we expect a random coding argument to work when the code-space size times the number of independent error directions (counting multiplicities) is smaller than $D$; hence
\be
K \ll \frac{d_\mu}{\mathcal V_{\lambda\mu}^{(r)}(d_\nu)}.
\ee
Therefore, the statistic
\be
J_k^{(r)}(\lambda,\mu,\nu) \coloneqq \log_2\frac{d_\mu}{\mathcal V_{\lambda\mu}^{(r)}(d_\nu)}
\ee
measures the logarithmic room left for the code after accounting for the cumulative error volume. It is then natural to define the Schur--Weyl conditional information-spectrum entropy as the largest amount of logarithmic coding room that is available with probability at least $1-\varepsilon$.
\be
H_{\mathrm{SW},k}^{\varepsilon,(r)}(A|E)_\psi
\coloneqq \max\left\{x\in\mathbb R:\Pr_{\psi^{\otimes k}}\bigl[J_k^{(r)}<x\bigl]\le\varepsilon\right\},
\qquad 0\le\varepsilon<1.
\label{eq:sw-spectrum-entropy}
\ee
The supremum is attained because the distribution has finite support.
Equivalently, $J_k^{(r)}$ may be viewed as a Schur--Weyl coarse graining of the conditional information density: heuristically, on typical branches $\log_2 d_\mu\simeq kH(B)_\rho$ while $\log_2 \mathcal V_{\lambda\mu}^{(r)}(d_\nu)\simeq kH(E)_\psi=kH(AB)_\rho$, so that $J_k^{(r)}\simeq k\bigl(H(B)_\rho-H(E)_\psi\bigr)=kH(A|E)_\psi=kI(A\rangle B)_\rho$.

For comparison, in the classical setting $(A,E)\sim p_{AE}$, the conditional information-spectrum formalism is built from the conditional surprisal
\be
\imath_{A|E}(a,e)\coloneqq -\log_2 p_{A|E}(a|e),
\ee
and its lower-tail quantiles~\cite{Han2003,DattaRenner2009}. Our construction has the same information-spectrum structure, but with the conditional surprisal replaced by the Schur--Weyl coding statistic $J_k^{(r)}(\lambda,\mu,\nu)$, which measures the branchwise logarithmic room available after accounting for the cumulative compatible error volume.

At finite block length this quantity records the error volume of the chosen universal decoder; its relation to the intrinsic information spectrum is given by \eqref{eq:advanced-sw-classical-coupling}. The theorem below is the weighted-channel extension of the random-subspace bound, with $L$ providing the coding slack.

\begin{theorem}[Universal entanglement distillation via Schur--Weyl coding]
\label{thm:universal-distillation}
For every $k$, integer $K\ge1$, supplied rank ceiling $r$, and $L\ge1$, there is a randomized one-way LOCC protocol depending only on $(d_A,d_B,k,K,r,L)$, whose seed-averaged Bell-state infidelity $\varepsilon_k\coloneqq 1-\langle\Phi_K|\omega_{A'B'}|\Phi_K\rangle$ on every input with $\rank\rho_{AB}\le r$ satisfies
\be
\varepsilon_k\le p_k+\frac{1-p_k}{L},\qquad
p_k\coloneqq \Pr[J_k^{(r)}<\log_2 K+\log_2 L].
\label{eq:sw-master}
\ee
In particular, $\log_2 K\le H_{\mathrm{SW},k}^{\varepsilon,(r)}(A|E)_\psi-\log_2 L$ implies $\varepsilon_k\le\varepsilon+(1-\varepsilon)/L$. For $0<\eta<\varepsilon<1$, the sufficient condition
\be
\log_2 K\le H_{\mathrm{SW},k}^{\varepsilon-\eta,(r)}(A|E)_\psi-\log_2(1/\eta)
\label{eq:entropy-achievable}
\ee
gives $\varepsilon_k\le\varepsilon$ by taking $L=1/\eta$.
Writing $\varepsilon_k(\rho_{AB};s)$ for the infidelity at fixed
classical seed $s$, each fixed input and $0<\zeta<1$ obey
\be
\Pr_s\!\left[
\varepsilon_k(\rho_{AB};s)>
\frac{p_k+(1-p_k)/L}{\zeta}
\right]\le\zeta.
\label{eq:distillation-seed-tail}
\ee
\end{theorem}

\begin{proof}
For $K=1$, the parties can output the one-dimensional product target while retaining their inputs, giving zero infidelity. We therefore assume $K\ge2$.
Alice and Bob first perform their local Schur measurements, obtaining $(\lambda,\mu)$, and retain their unitary registers. Alice sends $\lambda$ to Bob using classical communication.
In a fixed branch set $d=d_\lambda$, $D=d_\mu$ as in \cref{lem:normalized-polar} and define the correctable error region
\be
\mathcal T_{\lambda,\mu}^{(K)}\coloneqq\{\nu:\ell(\nu)\le r,\ g_{\lambda\mu\nu}>0,
\ KL\mathcal V_{\lambda\mu}^{(r)}(d_\nu)\le D\},
\qquad N=\sum_{\nu\in\mathcal T}g_{\lambda\mu\nu}d_\nu.
\label{eq:accepted-error-sectors}
\ee
Here $K$ is the correctable code size, while $L$ is an error-tolerance
parameter to be specified later. For simplicity, we write
$\mathcal T=\mathcal T_{\lambda,\mu}^{(K)}$ within the branch. The notation suppresses the dependence on $k$, the rank ceiling $r$, and the slack $L$; these parameters are fixed when the set is used.
If $\mathcal T=\varnothing$, Bob outputs a fixed state. This branch has $p_{\rm bad}=1$ and satisfies the claimed branch error bound trivially. The following estimates concern branches with $\mathcal T\ne\varnothing$.
Recall that the random-coding bound of \cref{lem:normalized-polar} is
controlled by $KN/D$. Let $\nu_\star\in\mathcal T$ have maximal
permutation dimension. Since the acceptance condition depends only on
$d_\nu$, $\mathcal T$ is an initial segment in $d_\nu$, including all
ties, and therefore
\begingroup
\setlength{\abovedisplayskip}{4pt}
\setlength{\belowdisplayskip}{4pt}
\be
N=\mathcal V_{\lambda\mu}^{(r)}(d_{\nu_\star}).
\ee
\endgroup
Applying \eqref{eq:accepted-error-sectors} to $\nu_\star$ gives
$KLN\le D$, hence $KN/D\le 1/L$. Thus the random-coding error on the
accepted region is bounded by $1/L$.
Moreover, for any accepted sector $\nu\in\mathcal T$, symmetry of the Kronecker coefficient implies that $\mathcal P_\mu$ occurs in $\mathcal P_\lambda\otimes\mathcal P_\nu$, and hence $D\le d\,d_\nu\le dN$. Together with $N\le D/(KL)$, this gives $d\ge KL\ge K$.
Consequently, if $d<K$, the correctable region is empty for every compatible $\mu$, and Alice declares failure. If $d\ge K$, Alice performs the code measurement below independently of $\mu$.
For $d\ge K$, Alice splits her Hilbert space as $d=mK+s$, $0\le s<K$, and chooses a random unitary rotation of a fixed orthogonal decomposition into $m$ complete $K$-dimensional blocks and one remainder. 
Let $C_t:\mathbb C^K\to\mathbb C^d$ be the isometry of block $t$. Alice measures with Kraus operators $C_t^T$, together with a map on the conjugate remainder subspace, and announces the outcome to Bob. 
The choice of unitary and blocks uses only $(\lambda,K)$. These operators form a complete instrument because their effects are the transposes of the orthogonal block projectors. 
Crucially, the transpose trick gives
\be
\bigl(C_t^T\otimes V_{\nu,\alpha}^\dagger\bigr)
\ket*{\Phi_{\mathcal P_\lambda}}_{\mathcal P_\lambda^A:\mathcal P_\lambda^{BE}}
=
\sqrt{\frac K{d_\lambda}}\,
\bigl(\id\otimes V_{\nu,\alpha}^\dagger C_t\bigr)
\ket{\Phi_K}_{A:BE}.
\label{eq:code-transpose-trick}
\ee
Thus every complete block occurs with probability $K/d$. Moreover, using
\eqref{eq:three-error-identities}, the conditional weight of a fixed
environment sector $\nu$ is
\be
\Tr\sum_{e,\alpha,\beta}S_{\alpha\beta}^{(\nu)}
(\id\otimes E_{\nu,\alpha,e}C_t)\Phi_K
(\id\otimes C_t^\dagger E_{\nu,\beta,e}^\dagger)
=
\Tr S^{(\nu)}
=
p_\nu.
\ee
Hence conditioning on a complete block leaves the sector probabilities
$p_\nu$ unchanged. Since the code measurement acts only on the permutation
factor, the matrices $S^{(\nu)}$ themselves stay independent of the random
code.
Conditioned on a complete block $t$, the accepted part of the physical code
state is therefore
\be
\rho_{\rm good}^{(t)}
&=
\sum_{\nu\in\mathcal T,e,\alpha,\beta}
S_{\alpha\beta}^{(\nu)}
(\id\otimes E_{\nu,\alpha,e}C_t)\Phi_K
(\id\otimes C_t^\dagger E_{\nu,\beta,e}^\dagger)
\\
&=
(\id\otimes W_{C_t})(\Phi_K\otimes T)
(\id\otimes W_{C_t}^\dagger),
\qquad
T=
\bigoplus_{\nu\in\mathcal T}
S^{(\nu)}\otimes\frac{\id_{d_\nu}}{d_\nu}.
\label{eq:physical-weighted-code-state}
\ee
At this point, we are exactly in the setup of \cref{lem:normalized-polar}. 
Now Bob associates to each complete block $C_t$ the polar decoder \eqref{eq:polar-recovery} of the weighted reference channel \eqref{eq:weighted-schur-channel}. The operator in \eqref{eq:physical-weighted-code-state} satisfies
\be
\Tr T
=
\sum_{\nu\in\mathcal T}\Tr S^{(\nu)}
=
\sum_{\nu\in\mathcal T}p_\nu
\eqqcolon p_{\rm good}.
\ee
For each fixed $t$, the code $C_t$ is marginally Haar distributed, while $T$ is independent of the random rotation. Hence \cref{lem:normalized-polar} gives
\be
\mathbb E_U\,\varepsilon_t^{\rm good}
\le
\beta\,\Tr T
=
p_{\rm good}\beta,
\ee
where $\varepsilon_t^{\rm good}$ is the failure weight of the accepted state \eqref{eq:physical-weighted-code-state}. Although the different blocks $C_t$ are correlated through the common Haar rotation, independence is not required: each complete block occurs with probability $K/d$, and linearity of expectation yields
\be
\mathbb E_U
\sum_{t=1}^{m}\frac Kd\,\varepsilon_t^{\rm good}
\le
\frac{mK}{d}\,p_{\rm good}\beta
=
(1-q)p_{\rm good}\beta,
\qquad
q\coloneqq \frac{s}{d}.
\ee
We now include the two remaining sources of failure: the residual subspace of dimension $s$, where Bob does not extract entanglement, and rejected sectors $\nu \notin \mathcal{T}$. The total error therefore satisfies
\be
\mathbb E\varepsilon_{\lambda\mu}
&\le q+(1-q)(p_{\rm bad}+p_{\rm good}\beta)\\
&=p_{\rm bad}+p_{\rm good}[q+(1-q)\beta]
\le p_{\rm bad}+p_{\rm good}\frac{KN}{D}
\le p_{\rm bad}+\frac{p_{\rm good}}L.
\label{eq:branch-weighted-error}
\ee
Indeed, $\beta\ge0$ and \eqref{eq:polar-beta} imply $q+(1-q)\beta\le q+\beta\le K/d+(KN/D-K/d)=KN/D$. 
Finally, averaging over $(\lambda,\mu)$ gives
\be
\varepsilon_k
\le
\sum_{\lambda,\mu}
p_{\lambda\mu}
\left(
p_{\rm bad}^{\lambda\mu}
+
\frac{p_{\rm good}^{\lambda\mu}}{L}
\right).
\ee
Since
\be
\nu\notin\mathcal T_{\lambda,\mu}^{(K)}
\quad\Longleftrightarrow\quad
J_k^{(r)}<\log_2 K+\log_2 L,
\ee
we have
\be
\sum_{\lambda,\mu}p_{\lambda\mu}p_{\rm bad}^{\lambda\mu}=p_k,
\qquad
\sum_{\lambda,\mu}p_{\lambda\mu}p_{\rm good}^{\lambda\mu}=1-p_k.
\ee
Therefore
\be
\varepsilon_k
\le
p_k+\frac{1-p_k}{L},
\ee
which proves \eqref{eq:sw-master}.
The construction is universal: all protocol choices depend only on $(d_A,d_B,k,K,r,L)$ and the observed Schur labels. Bob's decoder is fixed from the weighted reference channel, while \cref{lem:normalized-polar} guarantees the same recovery for every physical $T\ge0$. Since only second moments are used, Haar randomness may also be replaced by a finite weighted exact unitary $2$-design, an ensemble whose second moments match Haar measure. Thus the protocol is a single randomized one-way LOCC scheme for all inputs with $\rank\rho_{AB}\le r$.
The classical random choices can be sampled in advance for every
Alice sector $\lambda$, independently of the input, and the relevant
choice sent to Bob. Fixing this seed leaves a one-way LOCC protocol,
whose infidelity $\varepsilon_k(\rho_{AB};s)$ still averages over
measurement outcomes. Linearity gives
$\mathbb E_s\varepsilon_k(\rho_{AB};s)=\varepsilon_k$;
Markov's inequality and \eqref{eq:sw-master} prove
\eqref{eq:distillation-seed-tail}. This is a guarantee for each fixed
input; the successful seeds need not be common to all inputs.
\end{proof}

Equation~\eqref{eq:flat-choi-reference} shows why considering the weighted channel $\widehat{\mathcal N}_{\lambda\to\mu}^{\mathcal T}$ is useful. The unweighted reference is flat only over the channels $(\nu,\alpha)$, so comparing it to the physical error state incurs a domination factor $Q_{\mathcal T}$. By contrast, $\widehat{\mathcal N}_{\lambda\to\mu}^{\mathcal T}$ is flat over the $N$ elementary orthogonal error directions, and \cref{lem:normalized-polar} controls the physical operator $T$ directly, avoiding this multiplicity penalty.
The direct polar-recovery estimate also avoids the usual square-root loss incurred when a decoupling or second-moment bound is converted into a trace norm, thus yielding a linear loss of order $KN/D$ instead of $O(\sqrt{KN/D})$.
The cumulative volume also admits a simple spectral interpretation. For fixed $(\lambda,\mu)$, recall that $Q=Q_{\lambda\mu}^{(r)}$ counts all admissible multiplicities, whereas $Q_{\mathcal T}$ counts only the selected ones. Consider the reference state
\be
\tau_{\lambda\mu}
\coloneqq
\frac1Q
\bigoplus_{\nu:\ell(\nu)\le r}
\id_{g_{\lambda\mu\nu}}\otimes\pi_{d_\nu}.
\label{eq:universal-error-prior}
\ee
It is diagonal in the elementary error basis $(\nu,\alpha,e)$, with eigenvalue $(Qd_\nu)^{-1}$ on every direction in the $\nu$ sector. Hence the spectral cutoff at $(Qd_\nu)^{-1}$ selects precisely the error directions belonging to sectors with $d_\eta\le d_\nu$. Its rank is therefore the dimension of the corresponding error space,
\be
\mathcal V_{\lambda\mu}^{(r)}(d_\nu)
=
\rank\mathbf 1\!\left\{
\tau_{\lambda\mu}\ge\frac1{Qd_\nu}
\right\}.
\ee
Thus $\mathcal V_{\lambda\mu}^{(r)}(d_\nu)$ is exactly the number of orthogonal error directions retained by this spectral cutoff. The packing condition $K\mathcal V_{\lambda\mu}^{(r)}(d_\nu)\lesssim d_\mu$ then has the natural interpretation that Bob's permutation space must be large enough to accommodate a $K$-dimensional logical code for each selected error direction.

Alice's encoder is computationally efficient: it consists of a Schur transform, an approximate unitary $2$-design on the Specht register, and a measurement into fixed $K$-dimensional code blocks
\cite{Krovi2019efficienthigh,NakataHircheKoashiWinter2017,NakataMurao2013,CleveWatrous2000}. These steps have gate complexity $\operatorname{poly}(k,\log_2 d_A,\log_2(1/\varepsilon_{\mathrm{enc}}))$,
with an additional term at most $\varepsilon_{\mathrm{enc}}$ in the distillation error bound.\footnote{Let $\eta$ bound the difference between the design and Haar second-moment maps in induced Hilbert--Schmidt norm. Expanding \eqref{eq:gbound} and using the packing condition bounds the additional term in \eqref{eq:sw-master} by $(d_*^2+2d_*)\eta\le3k!\eta$, where $d_*\coloneqq\max_\lambda d_\lambda$. Thus $\eta\le\varepsilon_{\mathrm{enc}}/(6k!)$ leaves half the error budget for the Schur transform and gate synthesis. The required logarithmic precision is $O(k\log_2(k+1)+\log_2(1/\varepsilon_{\mathrm{enc}}))$, preserving polynomial gate complexity.} Efficient implementation of Bob's polar recovery remains a separate question.
The Haar randomization also admits a basis-independent
description in the permutation algebra. Let $R_\sigma^A$
denote the physical permutation of Alice's $k$ copies, and set
\be
G\coloneqq \frac{1}{\sqrt{k!}}
\sum_{\sigma\in S_k}g_\sigma R_\sigma^A,
\qquad
g_\sigma\overset{\mathrm{i.i.d.}}{\sim}
\mathcal N_{\mathbb C}(0,1),
\ee
where the coefficients are standard circular complex
Gaussians, normalized by $\mathbb E|g_\sigma|^2=1$.
Under Schur--Weyl duality,
\be
G=\bigoplus_\lambda
G_\lambda\otimes\id_{\mathcal Q_\lambda^A},
\qquad
G_\lambda\coloneqq
\frac{1}{\sqrt{k!}}
\sum_{\sigma\in S_k}g_\sigma p_\lambda(\sigma).
\ee
Schur orthogonality gives
\be
\mathbb E\!\left[
(G_\lambda)_{ab}
\overline{(G_\mu)_{cd}}
\right]
=\frac{\delta_{\lambda\mu}\delta_{ac}\delta_{bd}}
{d_\lambda}.
\ee
Since these entries are jointly circular Gaussian,
the matrices $\sqrt{d_\lambda}\,G_\lambda$ are independent
standard complex Ginibre matrices.
In particular, $G$ is invertible almost surely, and its
unitary polar factor is
\be
U\coloneqq G(G^\dagger G)^{-1/2}
=\bigoplus_\lambda
U_\lambda\otimes\id_{\mathcal Q_\lambda^A},
\qquad
U_\lambda\coloneqq
G_\lambda(G_\lambda^\dagger G_\lambda)^{-1/2}.
\ee
The unitary invariance of the Ginibre distribution implies
that the $U_\lambda$ are independent Haar unitaries on
$\mathcal P_\lambda^A$.
Thus the encoder's sectorwise Haar rotation can equivalently
be defined as the polar factor of a Gaussian element of
the physical permutation algebra. A Schur basis is still
used to specify the fixed rank-$K$ code blocks, but is not
needed to define their random orientation. This is an exact
distributional interpretation, rather than an efficient
implementation: sampling the displayed Gaussian sum and
realizing its polar factor are separate computational tasks.

\subsection{Achievability of the hashing bound} 
For probability vectors, write $H(p)\coloneqq-\sum_i p_i\log_2p_i$, $D(p\Vert q)\coloneqq\sum_i p_i\log_2(p_i/q_i)$, and $h_2(t)\coloneqq H(t,1-t)$, with $0\log_2 0=0$ and $D(p\Vert q)=+\infty$ if $p_i>0=q_i$ for some $i$. For $s\ge1$, set $a_s\coloneqq s(s-1)/2$ and $b_s\coloneqq(s-1)(s+2)/2$. We review some elementary bounds on the dimension of the symmetric group and unitary group irreps.

\begin{lemma}[Schur dimensions \cite{harrow2005applicationscoherentclassicalcommunication}]\label{lem:rank-dimensions}
For $\alpha\vdash k$ with $\ell(\alpha)\le s$, let $\bar\alpha=\alpha/k$. Then
\be
\frac{2^{kH(\bar\alpha)}}{(k+s)^{b_s}}
\le
d_\alpha
\le
2^{kH(\bar\alpha)},
\qquad
D_\alpha^{(s)}
\le
(k+1)^{a_s}.
\label{eq:rank-dimensions}
\ee
There are at most $(k+1)^{s-1}$ partitions $\alpha\vdash k$ with
$\ell(\alpha)\le s$. More generally, if $\ell(\alpha)\le a\le m$, then
\be
D_\alpha^{(m)}
\le
(k+1)^{am-a(a+1)/2}.
\ee
\end{lemma}

\begin{proof}
The hook and Weyl dimension formulas give \cite{weyl1946classical}
\be
d_\alpha
=
\frac{
k!\prod_{i<j}(\alpha_i-\alpha_j+j-i)
}{
\prod_{i=1}^s(\alpha_i+s-i)!
},
\qquad
D_\alpha^{(s)}
=
\prod_{1\leq i<j \leq s}
\left(
1+\frac{\alpha_i-\alpha_j}{j-i}
\right).
\ee
Consider the second equation. Every Weyl factor $1+(\alpha_i-\alpha_j)/(j-i)$ is at most $k+1$. Counting the pairs $i<j$ gives the upper bound
$D_\alpha^{(s)}\le(k+1)^{a_s}$ in \cref{eq:rank-dimensions}. More generally, if
$\ell(\alpha)\le a \le m$, only pairs with $i\le a$ contribute nontrivially,
and their number is
$am-a(a+1)/2$, which gives the last bound.
The hook formula for $d_\alpha$ further implies
\be
d_\alpha
\ge
\binom{k}{\alpha}(k+s)^{-a_s},
\ee
The upper bound $d_\alpha\le\binom{k}{\alpha}$ follows because a standard Young tableau is uniquely determined by the sets of entries in its rows.
For the lower type bound, set
$n_i=\sum_{j=i}^s\alpha_j$ and factor the multinomial coefficient as
\be
\binom{k}{\alpha}
=
\prod_{i=1}^{s-1}\binom{n_i}{\alpha_i}
\ge
\frac{
2^{\sum_i n_i h_2(\alpha_i/n_i)}
}{
\prod_i(n_i+1)
}
\ge
\frac{2^{kH(\bar\alpha)}}{(k+1)^{s-1}}.
\ee
Each binomial probability evaluated at
its empirical parameter is a mode and is therefore at least
$1/(n_i+1)$, while the entropy sum telescopes to $kH(\bar\alpha)$.
Combining this with
$\binom{k}{\alpha}\le2^{kH(\bar\alpha)}$ and
$a_s+s-1=b_s$ proves \cref{eq:rank-dimensions}.
\end{proof}

\begin{lemma}[Pinsker bounds for weak Schur sampling \cite{harrow2005applicationscoherentclassicalcommunication}]\label{lem:rank-pinsker}
Let $\rho_X$ be a quantum state of rank $r_X$ with nonzero spectrum $p$. Its weak Schur sampling outcome $\alpha$ satisfies $\ell(\alpha)\le r_X$, and
\be
 \Pr_{\rho_X}(\alpha)\le(k+1)^{a_{r_X}}2^{-kD(\bar\alpha\Vert p)},\qquad
 \Pr_{\rho_X}\!\left[\|\bar\alpha-p\|_1>\xi\right]
 \le(k+1)^{b_{r_X}}e^{-k\xi^2/2}.
\label{eq:rank-pinsker}
\ee
For $r_X\ge2$ and $0\le\xi\le1$, on the complementary event
$|H(\bar\alpha)-H(p)|\le f_{r_X}(\xi)$, where
$f_{r_X}(\xi)\coloneqq h_2(\xi/2)+(\xi/2)\log_2(r_X-1)$; for $r_X=1$ set $f_1=0$.
\end{lemma}
\begin{proof}
Since $\rho_X$ has rank $r_X$, the only partitions appearing have length $\ell(\alpha)\le r_X$.
Schur sampling has probability $d_\alpha s_\alpha(p)$, with $s_\alpha(p)$ being the corresponding Schur polynomial.
After ordering both $p_1\ge\cdots\ge p_{r_X}$ and $\alpha_1\ge\cdots\ge\alpha_{r_X}$, every weight $\beta$ of the highest-weight representation is majorized by $\alpha$. Hence $\prod_i p_i^{\beta_i}\le\prod_i p_i^{\alpha_i}$, so $s_\alpha(p)\le D_\alpha^{(r_X)}\prod_i p_i^{\alpha_i}$. Combining this with \cref{lem:rank-dimensions} gives the first inequality.
Summing over the at most $(k+1)^{r_X-1}$ partitions and using Pinsker's inequality $D(q\Vert p)\ge\|q-p\|_1^2/(2\ln2)$ gives the second bound. 
The final entropy bound then follows from the sharp Fannes inequality \cite{Audenaert2007}.
\end{proof}

\begin{lemma}[Multiplicity bounds]\label{lem:rank-Q}
For the supplied ceiling $r$, let $Q_{\lambda\mu}^{(r)}$ be the total admissible multiplicity introduced above.
On every physical $(\lambda,\mu)$ branch,
\be
Q_{\lambda\mu}^{(r)}
\le
(k+1)^\chi,
\qquad
\chi
\coloneqq
r_A r_Br-\frac{r_A(r_A+1)}2,
\qquad
\Gamma_k-\chi\log_2(k+1)
\le
J_k^{(r)}
\le
\Gamma_k,
\label{eq:rank-Q}
\ee
where
\be
\Gamma_k
\coloneqq
\log_2 d_\mu-\log_2 d_\nu.
\ee
\end{lemma}
\begin{proof}
The branching rule of \cref{lem:unit}, applied to the auxiliary spaces
$\mathbb C^{r_B}$ and $\mathbb C^r$, gives
\be
D_\lambda^{(r_Br)}
=
\sum_{\mu',\eta}
g_{\lambda\mu'\eta}
D_{\mu'}^{(r_B)}
D_\eta^{(r)}
\ge
D_\mu^{(r_B)}
Q_{\lambda\mu}^{(r)}
\ge
Q_{\lambda\mu}^{(r)}.
\ee
A physical branch satisfies
$\ell(\lambda)\le r_A\le r_Br_E\le r_Br$.
Hence the refined Weyl bound of \cref{lem:rank-dimensions}, applied with
$a=r_A$ and $m=r_Br$, gives
\be
D_\lambda^{(r_Br)}
\le
(k+1)^{r_A r_Br-r_A(r_A+1)/2}
=
(k+1)^\chi,
\ee
which proves the first claim.
Finally, since a physical triple satisfies $g_{\lambda\mu\nu}\ge1$,
\be
d_\nu
\le
\mathcal V_{\lambda\mu}^{(r)}(d_\nu)
\le
d_\nu Q_{\lambda\mu}^{(r)}.
\ee
Taking logarithms and using
$J_k^{(r)}
=
\log_2 d_\mu-\log_2\mathcal V_{\lambda\mu}^{(r)}(d_\nu)$
gives
\be
\Gamma_k-\log_2 Q_{\lambda\mu}^{(r)}
\le
J_k^{(r)}
\le
\Gamma_k,
\ee
and therefore \eqref{eq:rank-Q}.
\end{proof}

\begin{proposition}[Coherent-information rate]\label{prop:first-order}
Fix a rank ceiling $r$. For every $\rho_{AB}$ with
$\rank\rho_{AB}\le r$, every $\delta>0$, and every rate
\be
0\le R\le I(A\rangle B)_\rho-\delta,
\ee
the universal protocol with yield $K_k=2^{\lfloor kR\rfloor}$ for sufficiently large $k$ has exponentially vanishing infidelity.
\end{proposition}
\begin{proof}
Let $r_B=\rank\rho_B$ and $r_E=\rank\rho_{AB}$, and choose
$0<\xi\le1$ such that
$f_{r_B}(\xi)+f_{r_E}(\xi)\le\delta/4$ and set $L=2^{\delta k/2}$.
On the event
\be
\|\bar\mu-p_B\|_1\le\xi,
\qquad
\|\bar\nu-p_E\|_1\le\xi,
\ee
\cref{lem:rank-dimensions,lem:rank-pinsker,lem:rank-Q} give
\be
\begin{aligned}
J_k^{(r)}
&\ge
k\bigl[H(\bar\mu)-H(\bar\nu)\bigr]
-b_{r_B}\log_2(k+r_B)-\chi\log_2(k+1)
\\
&\ge
kI(A\rangle B)_\rho-\frac{\delta k}{2}
\end{aligned}
\ee
for all sufficiently large $k$, since the logarithmic terms are $o(k)$.
Moreover,
\be
\log_2 K_k+\log_2 L
\le
kR+\frac{\delta k}{2}
\le
kI(A\rangle B)_\rho-\frac{\delta k}{2}.
\ee
Hence the event $J_k^{(r)}<\log_2 K_k+\log_2 L$ can occur only if one of the
two Schur spectra deviates by more than $\xi$. By \cref{lem:rank-pinsker}
and the union bound,
\be
p_k
\le
\left[
(k+1)^{b_{r_B}}+(k+1)^{b_{r_E}}
\right]e^{-k\xi^2/2}
\le
e^{-c_1k}
\ee
for all sufficiently large $k$ and some $c_1>0$. Applying
\cref{thm:universal-distillation},
\be
\varepsilon_k
\le
p_k+\frac{1-p_k}{L}
\le
e^{-c_1k}+2^{-\delta k/2}
\le
e^{-ck}
\ee
for some $c>0$ and all sufficiently large $k$.
For each fixed input, Markov's inequality also gives
$\Pr_s[\varepsilon_k(\rho_{AB};s)>e^{-ck/2}]\le e^{-ck/2}$.
\end{proof}

\subsection{Second-order achievability}
Setting $I\coloneqq H(B)_\rho-H(AB)_\rho$, the coherent information variance is
\be
 V=V(A\rangle B)_\rho\coloneqq \Tr\rho_{AB}\bigl(\log_2\rho_{AB}-\log_2(\id_A\otimes\rho_B)-I\id_{AB}\bigr)^2.
\label{eq:information-variance}
\ee
Logarithms of $\rho_{AB}$ act on its support and are extended by zero on its kernel; $\log_2\rho_B$ is restricted to its support. Let $\Phi$ denote the standard normal distribution function.

\begin{lemma}[Schur and classical information spectra]
\label{lem:schur-classical}
Let $\psi_{ABE}$ purify $\rho_{AB}$, and let $(\mu,\nu)$ be the
joint weak Schur outcomes on $B^kE^k$. There is a coupling such that 
\begin{align}
\Gamma_k
   =\sum_{t=1}^k Z_t+\Delta_{E,k}-\Delta_{B,k},\qquad
\Delta_{X,k}\geq0,\quad
\mathbb E2^{\Delta_{X,k}}\leq(k+1)^{b_{r_X}},
\label{eq:sw-classical-coupling}
\end{align}
where $Z_t$ are independent copies of a finite random variable $Z$ with
$\mathbb EZ=I(A\rangle B)_\rho$ and
$\operatorname{Var}Z=V(A\rangle B)_\rho$. Explicitly, writing
$\rho_{AB}=\sum_j e_j\ketbra*{\varphi_j}{\varphi_j}$ and
$\rho_B=\sum_i b_i\ketbra{b_i}{b_i}$ on their supports, its law is
\begin{align}
P(i,j)\coloneqq e_j\bra*{\varphi_j}
  (\id_A\otimes\ketbra{b_i}{b_i})\ket*{\varphi_j},\qquad
Z(i,j)\coloneqq\log_2 e_j-\log_2 b_i.
\label{eq:sw-classical-law}
\end{align}
Consequently, for $\widehat Z_k\coloneqq \sum_t Z_t$ and all $t\in\mathbb R$ and $u_B,u_E>0$,
\begin{align}
\Pr[\widehat Z_k<t-u_E]-(k+1)^{b_{r_E}}2^{-u_E}
\leq\Pr[\Gamma_k<t]
\leq\Pr[\widehat Z_k<t+u_B]+(k+1)^{b_{r_B}}2^{-u_B}.
\label{eq:sw-classical-tails}
\end{align}
\end{lemma}

\begin{proof}
All operators below are restricted to the marginal supports. On
$(\supp \rho_X)^{\otimes k}$, Schur--Weyl duality gives
$\rho_X^{\otimes k}=\bigoplus_\alpha \id_{\mathcal P_\alpha}\otimes q_\alpha(\rho_X)$.
Since $d_\alpha\Tr q_\alpha(\rho_X)\leq1$, the
operator
\begin{align}
\widehat{\Delta}_{X,k}\coloneqq -\log_2 \rho_X^{\otimes k}
  -\sum_\alpha\log_2 d_\alpha\,\Pi_\alpha^X
\end{align}
is positive and commutes with $\rho_X^{\otimes k}$. Moreover, using \cref{lem:rank-dimensions},
\begin{align}
\Tr \rho_X^{\otimes k}2^{\widehat{\Delta}_{X,k}}
 =\sum_{\alpha:\ell(\alpha)\leq r_X}D_{\alpha}^{(r_X)}
 \leq(k+1)^{r_X-1+r_X(r_X-1)/2}.
\end{align}
The two Schur observables, the two $\widehat{\Delta}_{X,k}$, and
$\log_2 \rho_E^{\otimes k}-\log_2 \rho_B^{\otimes k}$ commute. Together with $\Pi_\lambda^A$, their joint
measurement on $\psi_{ABE}^{\otimes k}$ therefore defines a classical
probability distribution. Let $\Delta_{X,k}$ denote the outcome
of measuring $\widehat{\Delta}_{X,k}$. The corresponding operator identity
then gives \cref{eq:sw-classical-coupling} for these classical outcomes.
In the Schmidt basis
$\ket*\psi=\sum_j\sqrt{e_j}\ket*{\varphi_j}_{AB}\ket j_E$, the last observable
has the distribution of an independent sum with law \cref{eq:sw-classical-law}. Furthermore, the first two moments of the classical random variable $Z$ are
\begin{align}
\mathbb EZ=\Tr \rho_{AB}
 (\log_2 \rho_{AB}-\log_2 (\id_A\otimes\rho_B))=I,\qquad
\mathbb EZ^2=\Tr \rho_{AB}
 (\log_2 \rho_{AB}-\log_2 (\id_A\otimes\rho_B))^2.
\end{align}
Finally, Markov's inequality applied to $2^{\Delta_{X,k}}$ and
$\widehat Z_k-\Delta_{B,k}\leq\Gamma_k\leq \widehat Z_k+\Delta_{E,k}$ give \cref{eq:sw-classical-tails}.
\end{proof}
The coupling shows that the Schur coding statistic is governed by
the classical information sum $\widehat Z_k=\sum_{t=1}^k Z_t$.
Indeed, replacing the sampled environment dimension $d_\nu$ by
the cumulative error volume costs at most
$\log_2 Q_{\lambda\mu}^{(r)}\le\chi\log_2(k+1)$ bits. Hence
\begin{align}
\widehat Z_k-\Delta_{B,k}-\chi\log_2(k+1)
\le J_k^{(r)}
\le \widehat Z_k+\Delta_{E,k}.
\label{eq:advanced-sw-classical-coupling}
\end{align}
At fixed ranks, the defect bounds make these corrections
$O(\log_2 k)$ with high probability. Thus the classical sum $\widehat Z_k$
determines both the leading rate and the Gaussian fluctuations
of the Schur coding statistic.
Although $\widehat Z_k$ is a classical random variable, it also describes
the information spectrum associated with quantum hypothesis testing
between $\rho_{AB}^{\otimes k}$ and
$\id_{A^k}\otimes\rho_B^{\otimes k}$.
More precisely, its lower quantiles bound the corresponding
quantum hypothesis-testing divergence, up to error-parameter
shifts and logarithmic corrections
\cite[Theorem~14]{tomamichel2013hierarchy}.
The coupling above connects this testing interpretation to
the Schur coding statistic $J_k^{(r)}$.
The correction terms account for information that is present in
$\widehat Z_k$ but unavailable to the Schur code. For example, one copy of
a maximally entangled state has $\widehat Z_1=\log_2 d$, yet its permutation
registers are one-dimensional and offer no room for a nontrivial
code: $J_1^{(r)}=0$. Here $\Delta_{B,1}=\log_2 d$ and $\Delta_{E,1}=0$,
so the defect exactly accounts for this difference.

\begin{proposition}[Second-order achievability]
\label{prop:second-order}
Suppose $I>0$ and fix a supplied rank ceiling $r\ge r_E$. For $V>0$ and every fixed $0<\varepsilon<1$, a member of the universal family achieves
entanglement infidelity at most $\varepsilon$ and
\begin{align}
\log_2 K_k\geq kI+\sqrt{kV}\,\Phi^{-1}(\varepsilon)
-O(\log_2 k).
\label{eq:sw-second-order-rate}
\end{align}
The expansion specifies achievable target sizes; neither $I$ nor $V$ is used by the protocol once $K_k$ and its coding margin are supplied.
\end{proposition}

\begin{proof}
By \cref{lem:rank-Q}, we have
$Q_{\lambda\mu}^{(r)}\leq(k+1)^\chi$. Choose
$L_k=(k+1)^2$. The coding term $1/L_k$ in
\cref{eq:sw-master} is then $(k+1)^{-2}$.
By \cref{lem:schur-classical}, the event
$\Delta_{B,k}>(b_{r_B}+2)\log_2(k+1)$ has probability at most $(k+1)^{-2}$.
Set $z=\Phi^{-1}(\varepsilon)$ and, for any
$c>\chi+b_{r_B}+4$, choose the code size
\begin{align}
K_k
=
2^{\left\lfloor
kI+\sqrt{kV}\,z-c\log_2(k+1)
\right\rfloor},
\end{align}
which is at least one for all sufficiently large $k$, since $I>0$.
With $b=c-(\chi+b_{r_B}+4)>0$, the master bound and
\cref{eq:advanced-sw-classical-coupling} give
\begin{align}
\varepsilon_k
&\leq
\Pr\!\left[
\widehat Z_k<kI+\sqrt{kV}\,z-b\log_2(k+1)
\right]
+2(k+1)^{-2}.
\label{eq:second-order-classical-reduction}
\end{align}
Indeed, the two terms $(k+1)^{-2}$ come respectively from
$\Pr[\Delta_{B,k}>(b_{r_B}+2)\log_2(k+1)]$ and from the coding term
$L_k^{-1}$.
To estimate the remaining probability, recall that
$\widehat Z_k=\sum_{t=1}^k Z_t$, where the $Z_t$ are i.i.d. with
$\mathbb EZ=I$ and $\Var Z=V>0$. Writing
$\rho_3\coloneqq \mathbb E|Z-I|^3<\infty$, the Berry--Esseen theorem gives,
for a universal constant $C_{\rm BE}$,
\begin{align}
\sup_{x\in\mathbb R}
\left|
\Pr\!\left[
\frac{\widehat Z_k-kI}{\sqrt{kV}}\leq x
\right]
-\Phi(x)
\right|
\leq
\frac{C_{\rm BE}\rho_3}{V^{3/2}\sqrt{k}} .
\label{eq:berry-esseen}
\end{align}
The third absolute moment is finite because $Z$ has finite support. Taking
\begin{align}
\delta_k
\coloneqq
\frac{b\log_2(k+1)}{\sqrt{kV}},
\end{align}
we obtain from \eqref{eq:second-order-classical-reduction}
\begin{align}
\varepsilon_k
&\leq
\Phi(z-\delta_k)
+\frac{C_{\rm BE}\rho_3}{V^{3/2}\sqrt{k}}
+2(k+1)^{-2}.
\label{eq:second-order-BE}
\end{align}
Now $z=\Phi^{-1}(\varepsilon)$, so $\Phi(z)=\varepsilon$, while
$\delta_k\to0$. Let $\phi=\Phi'$ be the standard normal density. By the
mean-value theorem, for some $\xi_k\in[z-\delta_k,z]$,
\begin{align}
\Phi(z-\delta_k)
&=
\Phi(z)-\phi(\xi_k)\delta_k
\nonumber\\
&=
\varepsilon-\phi(\xi_k)
\frac{b\log_2(k+1)}{\sqrt{kV}} .
\end{align}
Since $\phi(z)>0$ and $\xi_k\to z$, for all sufficiently large $k$,
$\phi(\xi_k)\geq\phi(z)/2$. Hence
\begin{align}
\varepsilon_k
\leq
\varepsilon
-\frac{b\phi(z)}{2\sqrt V}
\frac{\log_2(k+1)}{\sqrt{k}}
+\frac{C_{\rm BE}\rho_3}{V^{3/2}\sqrt{k}}
+2(k+1)^{-2}.
\end{align}
The negative term is of order $\log_2 k/\sqrt{k}$, whereas the
Berry--Esseen remainder is $O(k^{-1/2})$ and the remaining term is
$O(k^{-2})$. Therefore the negative term dominates both remainders for
all sufficiently large $k$, and consequently $\varepsilon_k\leq\varepsilon$.
\end{proof}

\subsection{Perfect local recoverability}
\begin{proposition}[Perfect local recoverability]
\label{prop:perfect-local-recoverability}
Retain each party's Schur label and unitary register throughout the
protocol of \cref{thm:universal-distillation}. There are state-independent local channels $\mathcal S_A$ and
$\mathcal S_B$ such that, for the complete local outputs,
$\mathcal S_X(\omega_X^{\rm out})=\rho_X^{\otimes k}$, $X=A,B$.
\end{proposition}
\begin{proof}
For either party the Schur decomposition is
\be
U_{\rm Schur}^X\rho_X^{\otimes k}(U_{\rm Schur}^X)^\dagger
=\bigoplus_\xi p_\xi^X\pi_{\mathcal P_\xi}
\otimes\rho_{\mathcal Q_\xi}^X.
\ee
All subsequent local operations act only on permutation registers and
classical copies of the labels. Trace preservation and no-signalling
therefore preserve the unconditional marginal
$\bigoplus_\xi p_\xi^X\rho_{\mathcal Q_\xi}^X$ of the retained label
and unitary register, including after averaging all communicated
outcomes. The reconstruction channel discards the other local outputs,
adjoins $\pi_{\mathcal P_\xi}$ controlled by $\xi$, and applies the
inverse Schur transform. This gives the displayed input marginal
exactly. Only the retained Schur label and unitary register are needed; the permutation outputs and decoding environments can be discarded. This reconstructs the local states; it makes no assertion about
recovering their original joint correlations.
\end{proof}

\section{Universal quantum communication}\label{app:communication}
A similar construction to the one derived in \cref{thm:universal-distillation} transmits quantum information through an unknown memoryless channel, whose $k$ independent uses act as $\mathcal N^{\otimes k}$. 
We first exhibit a random ensemble whose law and code pairs are independent of that channel. A shared classical seed sampled before transmission works simultaneously for all channels satisfying $R\le I_c-\delta$ with high probability. Fixing a successful seed then yields a deterministic universal code.
Here the input $\xi$ is a state chosen in advance as part of the code design; it is not an unknown state to be estimated. 
For a channel $\mathcal N:\mathcal L(A')\to\mathcal L(B)$, fix $\xi\in\mathcal D(A')$ and a unit-vector purification $\ket{\phi_\xi}_{RA'}\in R\otimes A'$, with $R\simeq A'$, and let
$V_{\mathcal N}:A'\to BE$ be a Stinespring isometry. Define
\begin{align}
\ket{\psi_{\mathcal N}}_{RBE}
\coloneqq
(\id_R\otimes V_{\mathcal N})\ket{\phi_\xi}_{RA'}.
\end{align}
Write $\phi_\xi=\ketbra{\phi_\xi}{\phi_\xi}$, $\psi_{\mathcal N}=\ketbra{\psi_{\mathcal N}}{\psi_{\mathcal N}}\in\mathcal D(RBE)$, and $\omega_{\mathcal N}\coloneqq \Tr_E\psi_{\mathcal N}=(\idmap_R\otimes\mathcal N)(\phi_\xi)\in\mathcal D(RB)$. Then
\begin{align}
I_c(\xi,\mathcal N)
=
H(B)_{\psi_{\mathcal N}}
-
H(E)_{\psi_{\mathcal N}},
\end{align}
using purity of $\psi_{\mathcal N}$.
We use entanglement transmission as the communication criterion. For a
$K$-dimensional message space and an encoder--decoder pair
$(\mathcal E,\mathcal D)$, let
\begin{align}
\Lambda
\coloneqq
\mathcal D\circ\mathcal N^{\otimes k}\circ\mathcal E.
\end{align}
The rate is $\log_2 K/k$ qubits per use and the error is $1-F_e(\pi_K,\Lambda)$, with entanglement fidelity defined in \eqref{eq:entanglement-fidelity}. The limit $F_e(\pi_K,\Lambda)\to1$ means that a $K$-dimensional quantum system is transmitted while preserving its entanglement with an inaccessible reference. This criterion is asymptotically equivalent to the usual
quantum-state transmission criterion.
The advantage of this formulation is that
$(\idmap\otimes\Lambda)(\Phi_K)$ is the normalized Choi state of the
effective channel. Hence the maximally entangled-state recovery bound from
the distillation construction in \cref{thm:universal-distillation} can be reused directly, once the
source-side code is realized as a physical encoder on $A'$.
The ensemble bound below controls the seed-averaged error for each
channel. A diamond-norm net will turn it into simultaneous
high-probability success over the promised channel family.

\begin{lemma}[A channel-independent code ensemble]\label{lem:universal-code-ensemble}
Fix $\xi$, an integer $r\ge1$, $L\ge1$, and $1\le K\le d_{A'}^k$. There is a finite distribution of pairs $(\mathcal E_s,\mathcal D_s)$, with isometric encoders into $(A')^{\otimes k}$ and CPTP decoders from $B^{\otimes k}$, depending only on $(\xi,d_{A'},d_B,k,K,r,L)$, such that every channel with $\rank\omega_{\mathcal N}\le r$ obeys
\be
\mathbb E_s\varepsilon_k(\mathcal N;s)
\le p_k(\mathcal N)+\frac{1-p_k(\mathcal N)}{L},
\qquad p_k(\mathcal N)=\Pr_{\psi_{\mathcal N}}[J_k^{(r)}<\log_2 K+\log_2 L],
\label{eq:sw-transmission-bound}
\ee
where $\varepsilon_k(\mathcal N;s)\coloneqq 1-F_e(\pi_K,\mathcal D_s\circ\mathcal N^{\otimes k}\circ\mathcal E_s)$. Both the distribution and every code pair are independent of the unknown channel. In particular, $r=d_{A'}d_B$ removes the need for any rank promise.
\end{lemma}
\begin{proof}
Fix an
admissible channel $\mathcal N$ for the analysis, and abbreviate $p_k=p_k(\mathcal N)$.
The input state $\xi$ is fixed as part of the code design.  Its $k$-fold
tensor power has the Schur decomposition
\begin{align}
U_{\rm Schur}^{A'}\,
\xi^{\otimes k}\,
(U_{\rm Schur}^{A'})^\dagger
=
\sum_\lambda
p_\lambda\,
\ketbra{\lambda}{\lambda}
\otimes
\pi_{\mathcal P_\lambda}
\otimes
\xi_{\mathcal Q_\lambda},
\label{eq:communication-mixed-schur}
\end{align}
where $\xi_{\mathcal Q_\lambda}=q_\lambda(\xi)/\Tr q_\lambda(\xi)\in\mathcal D(\mathcal Q_\lambda^{A'})$ for $p_\lambda>0$; zero-probability sectors are omitted.
Thus, conditioned on the classical block label $\lambda$, the unitary
register is in the fixed known state $\xi_{\mathcal Q_\lambda}$ and the
permutation register is maximally mixed.  For communication we replace
this maximally mixed permutation register by the encoded logical input.
Given an isometry $C:\mathbb C^K\longrightarrow\mathcal P_\lambda$, Alice encodes $X\in\mathcal D(\mathbb C^K)$ by preparing, in Schur coordinates,
\begin{align}
\ketbra{\lambda}{\lambda}
\otimes
CXC^\dagger
\otimes
\xi_{\mathcal Q_\lambda}.
\label{eq:communication-schur-input}
\end{align}
Equivalently, in the physical input space, the encoding channel is
\begin{align}
\mathcal E_{\lambda,C}(X)
\coloneqq
(U_{\rm Schur}^{A'})^\dagger
\Bigl(
\ketbra{\lambda}{\lambda}
\otimes
CXC^\dagger
\otimes
\xi_{\mathcal Q_\lambda}
\Bigr)
U_{\rm Schur}^{A'} .
\label{eq:source-schur-encoder}
\end{align}
Let $A_0\simeq R_0\simeq\mathbb C^K$, where $R_0$ is the inaccessible
reference used to test entanglement transmission, and prepare
$\Phi_K$ on $R_0A_0$.  In a fixed $\lambda$ block, after applying the
code $C$ the state in the $\lambda$ sector is
\begin{align}
(\id_{R_0}\otimes C)\Phi_K
(\id_{R_0}\otimes C^\dagger)
\otimes
\xi_{\mathcal Q_\lambda}.
\label{eq:encoded-bell-schur-input}
\end{align}

We now introduce a purification only as an analytical device.  Let
$\ket{\varphi_\lambda}_{R_\lambda\mathcal Q_\lambda^{A'}}$ be any
purification of $\xi_{\mathcal Q_\lambda}$, where $R_\lambda$ is
unrelated to the communication reference $R_0$.  A purification of
\eqref{eq:encoded-bell-schur-input} is then
\begin{align}
(\id_{R_0}\otimes C)\ket{\Phi_K}_{R_0A_0}
\otimes
\ket{\varphi_\lambda}_{R_\lambda\mathcal Q_\lambda^{A'}}.
\label{eq:encoded-bell-purification}
\end{align}
Let $V_{\mathcal N}:A'\longrightarrow BE$ be a Stinespring isometry for $\mathcal N$.  Its tensor power
intertwines the permutation actions,
\begin{align}
V_{\mathcal N}^{\otimes k}R_\pi^{A'}
=
R_\pi^{BE}V_{\mathcal N}^{\otimes k},
\qquad
\pi\in S_k.
\end{align}
Since our Schur convention orders each block as
$\mathcal P_\lambda\otimes\mathcal Q_\lambda$, Schur's lemma therefore
gives
\begin{align}
U_{\rm Schur}^{BE}
V_{\mathcal N}^{\otimes k}
(U_{\rm Schur}^{A'})^\dagger
=
\bigoplus_\lambda
\id_{\mathcal P_\lambda}
\otimes
W_{\mathcal N,\lambda}
\label{eq:stinespring-schur-intertwiner}
\end{align}
for isometries $W_{\mathcal N,\lambda}:\mathcal Q_\lambda^{A'}\to\mathcal Q_\lambda^{BE}$.  Hence, in the global $\lambda$ block, the state
\eqref{eq:encoded-bell-purification} is mapped to
\begin{align}
(\id_{R_0}\otimes C)\ket{\Phi_K}
\otimes
(\id_{R_\lambda}\otimes W_{\mathcal N,\lambda})
\ket{\varphi_\lambda}.
\label{eq:encoded-bell-after-stinespring}
\end{align}
The physical channel therefore acts only on the unitary-register factor
in global Schur coordinates, while the logical information remains in
the permutation register. The bipartition is now between the retained reference systems and $BE$.
We next resolve the global $BE$ Schur block into the local Schur blocks
of $B$ and $E$.  Applying the same local--global Clebsch--Gordan
decomposition as in \cref{thm:schurpurif}, and before normalizing the output label $\mu$, this gives a completely positive map $\mathcal M_{\lambda\to\mu}:\mathcal L(\mathcal P_\lambda^{A'})\to\mathcal L(\mathcal P_\mu^B)$, identifying the input irrep with $\mathcal P_\lambda^{BE}$ through \eqref{eq:stinespring-schur-intertwiner},
\begin{align}
\mathcal M_{\lambda\to\mu}(X)
=
\sum_{\nu,e,\alpha,\beta}
\bigl[
\widetilde S_{\mathcal N,\lambda\mu}^{(\nu)}
\bigr]_{\alpha\beta}
E_{\nu,\alpha,e}
X
E_{\nu,\beta,e}^\dagger ,
\label{eq:communication-unnormalized-schur-channel}
\end{align}
with
$\widetilde S_{\mathcal N,\lambda\mu}^{(\nu)}\ge0$.
The map $\mathcal M_{\lambda\to\mu}$ is completely positive and
trace-nonincreasing; for $X\in\mathcal D(\mathcal P_\lambda^{A'})$, $\Tr\mathcal M_{\lambda\to\mu}(X)$ is the probability of Bob's Schur outcome $\mu$.
We now show that this probability is independent of the state inserted
into $\mathcal P_\lambda$.  For fixed $\lambda$, let
\begin{align}
\mathcal E_\lambda(X)
\coloneqq
(U_{\rm Schur}^{A'})^\dagger
\Bigl(
\ketbra{\lambda}{\lambda}
\otimes
X
\otimes
\xi_{\mathcal Q_\lambda}
\Bigr)
U_{\rm Schur}^{A'}.
\end{align}
Then $\mathcal M_{\lambda\to\mu}$ is equivalently obtained by applying
$\mathcal N^{\otimes k}\circ\mathcal E_\lambda$, performing Bob's
Schur transform, projecting onto $\mu$, and tracing over
$\mathcal Q_\mu$.
Permutation covariance gives
\begin{align}
\mathcal M_{\lambda\to\mu}
\bigl(
p_\lambda(\pi)Xp_\lambda(\pi)^\dagger
\bigr)
=
p_\mu(\pi)
\mathcal M_{\lambda\to\mu}(X)
p_\mu(\pi)^\dagger
\end{align}
for every $\pi\in S_k$.  Hence
$\mathcal M_{\lambda\to\mu}^\dagger(\id)$ commutes with the irreducible
representation $p_\lambda(\cdot)$, and Schur's lemma implies
\begin{align}
\mathcal M_{\lambda\to\mu}^\dagger(\id)
=
q_{\mu|\lambda}\id_{\mathcal P_\lambda},
\qquad
\sum_\mu q_{\mu|\lambda}=1.
\label{eq:schur-output-label-flatness}
\end{align}
Therefore every normalized state $X$ on $\mathcal P_\lambda$ satisfies
\begin{align}
\Pr[\mu\mid\lambda,X]
=
\Tr\mathcal M_{\lambda\to\mu}(X)
=
q_{\mu|\lambda}.
\label{eq:communication-mu-flatness}
\end{align}
In particular, this probability is independent of both the transmitted
state and the code $C$.
For $q_{\mu|\lambda}>0$ define
\begin{align}
\mathcal N_{\lambda\to\mu}
\coloneqq
\frac{\mathcal M_{\lambda\to\mu}}{q_{\mu|\lambda}}.
\end{align}
Equation \eqref{eq:schur-output-label-flatness} implies that
$\mathcal N_{\lambda\to\mu}$ is CPTP. As in the distillation setting, this denotes the physical conditional channel; its dependence on the underlying channel $\mathcal N$ and fixed input $\xi$ is suppressed. Defining
\begin{align}
S_{\mathcal N,\lambda\mu}^{(\nu)}
\coloneqq
\frac{
\widetilde S_{\mathcal N,\lambda\mu}^{(\nu)}
}{
q_{\mu|\lambda}
},
\end{align}
we obtain the normalized conditional Schur channel
\begin{align}
\mathcal N_{\lambda\to\mu}(X)
=
\sum_{\nu,e,\alpha,\beta}
\bigl[
S_{\mathcal N,\lambda\mu}^{(\nu)}
\bigr]_{\alpha\beta}
E_{\nu,\alpha,e}
X
E_{\nu,\beta,e}^\dagger .
\label{eq:communication-physical-schur-channel}
\end{align}
Moreover,
\begin{align}
\Tr S_{\mathcal N,\lambda\mu}^{(\nu)}
&=
\Pr[\nu\mid\lambda,\mu],
&
\sum_\nu
\Tr S_{\mathcal N,\lambda\mu}^{(\nu)}
&=
1.
\label{eq:communication-S-normalization}
\end{align}
The joint law of
$(\lambda,\mu,\nu)$ above is exactly the law obtained by the three local
Schur measurements on $\psi_{\mathcal N}^{\otimes k}$. The standard entanglement-transmission state of the conditional channel
$\mathcal N_{\lambda\to\mu}$ with code $C$ is
\begin{align}
\rho_{C|\lambda\mu}^{\rm comm}
\coloneqq
\left[
\id_{R_0}\otimes
\left(
\mathcal N_{\lambda\to\mu}\circ\mathcal C
\right)
\right](\Phi_K),
\label{eq:communication-coded-choi}
\end{align}
Substituting
\eqref{eq:communication-physical-schur-channel} gives
\begin{align}
\rho_{C|\lambda\mu}^{\rm comm}
=
\sum_{\nu,e,\alpha,\beta}
\bigl[
S_{\mathcal N,\lambda\mu}^{(\nu)}
\bigr]_{\alpha\beta}
(\id_{R_0}\otimes E_{\nu,\alpha,e}C)
\Phi_K
(\id_{R_0}\otimes
C^\dagger E_{\nu,\beta,e}^\dagger).
\label{eq:communication-coded-error-state}
\end{align}
This is precisely the coded physical Schur state appearing in the
polar-decoding analysis of
\cref{thm:universal-distillation}.  The difference is operational:
there the code arose from one branch of Alice's measurement on a
pre-existing Schur EPR register, whereas here Alice directly prepares
the encoded input $CXC^\dagger$.  Consequently a single random
isometry $C$ suffices.
Fix now $(\lambda,\mu)$ and write $d=d_\lambda$ and $D=d_\mu$.
Use the correctable region $\mathcal T_{\lambda,\mu}^{(K)}$ from \eqref{eq:accepted-error-sectors}, with the same suppressed dependence on $(k,r,L)$, and let
\begin{align}
N_{\lambda\mu}
\coloneqq
\sum_{\nu\in\mathcal T_{\lambda,\mu}^{(K)}}
g_{\lambda\mu\nu}d_\nu .
\end{align}
Restricting
\eqref{eq:communication-coded-error-state} to the accepted sectors gives
\begin{align}
\rho_{C|\lambda\mu}^{\rm good}
=
\sum_{\substack{
\nu\in\mathcal T_{\lambda,\mu}^{(K)}\\
e,\alpha,\beta}}
\bigl[
S_{\mathcal N,\lambda\mu}^{(\nu)}
\bigr]_{\alpha\beta}
(\id\otimes E_{\nu,\alpha,e}C)
\Phi_K
(\id\otimes C^\dagger
E_{\nu,\beta,e}^\dagger).
\end{align}
For each $\lambda$ with $d_\lambda\ge K$, Alice chooses $C$ before
Bob's output label is available. Bob measures $\mu$ and applies the
polar decoder for $C$ and the corresponding weighted reference channel;
if the accepted set is empty, he outputs any fixed state. This defines
a CPTP decoder $\mathcal D_{\lambda,C}$ independent of the physical
channel. Applying the same accepted-sector
packing bound and normalized polar-recovery lemma as in \cref{thm:universal-distillation}, and then averaging over 
the output label $\mu$, gives
\begin{align}
&\sum_{\lambda:d_\lambda\ge K}
p_\lambda\,
\mathbb E_C
F_e\!\left(
\pi_K,
\mathcal D_{\lambda,C}
\circ
\mathcal N^{\otimes k}
\circ
\mathcal E_{\lambda,C}
\right)
\nonumber\\
&\qquad\ge
(1-p_k)
\left(1-\frac1L\right).
\label{eq:communication-average}
\end{align}
Here we used the same accepted-event identity as in the distillation proof of \cref{thm:universal-distillation},
\[
\sum_{\lambda,\mu}
p_\lambda q_{\mu|\lambda}
\Pr\!\left[
\nu\in\mathcal T_{\lambda,\mu}^{(K)}
\,\middle|\,
\lambda,\mu
\right]
=
1-p_k.
\]
We now refine the mixed encoders into isometric encoders. Diagonalize
the known unitary-register state as
\begin{align}
\xi_{\mathcal Q_\lambda}
=\sum_a t_{\lambda,a}\ketbra*{v_{\lambda,a}}{v_{\lambda,a}}.
\end{align}
For $d_\lambda\ge K$, draw $\lambda$ with probability $p_\lambda$,
draw the isometry $C$ as above, and draw $a$ with probability
$t_{\lambda,a}$. The seed $s=(\lambda,C,a)$ specifies the isometric
encoder $\mathcal E_{\lambda,C,a}$ induced by
\begin{align}
\ket{x}\longmapsto
(U_{\rm Schur}^{A'})^\dagger
\bigl(\ket\lambda\otimes C\ket{x}\otimes\ket*{v_{\lambda,a}}\bigr),
\qquad
\mathcal D_s=\mathcal D_{\lambda,C}.
\label{eq:universal-isometric-encoder}
\end{align}
The decoder does not depend on $a$. Since
$\mathcal E_{\lambda,C}=\sum_a t_{\lambda,a}\mathcal E_{\lambda,C,a}$,
linearity of entanglement fidelity gives, for every physical channel,
\begin{align}
&\sum_a t_{\lambda,a}
F_e\!\left(\pi_K,
\mathcal D_{\lambda,C}\circ\mathcal N^{\otimes k}
\circ\mathcal E_{\lambda,C,a}\right)
\nonumber\\
&\qquad=
F_e\!\left(\pi_K,
\mathcal D_{\lambda,C}\circ\mathcal N^{\otimes k}
\circ\mathcal E_{\lambda,C}\right).
\end{align}
This refinement is performed after averaging the performance of the
mixed encoder; the output-label probabilities need not agree for the
different pure components.

If $d_\lambda<K$, no physical triple in this branch is accepted, since
$d_\mu\le d_\lambda d_\nu
\le d_\lambda\mathcal V_{\lambda\mu}^{(r)}(d_\nu)$.
For such labels, choose any fixed isometric encoder into
$(A')^{\otimes k}$ and any fixed CPTP decoder. Such an encoder exists
because $K\le d_{A'}^k$, and these branches contribute nonnegative
fidelity. Averaging over the complete seed distribution and using
\eqref{eq:communication-average}, we obtain
\begin{align}
\mathbb E_s\,
F_e\!\left(
\pi_K,
\mathcal D_s\circ
\mathcal N^{\otimes k}\circ
\mathcal E_s
\right)
&\ge
(1-p_k)\left(1-\frac1L\right),
\end{align}
and therefore, for
\(
\varepsilon_k(\mathcal N;s)
\coloneqq
1-
F_e(
\pi_K,
\mathcal D_s\circ
\mathcal N^{\otimes k}\circ
\mathcal E_s),
\)
\begin{align}
\mathbb E_s\,
\varepsilon_k(\mathcal N;s)
&\le
p_k+\frac{1-p_k}{L}.
\label{eq:sw-transmission-bound-proof}
\end{align}
Finally, the recovery estimate uses only the first and second unitary
moments. All preceding bounds therefore remain valid when $C$ is
sampled as the first $K$ columns of a unitary drawn from a finite
weighted exact $2$-design on $\mathcal P_\lambda$, as in the distillation
proof. Together with the finite choices of $\lambda$ and $a$, this
gives the claimed finite ensemble of isometric encoders and CPTP
decoders. Every probability and every code pair depends only on the
known input $\xi$ and the supplied code parameters, not on the
physical channel.
\end{proof}

\begin{theorem}[Universal quantum communication]
\label{thm:sw-transmission}
Fix an input $\xi$, an integer rank ceiling $r\ge1$, a target rate
$R\ge0$, and a gap $\delta>0$. Suppose at least one channel satisfies
\begin{align}
\rank\omega_{\mathcal N}\le r,
\qquad
R\le I_c(\xi,\mathcal N)-\delta .
\label{eq:channel-rate-promise}
\end{align}
Then there are finite code ensembles, sampled using a shared classical
seed $s$, with isometric encoders $\mathcal E_{k,s}$ and CPTP decoders
$\mathcal D_{k,s}$, depending only on
$(\xi,d_{A'},d_B,k,r,R,\delta)$, with
$K_k=2^{\lfloor kR\rfloor}$, such that for some $c>0$ and all
sufficiently large $k$,
\begin{align}
\Pr_s\!\left[
\sup_{\mathcal N\ {\rm satisfying}\ \eqref{eq:channel-rate-promise}}
\varepsilon_k(\mathcal N;s)>2^{-ck}
\right]
\le 2^{-ck},
\label{eq:universal-high-probability-transmission}
\end{align}
where $\varepsilon_k(\mathcal N;s)$ is the entanglement infidelity of
the sampled code. The seed is shared before transmission; no auxiliary
classical communication or shared entanglement is used. Fixing a
successful seed also yields a deterministic universal code.
More generally, for any specified nonempty family $\mathfrak I$ of
channels between the same spaces, every rate strictly below
\begin{align}
Q_{\rm reg}(\mathfrak I)
\coloneqq
\sup_{m\ge1}\frac1m
\max_{\xi\in\mathcal D((A')^{\otimes m})}\inf_{\mathcal N\in\mathfrak I}
I_c(\xi,\mathcal N^{\otimes m})
=
\lim_{m\to\infty}\frac1m
\max_{\xi\in\mathcal D((A')^{\otimes m})}\inf_{\mathcal N\in\mathfrak I}
I_c(\xi,\mathcal N^{\otimes m})
\label{eq:regularized-coherent-information}
\end{align}
is achieved with shared classical randomness and the same simultaneous
high-probability guarantee. The ensembles may depend on $\mathfrak I$,
but not on which member is used.
\end{theorem}

\begin{proof}
For $R=0$ take $K_k=1$. Assume $R>0$, let $\mathfrak I$ denote the
class in \eqref{eq:channel-rate-promise}, and set
$L_k=2^{k\delta/2}$. By
\cref{prop:first-order,lem:universal-code-ensemble}, the uniform rank
bounds $r_R\le d_{A'}$, $r_B\le d_B$, $r_E\le r$, together with the
common gap $\delta$, give constants $a>0$ and $k_0$ depending only on
$(\xi,d_{A'},d_B,r,R,\delta)$ such that
\begin{align}
\sup_{\mathcal N\in\mathfrak I}
\mathbb E_s\,\varepsilon_k(\mathcal N;s)
\le 2^{-ak},
\qquad k\ge k_0 .
\label{eq:uniform-random-channel-error}
\end{align}
Here the random seed $s$ specifies the complete encoder--decoder pair
constructed in \cref{lem:universal-code-ensemble}.
We now show that most seeds work simultaneously over $\mathfrak I$. Put
$\kappa=(d_{A'}d_B)^2$ and let
$\{\mathcal N_j\}_{j=1}^{M_\tau}\subset\mathfrak I$ be a
diamond-norm $\tau$-net with
\begin{align}
M_\tau\le(1+2/\tau)^\kappa .
\end{align}
By Markov's inequality, the union bound, and
\eqref{eq:uniform-random-channel-error}, for every $t>0$,
\begin{align}
\Pr_s\!\left[
\max_{j\le M_\tau}\varepsilon_k(\mathcal N_j;s)>t
\right]
\le \frac{M_\tau\,2^{-ak}}{t}.
\label{eq:simultaneous-code-selection}
\end{align}
For every fixed seed,
\begin{align}
\left|
\varepsilon_k(\mathcal N;s)-\varepsilon_k(\mathcal M;s)
\right|
&\le
\left\|
\mathcal N^{\otimes k}-\mathcal M^{\otimes k}
\right\|_\diamond
\nonumber\\
&\le
k\|\mathcal N-\mathcal M\|_\diamond .
\label{eq:channel-fidelity-continuity}
\end{align}
Consequently,
\begin{align}
\Pr_s\!\left[
\sup_{\mathcal N\in\mathfrak I}\varepsilon_k(\mathcal N;s)>t+k\tau
\right]
\le \frac{(1+2/\tau)^\kappa2^{-ak}}{t}.
\end{align}
Choose $b=a/[2(\kappa+1)]$, $t=2^{-bk}/2$, and
$\tau=2^{-bk}/(2k)$. The error threshold is $2^{-bk}$, while
the failure probability is at most
$2(1+4k)^\kappa2^{-ak/2}\le2^{-bk}$ for all sufficiently large $k$.
This proves \eqref{eq:universal-high-probability-transmission} with $c=b$.
In particular, a successful seed $s_k$ exists; fixing it fixes both
maps and consumes no shared randomness during transmission.
Taking $r=d_{A'}d_B$ gives the unrestricted-rank statement.
For the regularized assertion, let $\mathfrak I$ now be any specified
nonempty family and fix $R<Q_{\rm reg}(\mathfrak I)$. Then for some
finite $m$ and some input $\xi$ on $(A')^{\otimes m}$,
\begin{align}
\inf_{\mathcal N\in\mathfrak I}
I_c(\xi,\mathcal N^{\otimes m})
>mR .
\end{align}
Choose a positive gap below this difference and regard
$\mathcal N^{\otimes m}$ as one channel, with the trivial rank ceiling
$r_m=(d_{A'}d_B)^m$. The first part gives a common code ensemble with
$K_k=2^{\lfloor kmR\rfloor}$, whose worst-case error over
$\mathcal N\in\mathfrak I$ is exponentially small in $k$ with high probability over the shared seed.
Using $mk$ physical channel uses, and padding fewer than $m$ unused
positions when necessary, gives asymptotic rate $R$ and exponentially
vanishing error.
Finally, set
\begin{align}
a_m\coloneqq
\max_{\xi\in\mathcal D((A')^{\otimes m})}\inf_{\mathcal N\in\mathfrak I}
I_c(\xi,\mathcal N^{\otimes m}).
\end{align}
Product inputs and additivity of coherent information on tensor
products give $a_{m+\ell}\ge a_m+a_\ell$, $0\le a_m\le m\log_2 d_{A'}$, hence Fekete's lemma yields \eqref{eq:regularized-coherent-information}.
\end{proof}

\section{Connection to Petz--R\'enyi conditional entropies}
\label{app:min-entropy}

The entropy arising directly in the moment calculation is a
Petz--R\'enyi conditional entropy. For $0<\alpha<1$, write
\be
\mathcal I_\alpha
&\coloneqq -H_\alpha^\downarrow(A|B)_\rho
=D_\alpha(\rho_{AB}\Vert\id_A\otimes\rho_B)\\
&=\frac{1}{\alpha-1}
\log_2\Tr\!\left[
\rho_{AB}^{\alpha}
(\id_A\otimes\rho_B^{1-\alpha})
\right].
\label{eq:petz-coherent-information}
\ee
Here $D_\alpha(\omega\Vert\sigma)\coloneqq(\alpha-1)^{-1}\log_2\Tr(\omega^\alpha\sigma^{1-\alpha})$ is the Petz--R\'enyi divergence, and
$\downarrow$ denotes evaluation at the physical marginal.
For a purification $\psi_{ABE}$, duality gives
\be
\mathcal I_\alpha
=H_{2-\alpha}^\downarrow(A|E)_\psi
=-D_{2-\alpha}
(\rho_{AE}\Vert\id_A\otimes\rho_E),
\label{eq:petz-conditional-duality}
\ee
so $\mathcal I_\alpha$ is also a conditional entropy of $A$
given the purifying environment \cite{TomamichelBertaHayashi2014}.
As $\alpha\uparrow1$, it converges to
$H(A|E)_\psi=I(A\rangle B)_\rho$.
For orders above one, the Petz--R\'enyi divergence uses the same trace
formula, with value $+\infty$ unless the support of its first
argument is contained in that of its second argument.

To connect the sample bounds to the pure-state min-entropy
benchmark of Ref.~\cite{Leone_2025}, we also consider sandwiched
conditional R\'enyi entropies. Their optimized version tends to
conditional min-entropy as the order increases.
The sandwiched R\'enyi divergence is
\cite{MullerLennertEtAl2013,WildeWinterYang2014}
\be
\widetilde D_\alpha(\omega\Vert\sigma)
\coloneqq \frac{1}{\alpha-1}
\log_2\Tr\!\left[
\left(
\sigma^{\frac{1-\alpha}{2\alpha}}
\omega
\sigma^{\frac{1-\alpha}{2\alpha}}
\right)^\alpha
\right],
\ee
and, for $\alpha>1$, the corresponding conditional entropies are
\be
\widetilde H_\alpha^\downarrow(A|E)_\psi
&\coloneqq -\widetilde D_\alpha
(\rho_{AE}\Vert\id_A\otimes\rho_E),\\
\widetilde H_\alpha^\uparrow(A|E)_\psi
&\coloneqq \sup_{\sigma_E\in\mathcal D(E)}
\bigl\{-\widetilde D_\alpha
(\rho_{AE}\Vert\id_A\otimes\sigma_E)\bigr\}.
\label{eq:conditional-renyi}
\ee
Following Ref.~\cite{TomamichelBertaHayashi2014}, the tilde
indicates the sandwiched divergence and $\uparrow$ indicates
optimization over the reference. For an arbitrary fixed
$\sigma_E$, we write the negative divergence directly.
Negative powers act on the reference support; for $\alpha>1$,
the divergence is $+\infty$ if
$\supp\omega\not\subseteq\supp\sigma$.
For $1/2\le\alpha<1$, the same formula is used with its limiting
convention at singular references.

The optimized entropy satisfies
\be
\lim_{\alpha\to\infty}
\widetilde H_\alpha^\uparrow(A|E)_\psi
=H_{\min}(A|E)_\psi.
\label{eq:sandwiched-min-entropy-limit}
\ee
This follows from the infinite-order limit of the sandwiched
divergence to the max-relative entropy
\cite{MullerLennertEtAl2013,KRS2009}.
For pure $\rho_{AB}$, one can choose $E$ trivial, giving
$H_{\min}(A|E)_\psi=-\log_2\|\rho_A\|_\infty$,
the reduced-state min-entropy used in Ref.~\cite{Leone_2025}.

The coding analysis proceeds through $\mathcal I_\alpha$;
the comparisons in \cref{lem:regularized-min-entropy} translate these other promises
into lower bounds on it. In particular,
\be
H_{\min}(A|E)_\psi
\le\mathcal I_{1/2}
\le I(A\rangle B)_\rho.
\label{eq:petz-half-hierarchy}
\ee
The first inequality follows from min--max duality and the
second from monotonicity of the Petz--R\'enyi divergence in its order
\cite{KRS2009,TCR2010,TomamichelBertaHayashi2014}.
Thus a min-entropy promise applies directly at $\alpha=1/2$.

\subsection{Moments of the Schur--Weyl information spectrum}

The coupling used for second-order achievability also controls
exponential moments: the same correction that relates the
fluctuations of $\widehat Z_k$ to $J_k^{(r)}$ transfers the
Petz--R\'enyi moment bound to the Schur statistic.

\begin{lemma}[Petz--R\'enyi moments of the Schur spectrum]
\label{lem:petz-schur-moment}
Let $r\ge r_E$ be the supplied rank ceiling, let $\chi$ be as in
\cref{eq:rank-Q}, and set
\be
C_k\coloneqq (k+1)^{\chi+b_{r_B}},\qquad
v\coloneqq \frac{1-\alpha}{2-\alpha},\qquad 0<\alpha<1.
\ee
Then
\be
\mathbb E\,2^{-vJ_k^{(r)}}
\le C_k^v2^{-vk\mathcal I_\alpha}.
\label{eq:petz-schur-moment}
\ee
Consequently, the universal protocol with target $K$ and
coding margin $L\ge1$ satisfies
\be
\varepsilon_k
\le L^{-1}
+\left[C_kKL\,2^{-k\mathcal I_\alpha}\right]^v.
\label{eq:petz-universal-error}
\ee
\end{lemma}

\begin{proof}
Put $s=1-\alpha$, so $v=s/(1+s)$.
By \cref{lem:schur-classical,eq:advanced-sw-classical-coupling},
\be
J_k^{(r)}\ge\widehat Z_k-\Delta_k,\qquad
\Delta_k\coloneqq \Delta_{B,k}+\chi\log_2(k+1),\qquad
\mathbb E\,2^{\Delta_k}\le C_k.
\ee
The classical law \eqref{eq:sw-classical-law} gives
\be
\mathbb E\,2^{-sZ}
=\Tr\!\left[
\rho_{AB}^{1-s}(\id_A\otimes\rho_B^s)
\right]
=2^{-s\mathcal I_{1-s}}.
\ee
Since $\widehat Z_k$ is a sum of independent copies of $Z$,
H\"older's inequality with exponents $1+s$ and $(1+s)/s$ yields
\be
\mathbb E\,2^{-vJ_k^{(r)}}
&\le\mathbb E\!\left[
2^{-v\widehat Z_k}2^{v\Delta_k}
\right]\\
&\le
\bigl(\mathbb E\,2^{-s\widehat Z_k}\bigr)^{1/(1+s)}
\bigl(\mathbb E\,2^{\Delta_k}\bigr)^v\\
&\le C_k^v2^{-vk\mathcal I_{1-s}}.
\label{eq:petz-holder-transfer}
\ee
This step does not require independence between the information
sum and the Schur correction. Markov's inequality gives
\be
\Pr\!\left[J_k^{(r)}<\log_2 K+\log_2 L\right]
\le\left[C_kKL\,2^{-k\mathcal I_{1-s}}\right]^v.
\ee
Combining this with \eqref{eq:sw-master} proves the error bound.
\end{proof}

\subsection{Removing the local-rank dependence}

To replace the local-rank dependence in $C_k$ by dependence
only on the global rank ceiling, we shift Bob's marginal.
For $0<s<1$ and $\tau\ge0$, define
\be
\mathcal I_{1-s,\tau}
&\coloneqq D_{1-s}\bigl(
\rho_{AB}\Vert\id_A\otimes(\rho_B+\tau\id_B)
\bigr)\\
&=-\frac1s\log_2\Tr\!\left[
\rho_{AB}^{1-s}
\bigl(\id_A\otimes(\rho_B+\tau\id_B)^s\bigr)
\right].
\label{eq:regularized-affinity}
\ee
At $\tau=0$, this is $\mathcal I_{1-s}$.
The reference is left unnormalized to avoid introducing its
dimension-dependent trace. The shift replaces the local rank
in the spectral estimates by
$\Tr[\rho_B(\rho_B+\tau\id_B)^{-1}]\le1/\tau$.

\begin{lemma}[Regularized Schur moment]
\label{lem:regularized-schur}
Let $r\ge r_E$, $0<s<1$, and $\tau>0$.
Write $p(k)$ for the number of partitions of $k$, and set
\be
v\coloneqq \frac{s}{1+s},\qquad
a_{r,\tau}\coloneqq \frac r\tau+\frac{r^2}{2\tau^2},\qquad
C_{k,r,\tau}\coloneqq p(k)\exp(a_{r,\tau}).
\ee
For every $k\ge1$,
\be
\mathbb E\,2^{-vJ_k^{(r)}}
&\le C_{k,r,\tau}^{v}2^{-vk\mathcal I_{1-s,\tau}},\\
H_{\mathrm{SW},k}^{\varepsilon,(r)}(A|E)_\psi
&\ge k\mathcal I_{1-s,\tau}-\log_2 C_{k,r,\tau}
-\frac{1+s}{s}\log_2(1/\varepsilon),
\qquad 0<\varepsilon<1.
\label{eq:regularized-sw-entropy}
\ee
\end{lemma}

\begin{proof}
Use the probability law $P(i,j)$ from
\eqref{eq:sw-classical-law}, replacing $Z$ by
\be
Z_\tau(i,j)\coloneqq \log_2 e_j-\log_2(b_i+\tau).
\ee
Its exponential moment is
$\mathbb E\,2^{-sZ_\tau}=2^{-s\mathcal I_{1-s,\tau}}$.
The coupling construction in the proof of
\cref{lem:schur-classical} applies with
$\rho_B^{\otimes k}$ replaced, as a spectral observable,
by $(\rho_B+\tau\id_B)^{\otimes k}$.
Indeed, this operator also commutes with Bob's Schur projectors.
Writing $y_E,y_{B,\tau}$ for the joint spectral outcomes of
$\rho_E^{\otimes k}$ and $(\rho_B+\tau\id_B)^{\otimes k}$,
the outcome
$\widehat Z_{\tau,k}\coloneqq \log_2 y_E-\log_2 y_{B,\tau}$
has the law of a sum of $k$ independent copies of $Z_\tau$.
The environment correction remains $\Delta_{E,k}\ge0$.
Define the regularized Bob correction and the total correction by
\be
\Delta_{B,\tau,k}\coloneqq -\log_2(d_\mu y_{B,\tau}),\qquad
\Delta_{\tau,k}\coloneqq \Delta_{B,\tau,k}
+\log_2 Q_{\lambda\mu}^{(r)}.
\ee
The same dimension identity and multiplicity-volume bound give
\be
\log_2 d_\mu-\log_2 d_\nu
&=\widehat Z_{\tau,k}+\Delta_{E,k}-\Delta_{B,\tau,k},\\
J_k^{(r)}&\ge\widehat Z_{\tau,k}-\Delta_{\tau,k}.
\label{eq:regularized-defect-coupling}
\ee
The correction $\Delta_{B,\tau,k}$ need not be nonnegative;
only its exponential moment is needed.
The new step is to bound the multiplicity and the Bob correction
together. Let $x$ be the positive spectrum of
$\rho_B(\rho_B+\tau\id_B)^{-1}$, and let $y=x\otimes1^r$
repeat each entry $r$ times. With $s_\lambda$ denoting a Schur
polynomial and $p_\mu=\Pr(\mu)$, Schur--Weyl duality and the
branching identity give
\be
p_\mu\mathbb E\!\left[
2^{\Delta_{B,\tau,k}}\mid\mu
\right]
&=d_\mu^{-1}\Tr\!\left[
\Pi_\mu^B
\bigl(\rho_B(\rho_B+\tau\id_B)^{-1}\bigr)^{\otimes k}
\right]
=s_\mu(x),\\
s_\lambda(y)
&=\sum_{\mu',\eta}
g_{\lambda\mu'\eta}s_{\mu'}(x)s_\eta(1^r)
\ge Q_{\lambda\mu}^{(r)}s_\mu(x).
\label{eq:regularized-characters}
\ee
The inequality uses $s_\eta(1^r)\ge1$ whenever
$\ell(\eta)\le r$ \cite{Macdonald1995}.
Set $M\coloneqq \sum_{\lambda\vdash k}s_\lambda(y)$.
On every positive-probability $\mu$ branch,
$s_\mu(x)>0$ and $Q_{\lambda\mu}^{(r)}\le M/s_\mu(x)$.
Consequently,
\be
\mathbb E\,2^{\Delta_{\tau,k}}
&=\mathbb E\!\left[
Q_{\lambda\mu}^{(r)}2^{\Delta_{B,\tau,k}}
\right]\\
&\le\sum_{\mu:p_\mu>0}
\frac{M}{s_\mu(x)}\,
p_\mu\mathbb E\!\left[
2^{\Delta_{B,\tau,k}}\mid\mu
\right]
\le p(k)M.
\ee
This cancellation removes the local-rank dependence without
assuming independence between the Schur labels and spectral outcomes.
It remains to bound $M$. Since every $y_i<1$, Littlewood's
identity can be evaluated at one \cite{Macdonald1995}:
\be
M\le\sum_\eta s_\eta(y)
=\prod_i(1-y_i)^{-1}
\prod_{i<j}(1-y_iy_j)^{-1},
\ee
where the sum runs over all partitions.
Put $u_i=y_i/(1-y_i)$. Each $u_i$ equals $b_j/\tau$
for some eigenvalue $b_j$ of $\rho_B$, repeated $r$ times,
so $\sum_i u_i=r/\tau$.
Using $-\ln(1-a)\le a/(1-a)$ gives
\be
\ln\!\left[
\prod_i(1-y_i)^{-1}\prod_{i<j}(1-y_iy_j)^{-1}
\right]
&\le\sum_i u_i+
\sum_{i<j}\frac{u_iu_j}{1+u_i+u_j}\\
&\le\frac r\tau+\frac12\left(\frac r\tau\right)^2
=a_{r,\tau}.
\label{eq:je-regularized-generating-bound}
\ee
Therefore $\mathbb E\,2^{\Delta_{\tau,k}}\le C_{k,r,\tau}$, with
\be
\log_2 C_{k,r,\tau}
=\log_2 p(k)+a_{r,\tau}\log_2\mathrm e.
\label{eq:je-regularized-overhead}
\ee
Applying the H\"older estimate \eqref{eq:petz-holder-transfer}
to $(\widehat Z_{\tau,k},\Delta_{\tau,k})$ proves the moment bound.
Markov's inequality at
$k\mathcal I_{1-s,\tau}-\log_2 C_{k,r,\tau}
-v^{-1}\log_2(1/\varepsilon)$ then proves the entropy bound.
\end{proof}

\section{Sample-complexity analysis}\label{app:sample-complexity}
The coupling in \cref{lem:schur-classical} allows us to control the information density directly. This avoids paying a quadratic rank factor for the fluctuations of each empirical spectrum separately.

\begin{lemma}[Concentration of the classical information density]
\label{lem:information-concentration}
Let $Z$, $\widehat Z_k=\sum_{i=1}^k Z_i$, and $I=\mathbb EZ$ be as in
\cref{lem:schur-classical}, and set
\begin{align}
M\coloneqq \log_2(\mathrm e r_Br_E).
\end{align}
There is a universal constant $c_0>0$ such that, for every
$0<\delta\leq1$,
\begin{align}
\Pr[\widehat Z_k-kI\leq-k\delta/2]
\leq
2\exp\!\left(
-c_0\frac{k\delta^2}{M^2}
\right).
\label{eq:uniform-information-tail}
\end{align}
\end{lemma}

\begin{proof}
Recall that
$Z(i,j)=\log_2 e_j-\log_2 b_i$, and that under the law
\eqref{eq:sw-classical-law} the marginals of $i$ and $j$ are
$(b_i)$ and $(e_j)$, respectively.
We now use the standard sub-exponential norm
$\|X\|_{\psi_1}\coloneqq \inf\{s>0:\mathbb E e^{|X|/s}\leq2\}$.
For a probability vector $p$ supported on at most $s$ points, the
surprisal $Y=-\log_2 p(X)$ satisfies
$\|Y-\mathbb EY\|_{\psi_1}\leq C\log_2(\mathrm e s)$ for a universal
constant $C$. Indeed, this follows directly from
$\sum_xp_x^{1-\theta}\leq s^\theta$ and the standard equivalent
characterizations of sub-exponential random variables
\cite[Proposition~2.8.1]{vershynin2026highdimensional}.
Applying this separately to $-\log_2 b_i$ and $-\log_2 e_j$, and using the
triangle inequality for the $\psi_1$ norm, gives
\begin{align}
\|Z-I\|_{\psi_1}
\leq C_0\log_2(\mathrm e r_Br_E)
=C_0M
\end{align}
for a universal constant $C_0$.
The sub-exponential Bernstein inequality
\cite[Theorem~2.9.1]{vershynin2026highdimensional}
therefore gives
\begin{align}
\Pr[\widehat Z_k-kI\leq-k\delta/2]
&\leq
2\exp\!\left[
-c\min\left\{
\frac{k\delta^2}{M^2},
\frac{k\delta}{M}
\right\}
\right].
\end{align}
Since $0<\delta\leq1$ and $M\geq1$, the quadratic term is the smaller
one. Absorbing universal constants into $c_0$ proves
\eqref{eq:uniform-information-tail}.
\end{proof}

We focus on the rank-scaling regime, in which the actual ranks grow while
the rate gap $\delta>0$ and target infidelity $\varepsilon>0$ are fixed.
In this regime, the concentration requirement grows only
polylogarithmically with the ranks, as $O(\log_2^2(r_Br_E))$, whereas the
Schur overhead grows polynomially, up to logarithmic factors, as
$\widetilde O(r_A r_B r+r_B^2)$. The latter therefore determines the
asymptotic copy complexity:
\[
k=\widetilde O(r_A r_B r+r_B^2).
\]
This is the natural regime for comparing the rank dependence of the
protocol with dimension-dependent procedures such as global tomography.
More generally, the same scaling applies whenever the rank-dependent
overhead dominates the concentration requirement.

\begin{proposition}[Sample complexity]
\label{prop:sample-complexity}
Fix a supplied rank ceiling $r\geq r_E$, a target rate $R\geq0$,
$0<\delta\leq1$, and $0<\varepsilon<1/2$, and let
$I=I(A\rangle B)_\rho$. If $R\leq I-\delta$, then, in the
rank-dominated regime, the universal protocol produces
$K=2^{\lfloor kR\rfloor}$ with entanglement infidelity at most
$\varepsilon$ using
\begin{align}
k
=
\widetilde O\!\left(
\frac{r_A r_B r+r_B^2+\log_2(1/\varepsilon)}{\delta}
\right)
\label{eq:sample-complexity}
\end{align}
copies, independently of the Hilbert space dimensions. Here the
rank-dominated regime means that the bound above dominates the
concentration scale
$\delta^{-2}\log_2^2(\mathrm e r_Br_E)\log_2(1/\varepsilon)$.
The notation $\widetilde O$ hides the additional factor
\[
\log_2\!\left(
2+\frac{r_A r_B r+r_B^2+\log_2(1/\varepsilon)}{\delta}
\right),
\]
and hence only logarithmic dependence on the ranks and $1/\delta$, and
at most an additional $\log_2\log_2(1/\varepsilon)$ dependence on the target
infidelity.
\end{proposition}

\begin{proof}
Let
$\chi=r_A r_B r-r_A(r_A+1)/2\leq r_A r_B r$ and set
$L=3/\varepsilon$. It is sufficient that
\begin{gather}
k\geq
\frac{\log_2^2(\mathrm e r_Br_E)}{c_0\delta^2}
\ln\frac{6}{\varepsilon},
\label{eq:sample-condition-concentration}\\
(\chi+b_{r_B})\log_2(k+1)
+2\log_2\frac{3}{\varepsilon}
\leq
\frac{k\delta}{2}.
\label{eq:sample-condition-overhead}
\end{gather}
Indeed, with
$u=b_{r_B}\log_2(k+1)+\log_2(3/\varepsilon)$,
\cref{lem:schur-classical} gives
$\Pr[\Delta_{B,k}>u]\leq\varepsilon/3$. On the complementary event,
\eqref{eq:advanced-sw-classical-coupling} yields
\[
J_k^{(r)}
\geq
\widehat Z_k-(\chi+b_{r_B})\log_2(k+1)-\log_2(3/\varepsilon).
\]
Since $\log_2 K\leq kR$ and $R\leq I-\delta$,
\eqref{eq:sample-condition-overhead} implies that
$J_k^{(r)}<\log_2 K+\log_2 L$ only if
$\widehat Z_k<kI-k\delta/2$. Therefore
\begin{align}
\Pr[J_k^{(r)}<\log_2 K+\log_2 L]
&\leq
2\exp\!\left(
-c_0\frac{k\delta^2}{\log_2^2(\mathrm e r_Br_E)}
\right)
+\frac{\varepsilon}{3}
\nonumber\\
&\leq
\frac{2\varepsilon}{3},
\end{align}
where we used \cref{lem:information-concentration} and
\eqref{eq:sample-condition-concentration}. Together with the coding
term $1/L=\varepsilon/3$, this gives total infidelity at most
$\varepsilon$. Finally, solving \eqref{eq:sample-condition-overhead} gives
\begin{align}
k
=
O\!\left[
\frac{\chi+b_{r_B}+\log_2(1/\varepsilon)}{\delta}
\log_2\!\left(
2+\frac{\chi+b_{r_B}+\log_2(1/\varepsilon)}{\delta}
\right)
\right].
\end{align}
In the rank-dominated regime, this choice also satisfies
\eqref{eq:sample-condition-concentration}. Using
$\chi\leq r_A r_B r$ and $b_{r_B}\leq r_B^2$ gives
\eqref{eq:sample-complexity}.
\end{proof}

For comparison, consider full state tomography followed by a
distillation code tailored to the estimate. If $n$ fresh copies
are used for distillation and the final infidelity is certified
through the additive bound $\varepsilon_0+n\delta+\beta\le\varepsilon$,
where $\varepsilon_0$ is the estimated-state coding error,
$\delta$ the trace-distance reconstruction accuracy, and $\beta$
the tomography failure probability, then $\delta\le\varepsilon/n$.
Uniform full tomography over states of rank at most $r$ therefore
requires $m=\Omega(r d_A d_B n^2/\varepsilon^2)$ tomography copies
in the worst case and small-error regime, even allowing arbitrary
collective measurements \cite{ScharnhorstSpileckiWright2025}.
This is the cost of the stated additive certificate: the
$n^2/\varepsilon^2$ dependence is not a general lower bound on
tomography-based distillation, and other reconstruction metrics
can improve it \cite{PelecanosSpileckiWright2025}.
The comparison highlights the Hilbert space dimension cost of full state
reconstruction, which our protocol avoids by distilling directly
from the Schur--Weyl sectors.

\subsection{Other entropy promises}

The shift in \eqref{eq:regularized-affinity} lowers the Petz--R\'enyi quantity. We first bound this loss
and compare the entropy promises, then choose the shift
using a supplied gap above the target rate.

\begin{lemma}[Entropy comparison and regularization cost]
\label{lem:regularized-min-entropy}
Let $r\ge r_E$. For $0<s<1$ and $\tau>0$,
\be
2^{-s\mathcal I_{1-s,\tau}}
\le 2^{-s\mathcal I_{1-s}}+(r\tau)^s.
\label{eq:affinity-minentropy}
\ee
Moreover, $\mathcal I_{1-s}\ge h$ for each of the choices
\be
h=\mathcal I_{1-s},
&\qquad 0<s<1;\\
h=\widetilde H_\alpha^\uparrow(A|E)_\psi,
&\qquad s=\frac{\alpha-1}{2\alpha-1},
\quad 1<\alpha<\infty;\\
h=\widetilde H_\alpha^\downarrow(A|E)_\psi,
&\qquad s=\frac{\alpha-1}{\alpha},
\quad 1<\alpha<\infty;\\
h=H_{\min}(A|E)_\psi,
&\qquad s=\frac12.
\label{eq:renyi-choices}
\ee
For $s=(\alpha-1)/(2\alpha-1)$, the comparison also holds
with
$h=-\widetilde D_\alpha
(\rho_{AE}\Vert\id_A\otimes\sigma_E)$
for any reference state $\sigma_E$.
\end{lemma}

\begin{proof}
Since $\rho_B$ commutes with its shift, the scalar inequality
$(b+\tau)^s\le b^s+\tau^s$ gives
\be
2^{-s\mathcal I_{1-s,\tau}}
&\le 2^{-s\mathcal I_{1-s}}
+\tau^s\Tr\rho_{AB}^{1-s}\\
&\le 2^{-s\mathcal I_{1-s}}+(r\tau)^s,
\ee
where the last step uses $\rank\rho_{AB}\le r$.
For the optimized sandwiched entropy, set
$\beta=\alpha/(2\alpha-1)$.
Purification duality and the Araki--Lieb--Thirring
inequality yield
\be
\widetilde H_\alpha^\uparrow(A|E)_\psi
&=\inf_{\sigma_B\text{ state}}
\widetilde D_\beta
(\rho_{AB}\Vert\id_A\otimes\sigma_B)\\
&\le
\widetilde D_\beta
(\rho_{AB}\Vert\id_A\otimes\rho_B)\\
&\le
D_\beta(\rho_{AB}\Vert\id_A\otimes\rho_B)
=\mathcal I_\beta.
\ee
Thus $s=1-\beta=(\alpha-1)/(2\alpha-1)$.
For the physical-marginal sandwiched entropy, the mixed
Petz--sandwiched duality gives
\be
\widetilde H_\alpha^\downarrow(A|E)_\psi
&=\inf_{\sigma_B\text{ state}}
D_{1/\alpha}
(\rho_{AB}\Vert\id_A\otimes\sigma_B)\\
&\le\mathcal I_{1/\alpha},
\ee
so $s=(\alpha-1)/\alpha$.
Both dualities are stated in
Ref.~\cite[Theorem~3]{TomamichelBertaHayashi2014};
boundary reference states are understood by limits.
The min-entropy comparison is
\eqref{eq:petz-half-hierarchy}, and the Petz--R\'enyi choice gives the
equality. Finally, evaluation at a fixed reference gives
no larger entropy than optimization over all reference
states:
\[
-\widetilde D_\alpha
(\rho_{AE}\Vert\id_A\otimes\sigma_E)
\le \widetilde H_\alpha^\uparrow(A|E)_\psi.
\]
\end{proof}

Combining this lemma with \cref{lem:regularized-schur}
gives, for every admissible pair $(h,s)$, every $\tau>0$,
and $0<\eta<1$,
\be
H_{\mathrm{SW},k}^{\eta,(r)}(A|E)_\psi
\ge
-\frac{k}{s}\log_2\!\left(2^{-sh}+(r\tau)^s\right)
-\log_2 C_{k,r,\tau}
-\frac{1+s}{s}\log_2(1/\eta).
\label{eq:direct-renyi-entropy}
\ee
The free parameter $\tau$ may be optimized: increasing it
lowers the entropy term but reduces the Schur overhead.
We balance the two costs using a supplied entropy gap.

\begin{proposition}[Sample complexity from an entropy promise]
\label{prop:min-entropy-samples}
Fix an integer rank ceiling $r\ge1$, a target rate $R\ge0$,
a gap $0<\delta\le1$, and $0<\varepsilon<1/2$.
Choose $h$ and $s$ as in
\cref{lem:regularized-min-entropy}.
There is a constant $c_s>0$, depending only on $s$, such that
the protocol of \cref{thm:universal-distillation}, with
$K=2^{\lfloor kR\rfloor}$ and $L=2/\varepsilon$, has
infidelity at most $\varepsilon$ on every state satisfying
\be
\rank\rho_{AB}\le r,
\qquad h\ge R+\delta,
\ee
whenever the integer $k$ satisfies
\be
k\ge c_s\left(
\frac{r^4\,2^{2R}}{\delta^{\,1+2/s}}
+\frac{\log_2(1/\varepsilon)}{\delta}
\right).
\label{eq:min-entropy-samples}
\ee
For fixed local Hilbert spaces, the same protocol works
throughout this promised class.
\end{proposition}

\begin{proof}
Choose the shift to spend half of the supplied entropy gap:
\be
\tau\coloneqq
\frac{2^{-R-\delta/2}}{r}
\left(1-2^{-s\delta/2}\right)^{1/s}.
\label{eq:sample-complexity-shift}
\ee
By \cref{lem:regularized-min-entropy},
\be
2^{-s\mathcal I_{1-s,\tau}}
&\le 2^{-s(R+\delta)}+(r\tau)^s\\
&=2^{-s(R+\delta/2)},
\ee
and therefore $\mathcal I_{1-s,\tau}\ge R+\delta/2$.
The remaining half of the gap pays for the finite-copy
overhead. Writing $v=s/(1+s)$, the regularized moment
bound, Markov's inequality, and \eqref{eq:sw-master} give
\be
\varepsilon_k
&\le L^{-1}
+\left[
C_{k,r,\tau}KL\,2^{-k\mathcal I_{1-s,\tau}}
\right]^v\\
&\le\frac{\varepsilon}{2}
+\left[
\frac{2C_{k,r,\tau}}{\varepsilon}
\,2^{-k\delta/2}
\right]^v.
\ee
Consequently, $\varepsilon_k\le\varepsilon$ whenever
\be
\frac{k\delta}{2}
\ge
\log_2 C_{k,r,\tau}
+\left(2+\frac1s\right)\log_2\frac2\varepsilon,
\label{eq:min-entropy-explicit}
\ee
since $1+1/v=2+1/s$.
To solve this condition, use
\eqref{eq:je-regularized-overhead} and
\be
\log_2 p(k)\le c_0\sqrt{k},
\qquad
c_0\coloneqq \frac{\pi\sqrt{2/3}}{\ln2}.
\ee
Indeed, the partition generating function gives, for
every $t>0$,
\[
\ln p(k)
\le kt+\sum_{j\ge1}\frac{1}{j(e^{tj}-1)}
\le kt+\frac{\pi^2}{6t},
\]
and one takes $t=\pi/\sqrt{6k}$.
Set
\be
A\coloneqq a_{r,\tau}
=\frac r\tau+\frac{r^2}{2\tau^2},
\qquad
B\coloneqq A\log_2\mathrm e
+\left(2+\frac1s\right)\log_2\frac2\varepsilon.
\ee
It suffices that $k\delta/2\ge c_0\sqrt{k}+B$.
Solving this quadratic inequality in $\sqrt{k}$ gives
the explicit sufficient condition
\be
k\ge
\left\lceil
\left(
\frac{c_0+\sqrt{c_0^2+2\delta B}}{\delta}
\right)^2
\right\rceil.
\label{eq:renyi-copy-inversion}
\ee
Its right-hand side is
\[
O_s\!\left(
\delta^{-2}
+\frac{A}{\delta}
+\frac{\log_2(1/\varepsilon)}{\delta}
\right).
\]
Since
$1-2^{-s\delta/2}\ge(s\ln2)\delta/4$
for the stated parameter range,
\eqref{eq:sample-complexity-shift} implies
\be
A=O_s\!\left(r^4\,2^{2R}\delta^{-2/s}\right).
\ee
Substitution proves \eqref{eq:min-entropy-samples};
the term $\delta^{-2}$ is absorbed by its first term.
All constants are independent of the state and local
dimensions, and no additional logarithmic factors
are suppressed.
\end{proof}

The power of $\delta^{-1}$ is $5+2/(\alpha-1)$ for
$\widetilde H_\alpha^\uparrow$, $3+2/(\alpha-1)$ for
$\widetilde H_\alpha^\downarrow$, and $5$ for
$H_{\min}$. The arbitrary-reference promise in
\cref{lem:regularized-min-entropy} has the same exponent
as the optimized sandwiched entropy.
These guarantees require no marginal-rank assumption.
For fixed local spaces, the operations use only $(k,K,r,L)$.
The entropy promise determines sufficient choices of these
parameters; the state, purification, reference, and shift
are used only in the analysis.
A supplied lower bound $h_0\le h$ allows distillation
arbitrarily close to $h_0$. For
$0<\gamma<\min\{h_0,1\}$, take
$R=h_0-\gamma/2$ and $\delta=\gamma/2$.
Increasing $k$ to at least $\lceil2/\gamma\rceil$ ensures
\[
\frac{\log_2 K}{k}\ge h_0-\gamma.
\]
For a fixed finite-dimensional state,
$\widetilde H_\alpha^\uparrow(A|E)_\psi
\to H(A|E)_\psi=I(A\rangle B)_\rho$
as $\alpha\downarrow1$
\cite{TomamichelBertaHayashi2014},
but the sample bounds are not uniform in this limit.

\subsection{Sample-efficient regime}
\label{sec:sample-sketch}

Capping the target rate at $\log_2 n$, for a supplied size
parameter $n\ge2$, bounds the factor $2^{2R}$ by $n^2$.
This gives the following sample-efficient regime.

\begin{corollary}[Sample complexity at a capped conditional entropy rate]
\label{cor:capped-entropy}
Choose $h$ and $s$ as in
\cref{lem:regularized-min-entropy}.
Fix integers $n\ge2$ and $r\ge1$, a supplied lower bound
$h_0>0$, and parameters
\[
0<\delta<\min\{h_0,\log_2 n\},
\qquad \delta\le1,
\qquad 0<\varepsilon<1/2.
\]
For fixed local Hilbert spaces, there is a
state-independent one-way LOCC channel that, on every
state satisfying $\rank\rho_{AB}\le r$ and $h\ge h_0$,
produces a state with infidelity at most $\varepsilon$
relative to $\Phi_K$, with
\be
\frac{\log_2 K}{k}
&\ge \min\{h_0,\log_2 n\}-\delta,\\
k
&=O_s\!\left(
\frac{r^4\,2^{2\min\{h_0,\log_2 n\}}}
{\delta^{\,1+2/s}}
+\frac{\log_2(1/\varepsilon)}{\delta}
\right).
\label{eq:capped-promise}
\ee
The sample bound is independent of the local dimensions.
\end{corollary}

\begin{proof}
Put $t\coloneqq \min\{h_0,\log_2 n\}$.
Apply \cref{prop:min-entropy-samples} with target
$R=t-\delta/2$ and gap $\delta/2$.
The promise holds because
$R+\delta/2=t\le h_0\le h$.
Choose $k$ equal to the sufficient integer in
\eqref{eq:renyi-copy-inversion} with these parameters,
increasing it to $\lceil2/\delta\rceil$ if necessary.
Then
\be
\frac{\log_2 K}{k}
&=\frac{\lfloor kR\rfloor}{k}
\ge R-\frac1k
\ge t-\delta,\\
2^{2R}&\le2^{2t}\le n^2.
\ee
The extra $O(1/\delta)$ copies are absorbed by the stated
bound, and all choices are uniform over the promised class.
\end{proof}

For a supplied promise
$h_0\le\widetilde H_\alpha^\uparrow(A|E)_\psi$,
the corollary gives
\be
\frac{\log_2 K}{k}
&\ge\min\{h_0,\log_2 n\}-\delta,\\
k
&=O_\alpha\!\left(
\frac{r^4n^2}{\delta^{\,5+2/(\alpha-1)}}
+\frac{\log_2(1/\varepsilon)}{\delta}
\right),
\qquad 1<\alpha<\infty.
\label{eq:capped-alpha}
\ee
The other entropy promises use the exponents given above.
Varying the supplied bound $h_0$ describes the achievable
rate envelope, without requiring entropy estimation.
If $\min\{h_0,\log_2 n\}\le\delta$, outputting zero ebits
already satisfies the rate inequality.
For fixed $s$, polynomial bounds on the rank ceiling,
inverse gap, and $\log_2(1/\varepsilon)$ give polynomial
sample complexity. Explicitly, if
$r\le n^a$, $\delta\ge n^{-c}$, and
$\log_2(2/\varepsilon)\le2n^d$ for fixed $a,c,d\ge0$,
one can choose
\be
k\le C_{s,a,c,d}\left(
n^{\,4a+2+c(1+2/s)}
+n^{\,d+c}
\right),
\label{eq:polynomial-renyi-samples}
\ee
with a constant independent of the state and local
dimensions. The sample bound remains polynomial even
when the local spaces have dimension $2^n$.

For comparison with state-dependent coding guarantees based on
$\widetilde H_{1+u}^\downarrow(A|E)_\psi$, take $u>0$ and
$s=u/(1+u)$. The same entropy certifies the universal protocol,
with the overhead displayed above. The pure-state min-entropy
specialization gives the benchmark discussed in
Appendix~\ref{app:min-entropy}. These are rate and sample-complexity guarantees;
computational efficiency of the recovery requires a separate analysis.
Finally, at fixed $r,s,\tau$, the regularized Schur bound
has overhead $O_{r,s,\tau}(\sqrt{k}+\log_2(1/\eta))$ relative
to $k\mathcal I_{1-s,\tau}$. Taking $\eta_k=k^{-2}$ and
$L_k=k^2$ makes the infidelity at most $2k^{-2}$, while
the error-dependent correction and coding margin cost only
$O_s(\log_2 k)$ bits. Thus every positive rate below
$\mathcal I_{1-s,\tau}$ is achievable with vanishing error.
For a fixed finite-dimensional input, subsequently taking
$\tau\downarrow0$ recovers the un-shifted rate. The explicit
shift in \eqref{eq:sample-complexity-shift} gives the
corresponding finite-copy guarantee at a supplied positive
gap. 

\section{First- and second-order optimality}
For a bipartite state $\rho_{AB}$, let
\be
R(\rho)
\coloneqq
\min_{\sigma\in\mathrm{PPT}'}
D(\rho\Vert\sigma),
\qquad
\mathcal R(\rho)
\coloneqq
\mathrm{argmin}_{\sigma\in\mathrm{PPT}'}
D(\rho\Vert\sigma)
\ee
denote the Rains bound and the set of its minimizers, respectively. Here $\mathrm{PPT}'\coloneqq \{\sigma\in\mathcal L(AB):\sigma\ge0,\ \|\sigma^{T_B}\|_1\le1\}$, where $T_B$ denotes partial transpose; these positive operators need not be normalized. We use $D(\rho\Vert\sigma)\coloneqq\Tr\rho(\log_2\rho-\log_2\sigma)$ when $\supp\rho\subseteq\supp\sigma$, and $+\infty$ otherwise.
Since the second-order upper bound of
\cite[Thm.~7]{Fang_2019} holds for any
$\sigma\in\mathcal R(\rho)$, we may choose the minimizer giving the
tightest bound. Accordingly, define
\be
V_R^\varepsilon(\rho)
\coloneqq
\begin{cases}
\displaystyle
\max_{\sigma\in\mathcal R(\rho)}
V(\rho\Vert\sigma),
& 0<\varepsilon<\frac12,\\[1.2ex]
\displaystyle
\min_{\sigma\in\mathcal R(\rho)}
V(\rho\Vert\sigma),
& \frac12<\varepsilon<1,
\end{cases}
\ee
where $V(\rho\Vert\sigma)\coloneqq \Tr\rho[\log_2\rho-\log_2\sigma-D(\rho\Vert\sigma)\id_{AB}]^2$ is the quantum information variance, with logarithms evaluated on their supports and extended by zero on their kernels.
For $\varepsilon=\frac12$, the second-order term vanishes and the
choice of minimizer is immaterial.
Then, for every fixed $0<\varepsilon<1$,
\cite[Thm.~7]{Fang_2019} gives
\be
k I(A\rangle B)_\rho
+\sqrt{kV(A\rangle B)_\rho}\,
\Phi^{-1}(\varepsilon)
+O(\log_2 k)
\le
E_{D,\mathrm{LOCC}}^{(1),\varepsilon}(\rho^{\otimes k})
\le
kR(\rho)
+\sqrt{kV_R^\varepsilon(\rho)}\,
\Phi^{-1}(\varepsilon)
+O(\log_2 k),
\label{eq:general-second-order-sandwich}
\ee
where $E_{D,\mathrm{LOCC}}^{(1),\varepsilon}$ denotes the largest $\log_2 K$ achievable by LOCC with Bell-state infidelity at most $\varepsilon$, and $V(A\rangle B)_\rho$ is the coherent information variance.
Consequently, whenever
\be
R(\rho)=I(A\rangle B)_\rho,
\qquad
V_R^\varepsilon(\rho)=V(A\rangle B)_\rho,
\ee
the lower and upper bounds coincide to second order, and hence
\be
E_{D,\mathrm{LOCC}}^{(1),\varepsilon}(\rho^{\otimes k})
=
k I(A\rangle B)_\rho
+\sqrt{kV(A\rangle B)_\rho}\,
\Phi^{-1}(\varepsilon)
+O(\log_2 k).
\ee
In particular, maximally correlated states satisfy these conditions
\cite[Prop.~10]{Fang_2019}. Thus, for $I>0$ and $V>0$, \cref{prop:second-order} shows that our universal entanglement distillation protocol is optimal up to second order on this class of states.

\section{Entanglement distillation from almost i.i.d.\ sources}
\label{app:non-iid}
We first show that the common random permutation gives the same error model as in \cref{thm:universal-distillation}. We then estimate the Schur labels from the averaged two-position marginal and use the earlier coding bound to prove \cref{cor:non-iid-main}.

\subsection{Permutation averaging and the existing coding bound}
Alice chooses a uniformly random permutation, sends it to Bob, and both apply it to their positions. The later encoder and decoder do not use its value. The input averaged over this choice is therefore $\mathcal T_k(\omega^{(k)})$ in \eqref{eq:main-non-iid-twirl}.
Every permutation-invariant state can be expanded in a finite basis of tensor powers states. Thus there are states $\rho_j$ and real coefficients $c_j$ such that
\begin{align}
\mathcal T_k(\omega^{(k)})=\sum_jc_j\rho_j^{\otimes k},
\qquad \sum_jc_j=1.
\label{eq:non-iid-tensor-power-expansion}
\end{align}
The unit trace condition gives the last identity. The coefficients can be negative, and the states $\rho_j$ need not equal the reference state $\rho$ in the two-position assumption.
Let $\mathcal L_{\lambda\mu}$ select the local Schur labels $(\lambda,\mu)$ and trace the two unitary registers. The resulting unnormalized state $\widehat\theta_{\lambda\mu}(X)\coloneqq\mathcal L_{\lambda\mu}(\mathcal T_k(X))$ depends linearly on $X$. Since it is invariant under the simultaneous permutation action on $\mathcal P_\lambda\otimes\mathcal P_\mu$, Schur's lemma gives
\begin{align}
\widehat\theta_{\lambda\mu}(X)
\cong\bigoplus_\nu\frac{\id_{\mathcal P_\nu}}{d_\nu}
\otimes\widehat S_{\lambda\mu}^{(\nu)}(X),
\qquad
\widehat S_{\lambda\mu}^{(\nu)}(\omega^{(k)})
=\sum_jc_j\widehat S_{\lambda\mu}^{(\nu)}(\rho_j^{\otimes k}).
\label{eq:non-iid-linear-multiplicity-weights}
\end{align}
Both selecting a branch and extracting its blocks are linear operations, which proves the second equality. For the physical input, each block is positive. Dividing by the branch probability $p_{\lambda\mu}=\Tr\widehat\theta_{\lambda\mu}(\omega^{(k)})>0$ gives the normalized matrices $S^{(\nu)}$ and state $\theta_{\lambda\mu}$ in \eqref{eq:main-non-iid-identity-weights}. The expansion is linear for $\widehat S=p_{\lambda\mu}S$, before this division.
We can now identify the same error operators as before. Using the orthonormal vectors $\ket*{f_{\nu,\alpha,e}}=(\id\otimes F_{\nu,\alpha,e})\ket*{\Phi_{d_\lambda}}$ from \eqref{eq:flat-choi-reference}, the conditional state is
\begin{align}
\theta_{\lambda\mu}
&=\sum_{\nu,\alpha,\beta,e}\frac{S^{(\nu)}_{\alpha\beta}}{d_\nu}
\ket*{f_{\nu,\alpha,e}}\bra*{f_{\nu,\beta,e}}
=(\idmap\otimes\mathcal N_T)(\Phi_{d_\lambda}),
\\
T&\coloneqq\bigoplus_\nu S^{(\nu)}\otimes\pi_{d_\nu}\ge0,
\label{eq:non-iid-source-channel-and-weights}
\end{align}
where $\mathcal N_T$ is the weighted map in \eqref{eq:weighted-error-map}. Alice's marginal is $\pi_{d_\lambda}$ by Schur's lemma, so $\mathcal N_T$ is CPTP. Its Choi state has the same $\id\otimes S$ structure as in the earlier proof. The matrices $T$ are positive, so \cref{lem:normalized-polar} applies directly, regardless of the signs in \eqref{eq:non-iid-tensor-power-expansion}.
The joint Schur-label probabilities are
\begin{align}
p_{\lambda\mu\nu}^{\omega}
\coloneqq\Tr\!\left[\omega^{(k)}\Pi_\lambda^A\Pi_\mu^B\Pi_\nu^{AB}\right].
\label{eq:non-iid-schur-law}
\end{align}
The common permutation leaves these probabilities unchanged, and $\Tr S^{(\nu)}=p^\omega_{\nu|\lambda\mu}$. Every such label has at most $d_{AB}$ rows. The proof of \cref{thm:universal-distillation} therefore applies with $r=d_{AB}$ and gives
\begin{align}
\varepsilon_k(\omega^{(k)};L)
&\le p_k^{\omega}(L)+\frac{1-p_k^{\omega}(L)}{L},
&p_k^{\omega}(L)
&\coloneqq\Pr_{p^\omega}\!\left[J_k^{(d_{AB})}<\log_2K+\log_2L\right].
\label{eq:arbitrary-source-packing}
\end{align}
The accepted region, code measurement and recovery are unchanged: their proof uses only the error operators and positivity of $S^{(\nu)}$, which we have just established. It needs no further independence assumption. The local marginals after permutation averaging also have maximally mixed permutation factors, so the reconstruction map of \cref{prop:perfect-local-recoverability} applies to them.

\subsection{Two-marginal control of the Schur spectrum}
For a state $\tau_k$ on $(\mathbb C^d)^{\otimes k}$, $k\ge2$, write $\overline\tau_k^{(1)}=k^{-1}\sum_i(\tau_k)_i$ and $\overline\tau_k^{(2)}=[k(k-1)]^{-1}\sum_{i\ne j}(\tau_k)_{ij}$, with all positions identified with fixed ordered registers. The following bound controls the sampled Young diagram using these two-position marginals. It does not require $\tau_k$ to be permutation invariant.

\begin{lemma}[Schur-spectrum concentration from two marginals]
\label{lem:two-marginal-schur}
Let $\rho$ be a state on $\mathbb C^d$, let $p$ be its decreasing eigenvalue vector, and let $\Lambda$ be the Schur label sampled from $\tau_k$. If $\|\overline\tau_k^{(2)}-\rho^{\otimes2}\|_1\le a_k$, then
\begin{align}
\mathbb E\left\|\frac\Lambda k-p\right\|_2^2
\le\frac{d-\Tr\rho^2}{k}+3a_k
\le\frac dk+3a_k.
\label{eq:two-marginal-schur-variance}
\end{align}
\end{lemma}

\begin{proof}
Choose a Hermitian Hilbert--Schmidt orthonormal basis $T_1,\ldots,T_{d^2}$ of $\mathcal L(\mathbb C^d)$, and write $\overline T_{a,k}=k^{-1}\sum_iT_a^{(i)}$. The positive operator
\begin{align}
\mathsf V_{k,\rho}
&\coloneqq\sum_{a=1}^{d^2}
\left(\overline T_{a,k}-\Tr(\rho T_a)\id\right)^2
\nonumber\\
&=\frac dk\id+\frac2{k^2}\sum_{i<j}F_{ij}
-\frac2k\sum_i\rho^{(i)}+\Tr\rho^2\id
\label{eq:collective-quadratic-operator}
\end{align}
follows from $\sum_aT_a\otimes T_a=F$ and $\sum_aT_a^2=d\id$, where $F_{ij}$ exchanges positions $i,j$. To evaluate this operator on a Schur block, use the Young-tableau basis \cite{weyl1946classical}. The Jucys--Murphy operator $J_j=\sum_{i<j}F_{ij}$ has eigenvalue $c-r$ when the box containing $j$ is in row $r$ and column $c$. Summing over all boxes removes the tableau dependence and gives
\begin{align}
\left.\sum_{i<j}F_{ij}\right|_{\mathcal P_\lambda}
=\left[\sum_{a=1}^d\sum_{c=1}^{\lambda_a}(c-a)\right]\id
=\frac12\left[\sum_{a=1}^d\lambda_a(\lambda_a-2a+1)\right]\id.
\label{eq:non-iid-swap-content-formula}
\end{align}
On the same block, the largest eigenvalue of $\sum_i\rho^{(i)}$ is $\sum_a\lambda_ap_a$. To see this, use an eigenbasis of $\rho$: the weights are majorized by $\lambda$, and a highest-weight vector attains this value. Combining the swap and linear terms therefore gives

\begin{align}
\left.\mathsf V_{k,\rho}\right|_\lambda
\ge\left[\left\|\frac\lambda k-p\right\|_2^2
+\frac1{k^2}\sum_{a<b}(\lambda_a-\lambda_b)\right]\id
\ge\left\|\frac\lambda k-p\right\|_2^2\id.
\label{eq:casimir-schur-lower-bound}
\end{align}
Here partitions are padded with zeros to length $d$. Thus
\begin{align}
\mathbb E\left\|\frac\Lambda k-p\right\|_2^2
&\le\Tr[\tau_k\mathsf V_{k,\rho}]
\nonumber\\
&=\frac dk+\frac{k-1}k\Tr[F\overline\tau_k^{(2)}]
-2\Tr[\rho\overline\tau_k^{(1)}]+\Tr\rho^2.
\end{align}
Partial trace implies $\|\overline\tau_k^{(1)}-\rho\|_1\le a_k$. Since $\|F\|_\infty=1$, $\|\rho\|_\infty\le1$, and $\Tr[F\rho^{\otimes2}]=\Tr\rho^2$, comparison with the product marginals proves \eqref{eq:two-marginal-schur-variance}.
\end{proof}

\subsection{Deriving the coherent-information corollary}
\begin{proof}[Proof of \cref{cor:non-iid-main}]
Assume $I(A\rangle B)_\rho=H(B)_\rho-H(AB)_\rho\ge R+\delta$. Apply \cref{lem:two-marginal-schur} to the $B$ and $AB$ marginals; partial trace bounds both averaged two-position defects by $a_k$. Markov's inequality and $\|x\|_1^2\le d\|x\|_2^2$ give, for fixed tolerances $u_B,u_{AB}>0$,
\begin{align}
\Pr(\mathrm{bad})
\le\frac{d_B}{u_B^2}\left(\frac{d_B}k+3a_k\right)
+\frac{d_{AB}}{u_{AB}^2}\left(\frac{d_{AB}}k+3a_k\right),
\label{eq:non-iid-spectral-tail}
\end{align}
where the bad event is that either normalized Young diagram differs from its reference spectrum in $\ell^1$ distance by more than its tolerance. The entropy-continuity bound \cite{Audenaert2007} allows the tolerances to be chosen using only the dimensions and $\delta$, so that on the good event
$H(\mu/k)-H(\nu/k)\ge I(A\rangle B)_\rho-\delta/4$.
We next use \cref{lem:rank-dimensions,lem:rank-Q} with $r_A=d_A$, $r_B=d_B$ and $r=d_{AB}$. These bounds require only the corresponding row constraints, which hold here. With
$\chi_0\coloneqq d_Ad_Bd_{AB}-d_A(d_A+1)/2$, these estimates give
\begin{align}
J_k^{(d_{AB})}
\ge k[H(\mu/k)-H(\nu/k)]
-b_{d_B}\log_2(k+d_B)-\chi_0\log_2(k+1).
\label{eq:non-iid-packing-entropy}
\end{align}
For sufficiently large $k$, depending only on the dimensions and $\delta$, the logarithmic overhead is at most $k\delta/4$. Hence on the good event,
$J_k^{(d_{AB})}\ge k(R+\delta/2)\ge\log_2K_k+\log_2k$.
Taking $L=k$ in the existing cutoff bound \eqref{eq:arbitrary-source-packing} therefore gives
\begin{align}
p_k^\omega(k)\le\Pr(\mathrm{bad}),
\qquad
\varepsilon_k\le\Pr(\mathrm{bad})+k^{-1}
=O_{d_A,d_B,\delta}(a_k+k^{-1}).
\label{eq:non-iid-cutoff-consequence}
\end{align}
The reference state, gap and tolerances are used only to estimate the error; the protocol does not depend on them. The constant and the required lower bound on $k$ depend only on the dimensions and $\delta$. The same estimate therefore holds uniformly over reference states with this gap, as needed for the channel result below.
\end{proof}

\section{Quantum communication through almost i.i.d.\ channels}
\label{app:non-iid-communication}
Fix the known design state $\xi$ throughout, write $D=d_{A'}d_B$ and $\Omega=\Omega_{\xi,\mathcal M_k}$, and use the probe states defined in \eqref{eq:main-channel-two-marginal}. The code-ensemble lemma already supplies the encoder and recovery. We check that permutation averaging gives conditional Schur channels of the required positive-weight form and then apply the preceding concentration bound.

We use the shared-permutation channel $\widetilde{\mathcal M}_k$ defined in \eqref{eq:main-non-iid-channel-twirl}. Changing the permutation variable gives
$\widetilde{\mathcal M}_k\circ\operatorname{Ad}_{R_g^{A'}}=\operatorname{Ad}_{R_g^B}\circ\widetilde{\mathcal M}_k$.
Also,
$(\id_{R^{\otimes k}}\otimes R_\pi^{A'})\ket*{\phi_\xi}^{\otimes k}=((R_\pi^R)^\dagger\otimes\id_{(A')^{\otimes k}})\ket*{\phi_\xi}^{\otimes k}$,
so
\begin{align}
(\idmap_{R^{\otimes k}}\otimes\widetilde{\mathcal M}_k)(\phi_\xi^{\otimes k})
=\mathcal T_k(\Omega).
\label{eq:block-channel-twirled-state}
\end{align}
Let $p^\Omega$ be the Schur law \eqref{eq:non-iid-schur-law} with $A=R$ and $\omega^{(k)}=\Omega$. It is unchanged by the twirl and has $\ell(\nu)\le D$.
Keep the fixed input decomposition \eqref{eq:communication-mixed-schur} and the sector-input map $\mathcal E_\lambda$ from the proof of \cref{lem:universal-code-ensemble}. Let $\mathcal B_\mu$ extract Bob's $\mu$ block and trace its unitary register. Covariance implies that the CP branch
$\mathcal M_{\lambda\mu}\coloneqq\mathcal B_\mu\circ\widetilde{\mathcal M}_k\circ\mathcal E_\lambda$
intertwines the conjugation actions of $p_\lambda$ and $p_\mu$. Its effect therefore commutes with the irreducible $p_\lambda$, giving the same flatness statement as \eqref{eq:schur-output-label-flatness}:
\begin{align}
\mathcal M_{\lambda\mu}^{\dagger}(\id)
=q_{\mu|\lambda}\id_{\mathcal P_\lambda},
\qquad
\sum_\mu q_{\mu|\lambda}=1.
\label{eq:non-iid-channel-flat-branch}
\end{align}
For $q_{\mu|\lambda}>0$, the normalized channel $\mathcal N_{\lambda\mu}=\mathcal M_{\lambda\mu}/q_{\mu|\lambda}$ is CPTP. The Schur purification of $\phi_\xi^{\otimes k}$, with the reference unitary register traced out, leaves $\Phi_{d_\lambda}\otimes\xi_{\mathcal Q_\lambda}$ in the input sector. Hence the conditional state extracted from \eqref{eq:block-channel-twirled-state} is
\begin{align}
\theta_{\lambda\mu}
=(\idmap\otimes\mathcal N_{\lambda\mu})(\Phi_{d_\lambda}),
\qquad
\Pr(\lambda,\mu)=p_\lambda q_{\mu|\lambda}.
\label{eq:non-iid-channel-choi-identification}
\end{align}
By \eqref{eq:non-iid-source-channel-and-weights}, this is precisely the weighted Schur error state already analyzed in \cref{lem:normalized-polar}. These two checks replace the i.i.d. step in the code-ensemble proof; the accepted-sector polar estimate and encoder construction in \cref{lem:universal-code-ensemble} otherwise apply unchanged, with $r=D$.
In particular, the mixed encoder \eqref{eq:source-schur-encoder} and its isometric refinement \eqref{eq:universal-isometric-encoder} give, for every $L\ge1$ and $1\le K\le d_{A'}^k$,
\begin{align}
\mathbb E_s\varepsilon_k(\mathcal M_k;s)
&\le p_k^\Omega(L)+\frac{1-p_k^\Omega(L)}{L},
&p_k^\Omega(L)
&\coloneqq\Pr_{p^\Omega}\!\left[J_k^{(D)}<\log_2K+\log_2L\right].
\label{eq:arbitrary-channel-packing}
\end{align}
The refinement preserves averaged fidelity by linearity; the flat probabilities are needed only for the mixed encoder, as in the earlier proof. Include the permutation $\pi$ in the shared seed, apply $R_\pi^{A'}$ to its isometric encoder and $(R_\pi^B)^\dagger$ before its decoder. Averaging over $\pi$ implements the twirled channel exactly. This extends \eqref{eq:sw-transmission-bound} to the block channel.

\subsection{Deriving the uniform rate corollary}
\begin{proof}[Proof of \cref{cor:non-iid-communication}]
Fix $(\mathcal M_k,\mathcal N)\in\mathfrak F_k$ and apply the concentration-and-packing argument of \cref{cor:non-iid-main} to $\Omega$, with $A=R$ and reference state $\rho=\omega_{\xi,\mathcal N}$. Its averaged two-position defect is at most $a_k$, and $I(R\rangle B)_\rho=I_c(\xi,\mathcal N)\ge R_0+\delta$. The uniformity established after \eqref{eq:non-iid-cutoff-consequence} therefore gives
\begin{align}
p_k^\Omega(k)=O_{d_{A'},d_B,\delta}(a_k+k^{-1})
\end{align}
for $K_k=2^{\lfloor kR_0\rfloor}$, uniformly over $\mathfrak F_k$. This is a statement about the Schur-label tail, independent of the operational task. For transmission, use the ensemble bound \eqref{eq:arbitrary-channel-packing} with $L=k$ to obtain \eqref{eq:non-iid-channel-error-main}. The encoder fits into the physical input because $R_0<I_c(\xi,\mathcal N)\le H(\xi)\le\log_2d_{A'}$.
Markov's inequality gives high-probability success for each channel fixed before the seed, with a uniform error bound.
\end{proof}

\bibliography{biblio}

\end{document}